\documentclass{article}[12pt]

\usepackage{amssymb}
\usepackage{amsmath}
\usepackage{enumerate}
\usepackage{latexsym}
\usepackage{mathrsfs}
\usepackage{amsthm}
\usepackage{verbatim}
\usepackage{graphicx}
\usepackage{epstopdf}
\usepackage{epsfig}
\usepackage{color}
\usepackage{subfig}
\usepackage{float}
\usepackage{color}
\usepackage{titlesec}
\usepackage{bm}
\usepackage[colorlinks,linkcolor=blue,anchorcolor=green,citecolor=red]{hyperref}
\usepackage[T1]{fontenc}

\usepackage{charter}
\usepackage{pythonhighlight}
\usepackage{fancyhdr}
\usepackage{amscd}
\usepackage{stmaryrd}
\usepackage{trimclip}

\usepackage{caption}
\usepackage{algorithm} %format of the algorithm
\usepackage{algorithmic} %format of the algorithm
\usepackage{apptools}
\usepackage{tikz}

\newcommand{\newbrace}[1][]{
	\begin{tikzpicture}[baseline=-0.5ex]
		\draw[#1] (0,0) -- (0.3,0.3);
		\draw[#1] (0,0) -- (0.3,-0.3);
	\end{tikzpicture}
}

\newenvironment{casesnew}[1][->]%
{\;\newbrace[#1]\;\begin{array}{@{}l@{}}}%
	{\end{array}}

\usepackage[section]{placeins}

\makeatletter
\@addtoreset{equation}{section}
\makeatother

\makeatletter % `@' now normal "letter"
\@addtoreset{equation}{section}
\makeatother  % `@' is restored as "non-letter"
\renewcommand\theequation{\oldstylenums{\thesection}%
	.\oldstylenums{\arabic{equation}}}

\usepackage{geometry}
\usepackage{listings}
\usepackage{setspace}
\definecolor{Code}{rgb}{0,0,0}
\definecolor{Decorators}{rgb}{0.5,0.5,0.5}
\definecolor{Numbers}{rgb}{0.5,0,0}
\definecolor{MatchingBrackets}{rgb}{0.25,0.5,0.5}
\definecolor{Keywords}{rgb}{0,0,1}
\definecolor{self}{rgb}{0,0,0}
\definecolor{Strings}{rgb}{0,0.63,0}
\definecolor{Comments}{rgb}{0,0.63,1}
\definecolor{Backquotes}{rgb}{0,0,0}
\definecolor{Classname}{rgb}{0,0,0}
\definecolor{FunctionName}{rgb}{0,0,0}
\definecolor{Operators}{rgb}{0,0,0}
\definecolor{Background}{rgb}{0.98,0.98,0.98}
\lstdefinelanguage{Python}{
	numbers=left,
	numberstyle=\footnotesize,
	numbersep=1em,
	xleftmargin=1em,
	framextopmargin=2em,
	framexbottommargin=2em,
	showspaces=false,
	showtabs=false,
	showstringspaces=false,
	frame=l,
	tabsize=4,
	basicstyle=\ttfamily\small\setstretch{1},
	backgroundcolor=\color{Background},
	commentstyle=\color{Comments}\slshape,
	stringstyle=\color{Strings},
	morecomment=[s][\color{Strings}]{"""}{"""},
	morecomment=[s][\color{Strings}]{'''}{'''},
	morekeywords={import,from,class,def,for,while,if,is,in,elif,else,not,and,or,print,break,continue,return,True,False,None,access,as,,del,except,exec,finally,global,import,lambda,pass,print,raise,try,assert},
	keywordstyle={\color{Keywords}\bfseries},
	morekeywords={[2]@invariant,pylab,numpy,np,scipy},
	keywordstyle={[2]\color{Decorators}\slshape},
	emph={self},
	emphstyle={\color{self}\slshape},
}
\makeatletter
\DeclareRobustCommand{\shortto}{%
	\mathrel{\mathpalette\short@to\relax}%
}

\newcommand{\short@to}[2]{%
	\mkern2mu
	\clipbox{{.4\width} 0 0 0}{$\m@th#1\vphantom{+}{\shortrightarrow}$}%
}
\makeatother

\AtAppendix{\counterwithin{theorem}{subsection}}
\AtAppendix{\counterwithin{definition}{subsection}}
\AtAppendix{\counterwithin{remark}{subsection}}

\newtheorem{theorem}{Theorem}[section]
\newtheorem{lemma}[theorem]{Lemma}
\newtheorem{proposition}[theorem]{Proposition}
\newtheorem{remark}{Remark}[section]
\theoremstyle{definition}

\DeclareMathOperator*{\argmin}{arg\,min}
\DeclareMathOperator{\Ima}{Im}

\begin{document}
	
	\title{Some results on the Guest-Hutchinson modes and periodic mechanisms of the Kagome lattice metamaterial }
	\author{Xuenan Li\thanks{New York University, xl2643@nyu.edu}\qquad Robert V. Kohn\thanks{New York University, kohn@cims.nyu.edu}}
	
	\maketitle
	
	\begin{abstract}
		Lattice materials are interesting mechanical metamaterials, and their mechanical properties are often related to the presence of mechanisms. The existence of periodic mechanisms can be indicated by the presence of Guest-Hutchinson (GH) modes, since GH modes are sometimes infinitesimal versions of periodic mechanisms. However, not every GH mode comes from a periodic mechanism. This paper focuses on: (1) clarifying the relationship between GH modes and periodic mechanisms; and (2) answering the question: which GH modes come from periodic mechanisms? We focus primarily on a special lattice system, the Kagome lattice. Our results include explicit formulas for all two-periodic mechanisms of the Kagome lattice, and a necessary condition for a GH mode to come from a periodic mechanism in general. We apply our necessary condition to the two-periodic GH modes, and also to some special GH modes found by Fleck and Hutchinson using Bloch-type boundary conditions on the unit cell of the Kagome lattice.
	\end{abstract}
	
	\section{Introduction}
	Mechanical metamaterials are artificial materials. These carefully designed materials exhibit properties that cannot be realized by conventional materials \cite{bertoldi2017flexible}. Among the family of mechanical metamaterials, 2-dimensional lattice materials are already interesting. For any given 2D lattice, i.e. a spatially periodic structure in the plane, we can consider it as an elastic material by viewing the edges as Hookean springs and the nodes as perfect hinges. We emphasize that we are studying a nonlinear problem in considering the mechanics of these lattice materials, i.e. we are interested in large local and global deformations. In particular, for a deformation $u(x)$, the elastic energy of the spring connecting $x_i$ and $x_j$ is the elastic constant times $\big(\frac{|u(x_i) - u(x_j)|}{|x_i - x_j|} - 1\big)^2$.
	
	In many cases, the elastic behavior of such lattice materials is related to the presence of \textit{mechanisms}, i.e. deformations other than rigid motions that have zero elastic energy. Lattice materials with mechanisms are degenerate elastic materials, i.e. they cannot sustain every boundary load, due to the existence of deformations with zero elastic energy. Among mechanisms, there is a special type of mechanism that deforms the reference lattice to a different periodic structure, i.e. a new lattice. We call such mechanisms \textit{periodic mechanisms}. Some periodic mechanisms have the property that their deformed states have the same periodicity as that of the reference lattice. We call these \textit{one-periodic mechanisms}. In other cases, the periodicity of the deformed state is $N^2$ times larger than the reference lattice ($N$ times larger in each lattice direction). We call these \textit{$N$-periodic mechanisms}. Many lattice materials have periodic mechanisms: examples include the family of Kagome lattices (see e.g. \cite{kapko2009collapse, sun2012surface}), the rotating square metamaterial\footnote{To view the rotating square metamaterial as a lattice material, we can start with the square lattice and add extra diagonal springs to make some squares rigid. For an illustrative example, see Figure 6 in the supplementary information of \cite{czajkowski2022conformal}.} (see e.g. \cite{czajkowski2022conformal, deng2020characterization}), and a variant of the rotating square metamaterial (see e.g. \cite{zheng2022continuum, zheng2022modeling}).
	
	The physics literature has already developed some understanding of periodic mechanisms via matrix methods (see e.g. \cite{hutchinson2006structural, mao2018maxwell}). If we view mechanisms as buckling patterns, then it is natural to look for a linear elastic way to predict the existence of a periodic mechanism, by analogy to the use of a linear elastic calculation to predict the buckling of a beam under compressive loading (see e.g. \cite{stakgold1971branching} around Equation 1.6). The onset of a periodic mechanism is sometimes indicated by the presence of what are now called Guest-Hutchinson modes (henceforth: GH modes). Guest and Hutchinson \cite{guest2003determinacy} studied the linear mechanics of lattice materials and concluded that Maxwell lattices {\color{black}(see section \ref{sec:preliminary-GH} for a discussion of Maxwell lattices)} cannot be simultaneously statically and kinematically stable\footnote{\color{black}The definition of kinematically stable is on page 7 after we introduce the GH mode; the definition of statically stable is on page 11 before we prove Theorem \ref{thm:GH}.}. As a consequence of their result, a statically stable Maxwell lattice must have at least one GH mode. This is known as the Guest-Hutchinson theorem, and we review it from a homogenization perspective in section \ref{sec:preliminary}.  
	
	However, the relationship between periodic mechanisms and GH modes is not simple. One mystery is that some lattice materials are known to have periodic mechanisms but do not have GH modes, for example, {\color{black}{the twisted Kagome lattice (T-T) lattice}} \cite{hutchinson2006structural}.{\color{black}{ Another mystery is that some GH modes do not come from periodic mechanisms \cite{borcea2010periodic}. In fact, the GH modes and periodic mechanisms of a lattice system are like the infinitesimal flexes and nonlinear flexes of a finite bar framework. From the rigidity theory of Connelly and Whiteley \cite{connelly1996second}, we know that for a finite bar framework, not every infinitesimal flex comes from a fully nonlinear flex; there is a necessary condition for an infinitesimal flex to come from a nonlinear flex. The necessary condition is known as the second-order stress test. Similarly, for our lattice systems, we shall derive a necessary condition for a GH mode to come from a periodic mechanism.}}
	
	{\color{black}{These mysteries motivate the key questions investigated in this paper: (1) what is the relationship between the GH modes and the periodic mechanisms of a lattice material? (2) which GH modes come from periodic mechanisms?}} We study these questions in detail for a special and very interesting family of lattice materials: the Kagome lattice (see Figure \ref{fig:standardkagomespring}) and its images under periodic mechanisms. Each lattice within the family  is made up of equilateral triangles and hexagons. If the hexagons are regular hexagons, then we call it the standard Kagome lattice. We call the images of the standard Kagome lattice under periodic mechanisms the deformed Kagome lattices. The standard and deformed Kagome lattice are examples of Maxwell lattices \cite{mao2018maxwell}. Such lattices are great places to find mechanisms. There is indeed a well-known one-periodic mechanism of the standard Kagome lattice (see section \ref{one-periodic-mechanism} for a discussion of this one-periodic mechanism). The deformed state under this one-periodic mechanism has a special name, {\color{black}{the twisted Kagome lattice (T-T) lattice}}. The standard Kagome lattice also has periodic mechanisms with other periodicities, as was discussed at some length in \cite{kapko2009collapse}. Our results include the following:
	\begin{enumerate}[(a)]
		\item We find explicit formulas for all two-periodic mechanisms of the standard Kagome lattice (see section \ref{sec:two-periodic-mechanism}).
		\item We identify a necessary condition for a GH mode to come from a mechanism (see section \ref{sec:necessary-condition}).
		\item We study the GH modes found by Fleck and Hutchinson using Bloch-type boundary conditions in \cite{hutchinson2006structural}, showing that in most cases these special GH modes do not correspond to periodic mechanisms (see section \ref{subsection:fleck-hutchinson}).
		\item We find all $N$-by-one periodic mechanisms for any $N \geq 2$, and also examples of some non-periodic mechanisms (see section \ref{subsec:non-periodic}).
		\item {\color{black} We find a special case where every GH mode must come from a mechanism} (see section \ref{sec:open-questions}).
	\end{enumerate}
	
	To make the paper self-contained, we start in sections 2 and 3 with a systematic review about the mechanics of periodic lattices, including the definition of GH modes, etc. In particular, we show in section \ref{subsec:GH-mechanism} that the infinitesimal version of a periodic mechanism is a GH mode only when the infinitesimal macroscopic deformation vanishes. This explains why the standard Kagome lattice has a GH mode, while the twisted Kagome lattice does not. 
	
	This paper focuses primarily on mechanisms of lattice materials. However, a different question about these lattice materials is whether we can write a meaningful macroscopic energy for them. Take the Kagome lattice as an example. If we fill in a region with a scaled version of the Kagome lattice, i.e. setting the side length of each triangle to be $\epsilon$, is there a sense that we can view the region as a nonlinear elastic material as $\epsilon$ approaches zero? For a periodic mixture of nonlinearly elastic materials that are non-degenerate, the answer is yes using homogenization theory (see e.g. \cite{braides1985homogenization, muller1987homogenization}). But for a lattice material, the answer is not obvious, especially when it has mechanisms. Actually, there is a meaningful macroscopic energy for the Kagome lattice, and it only vanishes at compressive conformal maps. This is the focus of our forthcoming paper \cite{liforthcoming}, where we provide a rigorous framework for the discussion about the macroscopic behavior of the Kagome metamaterial in the physics literature \cite{czajkowski2022conformal}.
	
	We close this introduction with a brief discussion of some related work on lattice materials.
	\begin{itemize}
		\item The paper \cite{fruchart2020dualities} by M. Fruchart et.al. discussed a duality they found in the band structure of the elastic waves in the twisted Kagome lattice. They observed that two twisted Kagome lattices with different twisting angles $\theta$ and $\theta^*$ have the same vibrational spectrum and band structure if the two angles are related by $\theta + \theta^* = 2\theta_c$. At the critical point $\theta_c$, there is a two-fold degenerate Dirac point in the Brillouin zone. Our focus is different from that of \cite{fruchart2020dualities}, since we focus on finding mechanisms of the Kagome lattice (and the relationship between mechanisms and GH modes), not the vibrational properties of the deformed lattices under these mechanisms.
		
		\item A different focus of work on the Kagome lattice family and its variant, the \textit{deformed Kagome lattice}\footnote{The term "deformed Kagome lattice" has sometimes been used for a lattice whose unit cell consists of triangles and hexagons but in which not all edge lengths are the same, see e.g. \cite{rocklin2017transformable}. Our usage will be different: in our paper, a "deformed Kagome lattice" is always an image of the standard Kagome lattice under a periodic mechanism.}, involves their topological properties.  For a finite piece of these lattices, the first-order flexible modes can reside on one side of the lattice; this phenomenon is called topological polarization (see e.g. \cite{kane2014topological, mao2018maxwell}). The paper \cite{sun2012surface} studied the concentration of first-order flexible modes of the twisted Kagome lattice with different boundary conditions; the paper \cite{rocklin2017transformable} harnessed mechanisms of the deformed Kagome lattice and discovered novel domain structures that control the stiffness along edges and domain walls.
		
		%        \item The paper \cite{lakes2022extremal} by R.S. Lakes studied hinged lattices with {\color{blue}"nonlinear Poisson's ratio -1", i.e. lattices in which isotropic dilations or compressions are the only macroscopic deformations}, and found that these materials do not in general obey the theory of linear elasticity. The Kagome lattice supports Lake's conclusion: it is known to have Poisson's ratio -1, since compressing it in one lattice direction causes the same amount of compression in the other lattice direction. However, linear elasticity has problems in this setting; in fact, we show in section \ref{subsec:discontinuity} that the linear effective tensor $A_{\text{eff}}$ is a discontinuous function of the compression ratio. {\color{blue}There are other hinged lattices that do not obey the theory of  linear elasticity (see e.g. \cite{seppecher2011linear}).} But our main focus is different from Lakes'. We focus primarily on the various mechanisms of the Kagome lattice that achieve the negative Poisson's ratio -1.
		\item {\color{black}The setting of linear elasticity has its own problems when studying the effective behavior of hinged lattice systems. The paper \cite{lakes2022extremal} by R.S. Lakes found that hinged lattices with "nonlinear Poisson's ratio -1", i.e. lattices in which isotropic dilations or compressions are the only macroscopic deformations, do not in general obey the theory of linear elasticity. The family of Kagome-type lattices is also rich in examples whose effective behavior cannot be reproduced by theory of linear elasticity. The paper by Nassar et.al. \cite{nassar2020microtwist} proposed a new effective theory that is capable of rendering polarization effects on a macroscopic scale in Kagome-type lattices. There are other hinged lattices that do not obey the theory of  linear elasticity (see e.g. \cite{seppecher2011linear}). However, our main focus is not on problems caused by the setting of linear elasticity. Instead, we focus primarily on studying the relationship between GH modes and periodic mechanisms.}
		
		\item {\color{black}There are studies about mechanical metamaterials that achieve different types of macroscopic deformations other than isotropic dilations or compressions. The papers \cite{milton2013adaptable,milton2013complete} by Milton discussed the possible macroscopic deformations of periodic nonlinear affine unimode (bimode) metamaterials constructed by rigid bars and rotation-free pivots, i.e. materials for which macroscopic deformations can only follow along a one-dimensional (two-dimensional) trajectory in the space of deformations. The paper \cite{zheng2022continuum} studied a family of planar kirigami with unit cells of four convex quadrilateral panels and four parallelogram slits and discovered that this family of planar kirigami can have elliptic or hyperbolic types of responses due to loads. But our focus is different from these papers. We focus on the various mechanisms of the Kagome lattice.}
	\end{itemize}
	
	The structure of the paper is as follows. Section \ref{sec:preliminary} reviews the linear elastic mechanics of periodic lattices; it includes a self-contained definition and discussion of GH modes. In section \ref{one-periodic-mechanism}, we clarify the relationship between infinitesimal versions of periodic mechanisms and GH modes. We also use the one-periodic mechanism of the Kagome lattice as a vivid example. We present our construction of a three-parameter two-periodic mechanism on the standard Kagome lattice in section \ref{sec:two-periodic-mechanism}. The existence of such two-periodic mechanisms indicate that there are many ways to buckle (compress) the Kagome lattice. Some "buckling patterns" even preserve the macroscopic deformation. In section \ref{sec:necessary-condition}, we answer the question: which GH modes come from periodic mechanisms? We identify a necessary condition and apply it to the two-periodic GH modes and the Fleck-Hutchinson modes of the standard Kagome lattice. In section \ref{subsec:non-periodic}, we present our constructions of $N$-by-one periodic mechanisms and some non-periodic mechanisms. We close this paper by presenting a special case where every GH mode comes from a periodic mechanism in section \ref{sec:open-questions}.
	
	\section{Preliminaries and notation}\label{sec:preliminary}
	In this section, we review the existing matrix methods for lattice structures and we discuss the linear elastic mechanics of lattice systems from a homogenization perspective. The section includes a self-contained discussion about the existence of GH modes in statically stable Maxwell lattices (we call this the \textit{Guest-Hutchinson Theorem}). The static stability of a lattice is non-degeneracy in our language, so this theorem specialized to 2D lattices says
	\begin{theorem}[Guest-Hutchinson Theorem]\label{thm:GH}
		If  a 2-dimensional Maxwell lattice has a non-degenerate effective tensor $A_{\text{eff}}$, then it must have a GH mode.
	\end{theorem}
	This theorem is proved in section \ref{subsec:effective}.
	
	\subsection{GH modes and some facts about the standard Kagome lattice}\label{sec:preliminary-GH}
	First we review the definition of GH modes and the related matrix methods \footnote{{\color{black}The literature on this topic is vast, and we do not attempt a historical review. Our treatment is consistent, for example, with \cite{guest2003determinacy, hutchinson2006structural, mao2018maxwell}.}}. A 2-dimensional lattice system consists of vertices and edges that form a periodic tiling. The unit cell of each lattice includes vertices and edges that cover the whole plane when translated by primitive vectors. For example, the standard Kagome lattice in Figure \ref{fig:standardkagomespring} has three vertices (A,B,C) and six edges in the unit cell and they cover the whole plane when translated by the two primitive vectors $\bm{v}_1, \bm{v}_2$. Such a lattice system can be considered as an elastic material if we view each edge as a Hookean spring. The nonlinear elastic energy of a deformation $u(x) \in \mathbb{R}^2$ on the spring connecting vertices $x_i, x_j$ is 
	\begin{align}\label{nonlinear_elastic_energy}
		{\color{black} \frac{1}{2}} k_{ij} \bigg(|u(x_i)-u(x_j)| - |x_i - x_j|\bigg)^2 = {\color{black} \frac{1}{2}} k_{ij} l_{ij}^2 \left(\frac{|u(x_i)-u(x_j)|}{|x_i - x_j|} - 1\right)^2,
	\end{align}
	where $k_{ij}$ is the spring constant for the spring between $x_i$ and $x_j$ and $l_{ij} = |x_i - x_j|$ is the original spring length. If the deformation is very close to the identity, i.e. $u(x) = x + v(x)$ and the displacement $v(x)$ is small, then the nonlinear elastic energy in \eqref{nonlinear_elastic_energy} can be approximated by
	\begin{align*}
		{\color{black} \frac{1}{2}} k_{ij} l_{ij}^2 \left(\frac{|x_i - x_j + v(x_i) - v(x_j)|}{|x_i - x_j|} - 1\right)^2 &= {\color{black} \frac{1}{2}}  k_{ij} l_{ij}^2 \bigg( \Big(1 + 2\left \langle \frac{x_i - x_j}{|x_i - x_j|}, \frac{v(x_i) - v(x_j)}{|x_i-x_j|}\right \rangle + \frac{|v(x_i) - v(x_j)|^2}{|x_i-x_j|^2}\Big)^{1/2} - 1\bigg)^2\\
		&= {\color{black} \frac{1}{2}}  k_{ij} l_{ij}^2 \bigg(\left \langle \frac{x_i - x_j}{|x_i - x_j|}, \frac{v(x_i) - v(x_j)}{|x_i-x_j|}\right \rangle + h.o.t\bigg)^2\\
		&\approx {\color{black} \frac{1}{2}} k_{ij}\left \langle \frac{x_i - x_j}{|x_i -x_j|}, v(x_i) - v(x_j)\right \rangle^2,
	\end{align*}
	where $h.o.t$ means higher order terms w.r.t. $v(x)$. This leading order term is the linear elastic energy of the small displacement $v(x)$. The squared term is the change of length in the spring direction caused by the displacement $v(x)$. We call it the \textit{first-order spring extension} $e_{ij}$
	\begin{align}\label{first_order_bond_extension}
		\qquad e_{ij} = \Big\langle \hat{b}_{ij}, v(x_i) - v(x_j)\Big\rangle, \qquad \hat{b}_{ij} = \frac{x_i - x_j}{|x_i - x_j|}.
	\end{align}
	Notice that the first-order spring extension $e_{ij}$ is linear in the displacement $v(x)$.
	
	For a spring system with finite size, i.e. one where the number of vertices and edges is finite, a small displacement $v(x)$ that makes all the first-order spring extensions \eqref{first_order_bond_extension} vanish can be achieved by solving a linear system (see e.g. \cite{pellegrino1986matrix}). For a lattice, we can get a similar linear system by assuming the small displacement $v(x)$ is periodic. The periodicity of $v(x)$ might not be the same as the reference lattice. If $v(x)$ shares the same periodicity as that of the reference lattice, we call $v(x)$ one-periodic; if the periodicity of $v(x)$ is $N$ times larger in each lattice direction, we call $v(x)$ $N$-periodic. With periodicity, the displacement $v(x)$ can be reduced to a vector $\bm{v} \in \mathbb{R}^{2n}$ consisting of deformations $v(x_i) \in \mathbb{R}^2$ for every vertex $x_i$ in the unit cell ($n$ is the number of vertices in the unit cell). The first-order spring extension $e_{ij}$ defined in \eqref{first_order_bond_extension} is also periodic because the displacement $v(x)$ and $\hat{b}_{ij}$ are periodic. We can assemble all the first-order spring extensions $e_{ij}$ for each spring in the unit cell as a vector $\bm{e} \in \mathbb{R}^d$, where $d$ is the number of edges in the unit cell. The linear relationship between $v(x)$ and $e_{ij}$ in \eqref{first_order_bond_extension} can be written in matrix-vector form as
	\begin{align}\label{compatibility_matrix}
		\bm{e} = C\bm{v},
	\end{align}
	where $C \in \mathbb{R}^{(2n) \times d}$ is the so-called \textit{compatibility matrix}; it is determined by the geometry of the lattice system and the periodicity of $v(x)$.  For example, the standard Kagome lattice in Figure \ref{fig:standardkagomespring}(b) has three vertices $A,B,C$ and six edges in the unit cell. For any one-periodic displacement $v(x)$, its vector form $\bm{v} = \left(v(A)^T \: v(B)^T \: v(C)^T\right)^T$ is in $\mathbb{R}^6$ and the vector form $\bm{e}$ of $e_{ij}$ is also in $\mathbb{R}^6$. Thus, the compatibility matrix $C \in \mathbb{R}^{6 \times 6}$ is a square matrix.
	\begin{figure}[!htb]
		\begin{minipage}[b]{.48\linewidth}
			\centering
			\subfloat[]{\label{}\includegraphics[width=0.8\linewidth]{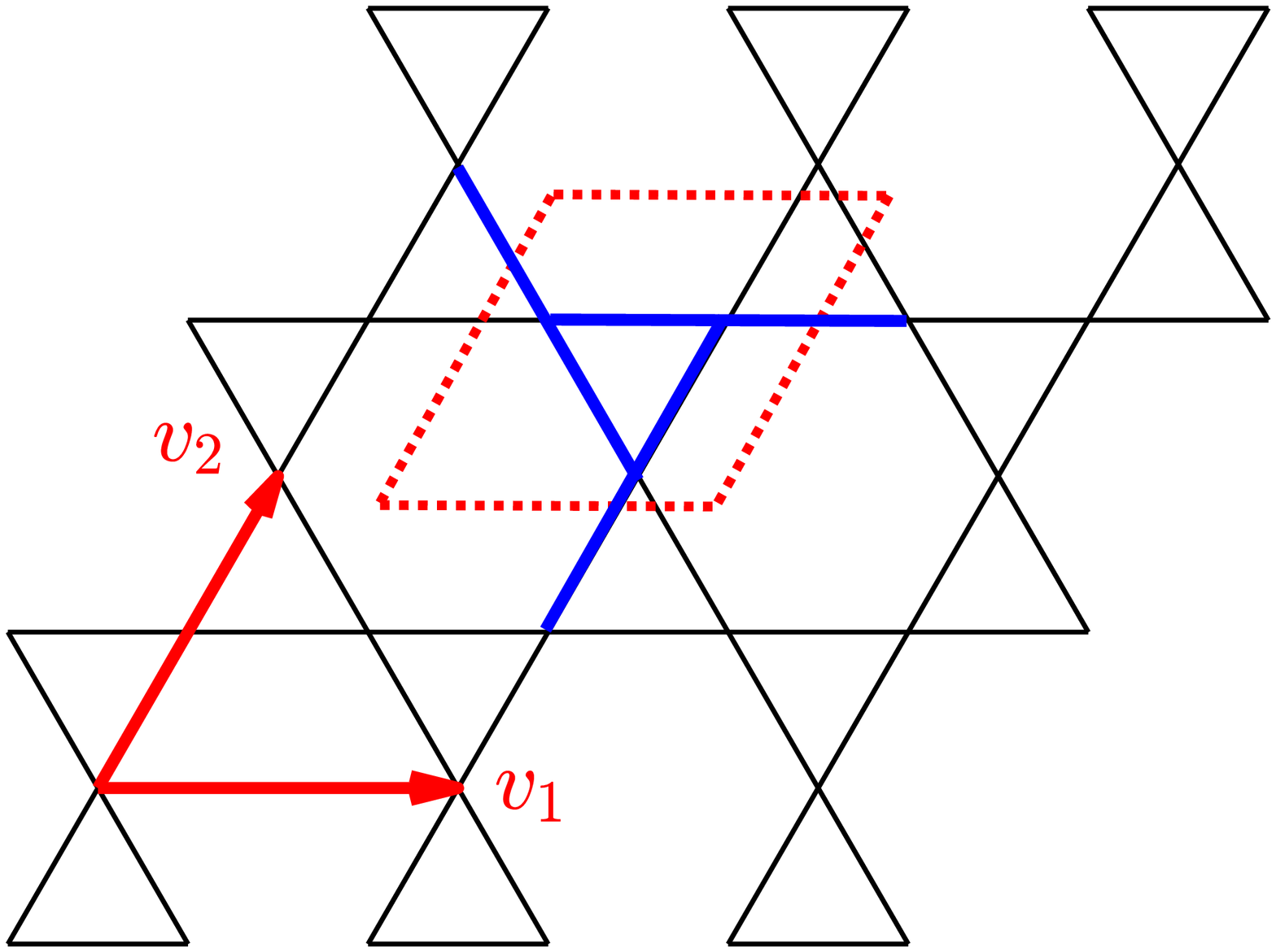}}
		\end{minipage}
		\hfill
		\begin{minipage}[b]{.48\linewidth}
			\centering
			\subfloat[]{\label{}\includegraphics[width=0.6\linewidth]{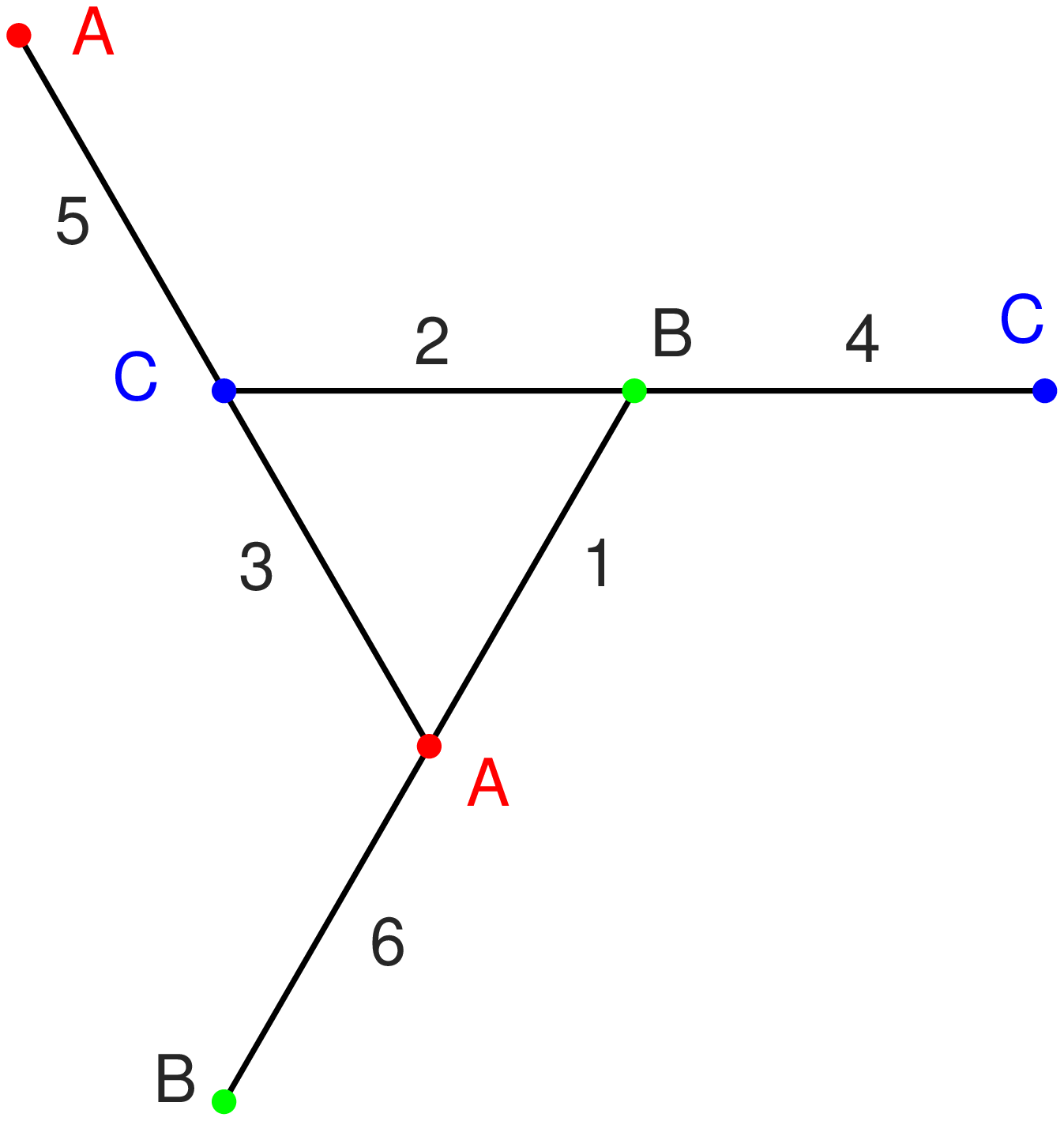}}
		\end{minipage}
		\caption{(a) The standard Kagome lattice: it has two primitive vectors $\bm{v}_1, \bm{v}_2$ and the dotted rhombus is the smallest unit cell of the standard Kagome lattice; (b) a zoomed-in version of the unit cell: the unit cell of the standard Kagome lattice contains six edges and three vertices, marked as $A,B,C$.}
		\label{fig:standardkagomespring}
	\end{figure}
	
	The transpose of the compatibility matrix is used to transform tensions in the springs to the net forces on vertices (the word "tension" does not restrict springs to be stretched only; tensions can also be negative to depict the contractions of springs). To see the relationship between tensions and net forces, we observe that if the tension in the spring between $x_i$ and $x_j$ is $t_{ij}$, then the force given by this spring at the end node $x_i$ is $t_{ij} \hat{b}_{ij}$, where $\hat{b}_{ij}$ indicates the spring direction from $x_j$ to $x_i$. Then the net force $f_i \in \mathbb{R}^2$ on vertex $x_i$ is the sum of forces over all springs connected to vertex $x_i$
	\begin{align}\label{def:netforce}
		f_i =\sum_{ j\sim i} t_{ij} \hat{b}_{ij},
	\end{align}
	where $j \sim i$ means there is a spring between $x_i, x_j$. In the periodic setting, tensions and net forces are periodic; we can assemble tensions in the unit cell as a vector $\bm{t} \in \mathbb{R}^{d}$, and net forces as a vector $\bm{f} \in \mathbb{R}^{2n}$. The linear relationship in \eqref{def:netforce} in the vector form is in fact
	\begin{align}\label{compatibility_transpose}
		\bm{f} = C^T \bm{t},
	\end{align}
	where $C$ is the compatibility matrix. We call the null vectors of $C^T$ \textit{self-stresses}. Each self-stress represents a way to have tensions in springs such that all vertices have zero net forces; hence the lattice material remains in equilibrium. 
	
	We are interested in the null space of the compatibility matrix $C$, since the null vectors of $C$ correspond to periodic displacements that preserve the lengths of the springs to first order. The compatibility matrix has two trivial null vectors, namely the 2-dimensional translations (rotations are ruled out by the periodicity of $v(x)$). A null vector which is not a translation vector is known as a \textit{Guest-Hutchinson (GH) mode} \cite{guest2003determinacy}. We define the space of GH modes as the null space of $C$ modulo the translations, i.e. two null vectors whose difference is a 2-dimensional translation are the same GH mode. If a lattice material does not have GH modes, i.e. the compatibility matrix $C$ has a 2-dimensional null space, then we call it \textit{kinematically stable}.
	
	To look for examples of lattices with GH modes, it is useful to consider the class of Maxwell lattices \cite{guest2003determinacy, hutchinson2006structural, mao2018maxwell}. Maxwell lattices sit on the boundary between flexible and rigid lattice systems. The definition of a Maxwell lattice is a lattice system whose compatibility matrix $C$ is a \textit{square} matrix. An equivalent definition of a 2-dimensional Maxwell lattice (see e.g. \cite{mao2018maxwell}) is that the average number of edges connecting each vertex is 4. It is easy to see using this definition that the standard Kagome lattice in Figure \ref{fig:standardkagomespring} is a Maxwell lattice.
	
	The standard Kagome lattice has GH modes for any periodicity. In fact, we shall show that the compatibility matrix in the $N$-periodic case $C_N$ has a $3N$-dimensional null space. By ignoring translations, we achieve that the space of $N$-periodic GH  modes is $(3N-2)$-dimensional. To explain why $C_N$ has a $3N$-dimensional null space, we start by analyzing the special structure of the compatibility matrix $C_N$ for any periodicity $N$. Let us first take a look at the one-periodic GH mode $\varphi_1(x)$. There are 6 springs in the one-periodic unit cell, and they give six linear constraints on $\varphi_1(x)$ (vertices used here are marked in Figure \ref{fig:self-stresses}(a))
	\begin{align*}
		\big[\varphi_1(B_1) - \varphi_1(C_1)\big] \cdot \hat{b}_1 = 0, \qquad \big[\varphi_1(C_1) - \varphi_1(B_1)\big] \cdot \hat{b}_1 = 0,\\
		\big[\varphi_1(B_1) - \varphi_1(A_1)\big] \cdot \hat{b}_2 = 0, \qquad \big[\varphi_1(A_1) - \varphi_1(B_1)\big] \cdot \hat{b}_2 = 0,\\
		\big[\varphi_1(A_1) - \varphi_1(C_1)\big] \cdot \hat{b}_3 = 0, \qquad \big[\varphi_1(C_1) - \varphi_1(A_1)\big] \cdot \hat{b}_3 = 0,
	\end{align*}
	where $\hat{b}_1, \hat{b}_2, \hat{b}_3$ are unit vectors in the horizontal, 60 degree and 120 degree direction. It can be observed that the six constraints reduce to three linearly independent constraints (the three constraints on the left). Therefore, the one-periodic compatibility matrix $C_1 \in \mathbb{R}^{6\times 6}$ has a 3-dimensional null space; and the space of one-periodic GH modes is 1-dimensional.
	
	This calculation, in fact, gives two important geometric observations that can be generalized to higher periodicity: (a) the periodicity condition on $\varphi_1(x)$ kills one condition on each line in the lattice direction; (b) the linear conditions on lines with different lattice directions are linearly independent. To see them explicitly in the case with higher periodicity, let us consider the two-periodic GH mode $\varphi_2(x)$. There are in total $6*2^2 = 24$ springs in the two-periodic unit cell. To see (a), we first focus on the horizontal solid line across the four vertices $C_1, B_1, C_2, B_2$ in Figure \ref{fig:self-stresses}(d): it gives four linear constraints on $\varphi_2(x)$
	\begin{align*}
		& \big[\varphi_2(B_1) - \varphi_2(C_1)\big] \cdot \hat{b}_1 = 0,\\
		& \big[\varphi_2(C_2) - \varphi_2(B_1)\big] \cdot \hat{b}_1 = 0,\\
		& \big[\varphi_2(B_2) - \varphi_2(C_2)\big] \cdot \hat{b}_1 = 0,\\
		& \big[\varphi_2(C_1) - \varphi_2(B_2)\big] \cdot \hat{b}_1 = 0,
	\end{align*}
	where the last condition is on $\varphi_2(C_1)$ due to periodicity. It can be observed that the last condition is the sum of the first three conditions, and these three conditions are linearly independent. Therefore, the four conditions on the horizontal solid line reduce to three conditions. Similarly, the four conditions on the horizontal dotted line in Figure \ref{fig:self-stresses}(d) reduce to three conditions. The three conditions on the solid line and the three conditions on the dotted line are independent because the vertices are different on the solid and dotted lines. Thus, we obtain (a). We also get six linearly independent conditions in the 60 degree direction and another six linearly independent conditions in the 120 degree direction. It can be checked that the conditions on different lattice directions are linearly independent. This gives (b). Therefore, there are in total $(4-1)*2*3=18$ linearly independent conditions in the two-periodic case (2 means two lines in each lattice direction; and 3 means three lattice directions). Consequently, the compatibility matrix $C_2\in \mathbb{R}^{24\times 24}$ has a 6-dimensional null space; and the space of two-periodic GH modes is 4-dimensional.
	
	The two geometric observations (a)-(b) also work in the $N$-periodic case. Let us first count the reductions of the conditions. In the horizontal direction, there are $N$ lines in the $N$-periodic unit cell, and each line contains $2N$ springs. Using (a), we reduce the $2N^2$ conditions in the horizontal direction by amount $N$. We can reduce the same amount in the 60 and 120 degree directions. Therefore, we reduce the conditions by a total amount $3N$. Using (b), we know the remaining $6N^2-3N$ conditions  are linearly independent. This indicates that the compatibility matrix in the $N$-periodic case $C_N \in \mathbb{R}^{(6N^2) \times (6N^2)}$ has a $3N$-dimensional null space; and the space of $N$-periodic GH modes is $(3N-2)$-dimensional.
	
	We can easily give a basis for the space of $N$-periodic self-stresses ($\ker(C_N^T)$). The standard Kagome lattice has straight lines in the three lattice directions, i.e. horizontal, 60 and 120 degree directions. Let us consider self-stresses that are constant on a line in one of the lattice directions, and zero elsewhere. The one-periodic standard Kagome lattice has three linearly independent self-stresses, shown in Figure \ref{fig:self-stresses}(a)-(c). The two-periodic standard Kagome lattice has six self-stresses, shown in Figure \ref{fig:self-stresses}(d)-(f). We have two linearly independent self-stresses in each lattice direction, since each lattice direction has two straight lines (see the solid and dotted lines in Figure \ref{fig:self-stresses}(d)-(f)). For the $N$-periodic standard Kagome lattice, there are $N$ straight lines in each lattice direction. Thus, we get $3N$ linearly independent $N$-periodic self-stresses, and they form a basis of the space of $N$-periodic self-stresses.
	\begin{figure}[!htb]
		\begin{minipage}[b]{.32\linewidth}
			\centering
			\subfloat[]{\label{}\includegraphics[width=0.8\linewidth]{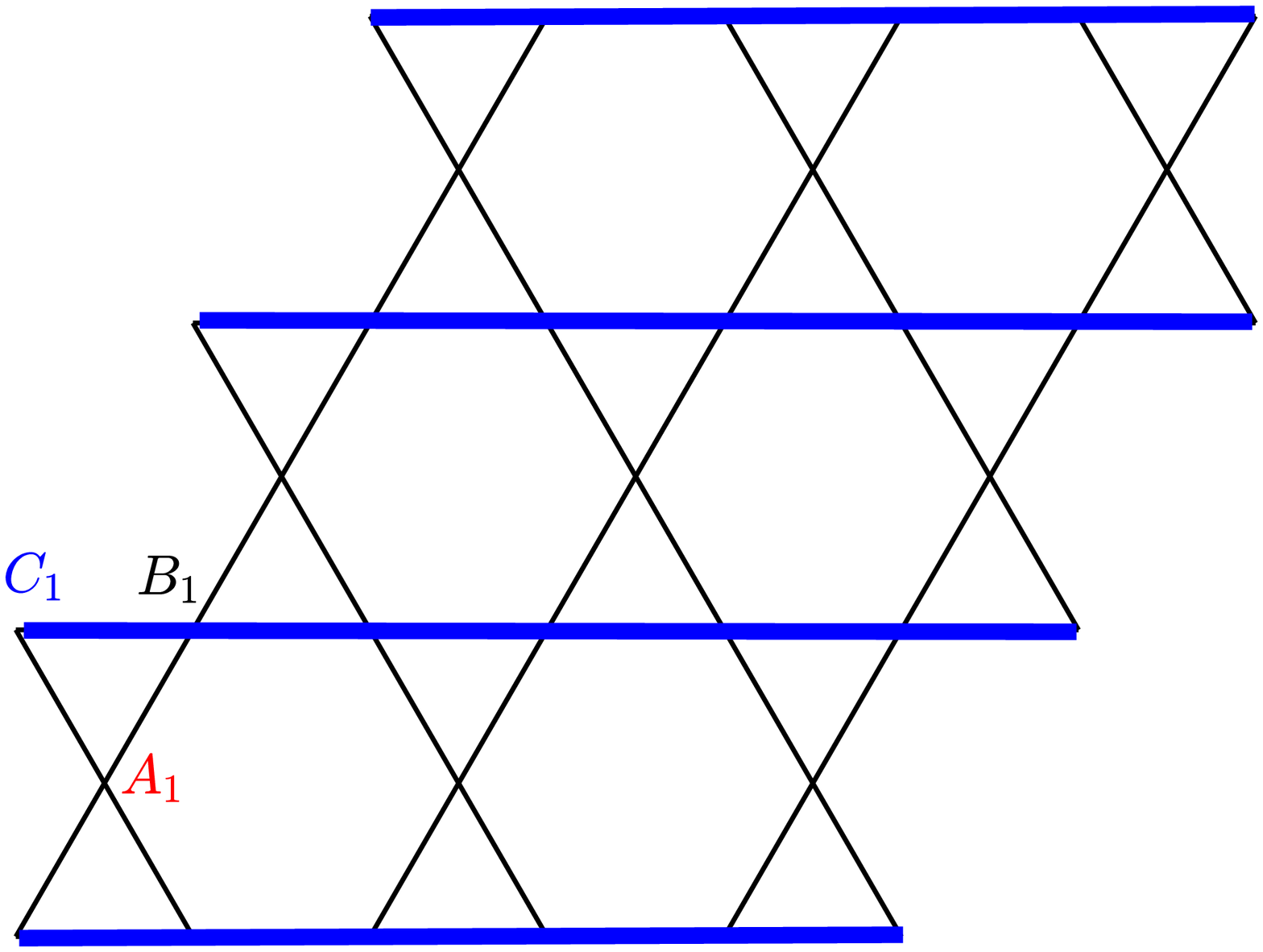}}
		\end{minipage}
		\begin{minipage}[b]{.32\linewidth}
			\centering
			\subfloat[]{\label{}\includegraphics[width=0.8\linewidth]{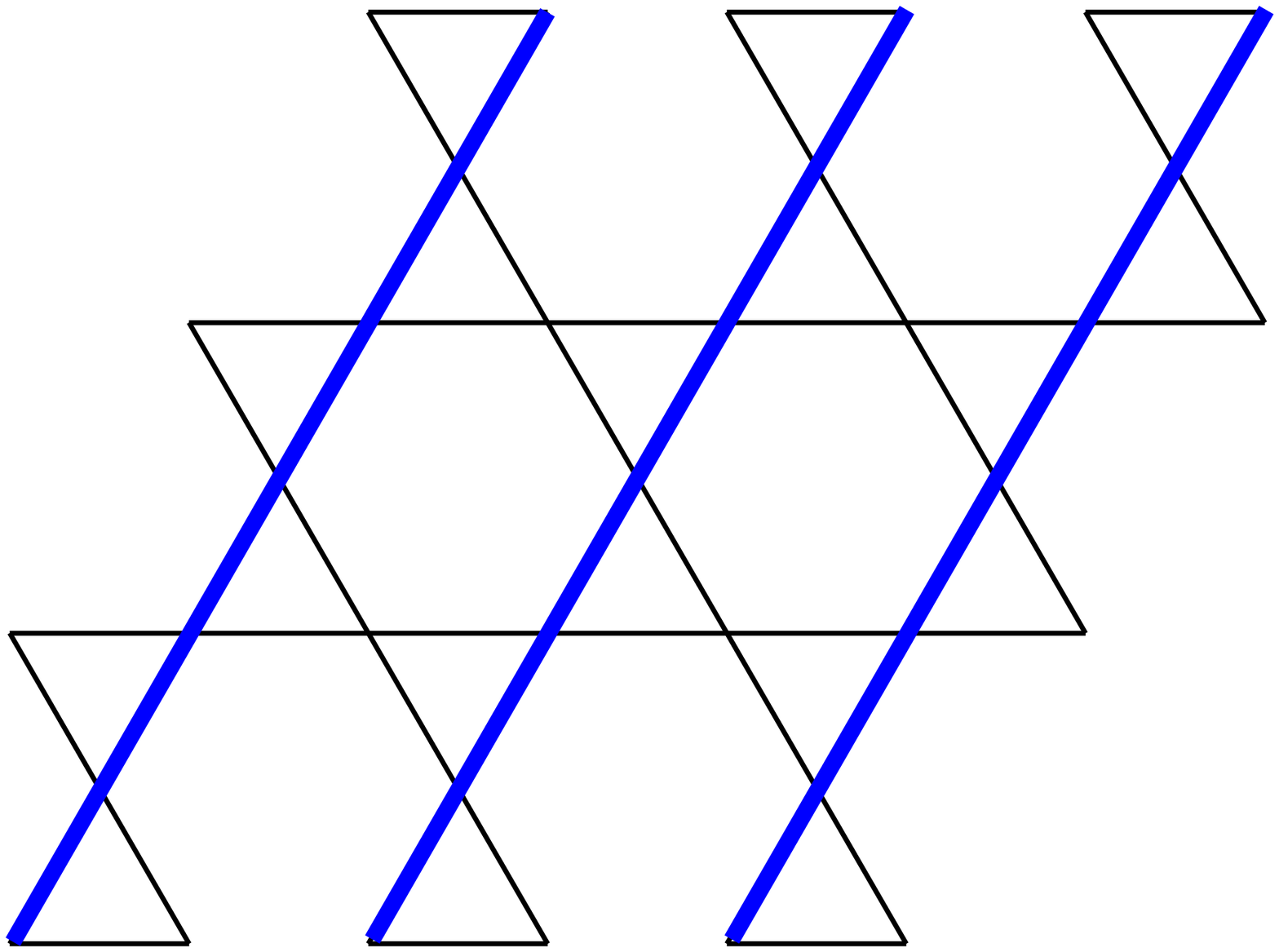}}
		\end{minipage}
		\begin{minipage}[b]{.32\linewidth}
			\centering
			\subfloat[]{\label{}\includegraphics[width=0.8\linewidth]{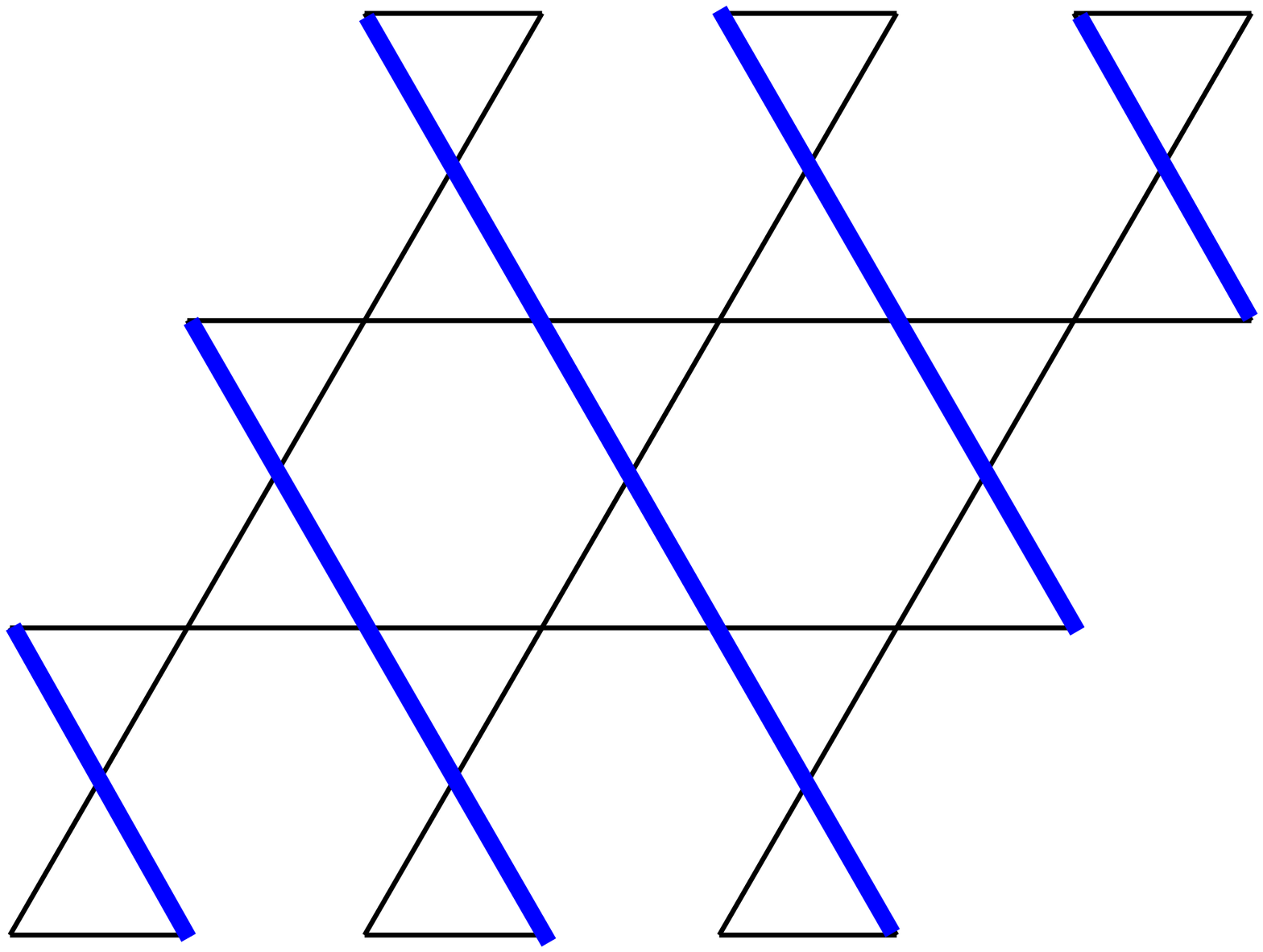}}
		\end{minipage}\par \medskip
		\begin{minipage}[b]{.32\linewidth}
			\centering
			\subfloat[]{\label{}\includegraphics[width=0.8\linewidth]{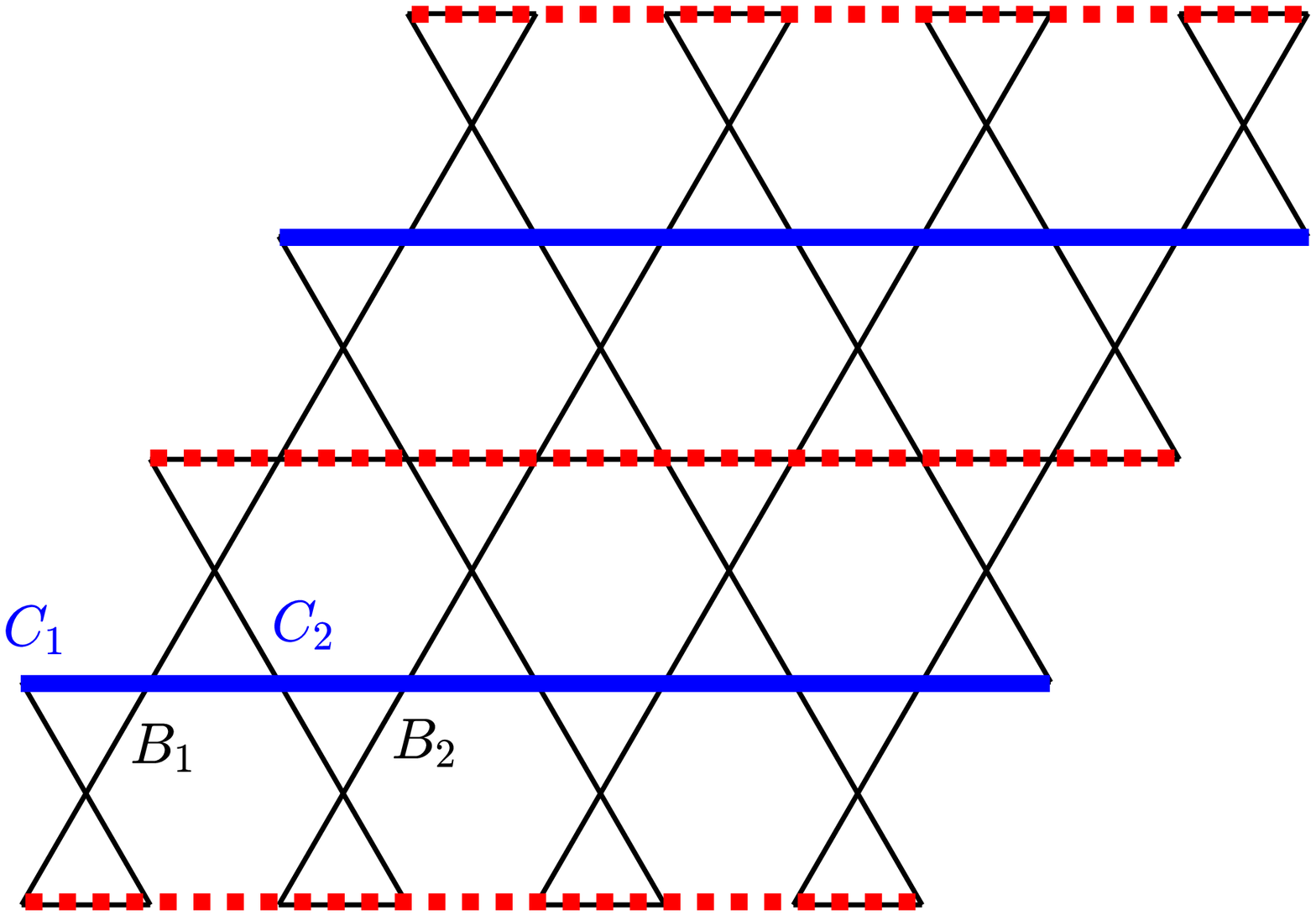}}
		\end{minipage}
		\begin{minipage}[b]{.32\linewidth}
			\centering
			\subfloat[]{\label{}\includegraphics[width=0.8\linewidth]{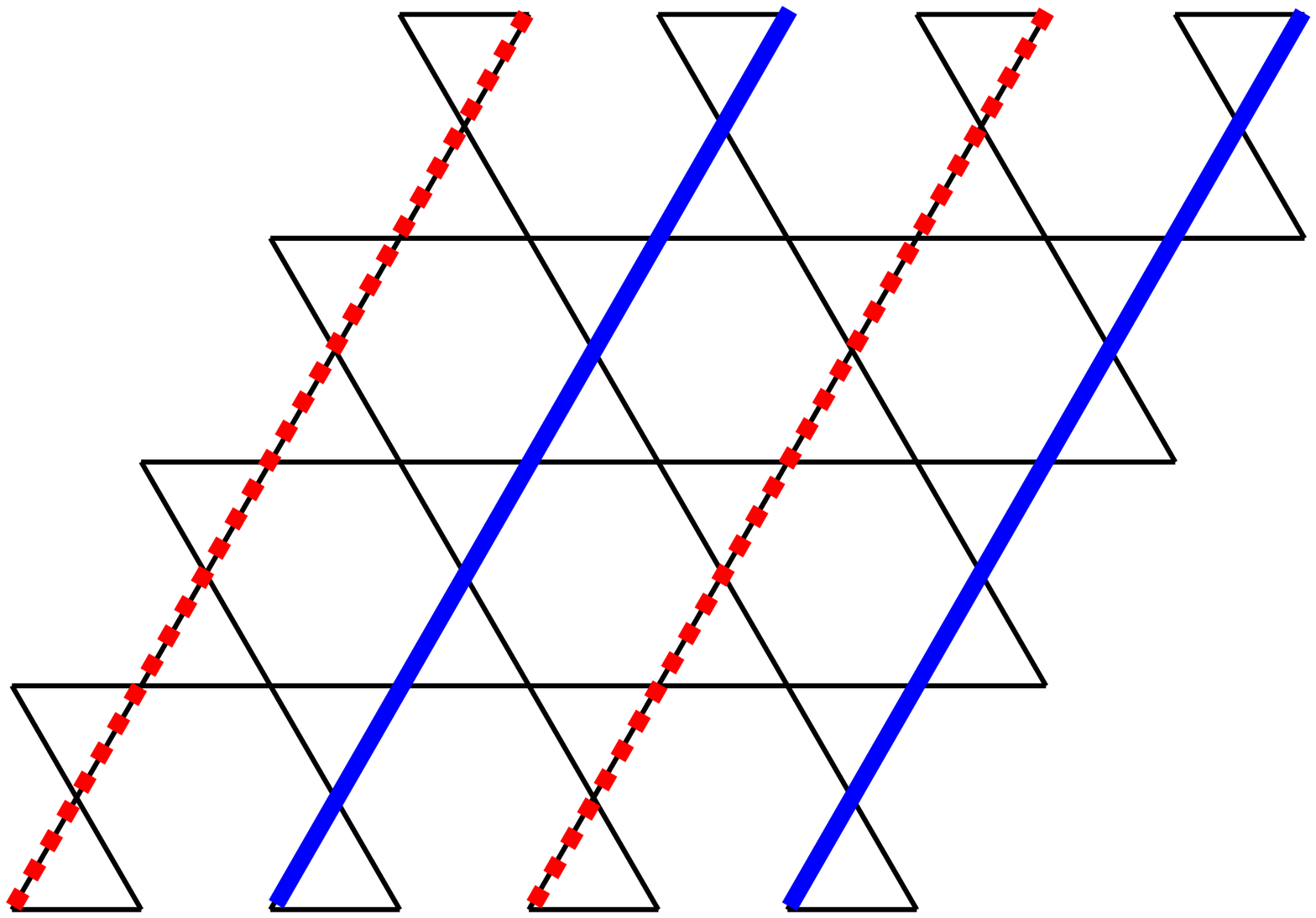}}
		\end{minipage}
		\begin{minipage}[b]{.32\linewidth}
			\centering
			\subfloat[]{\label{}\includegraphics[width=0.8\linewidth]{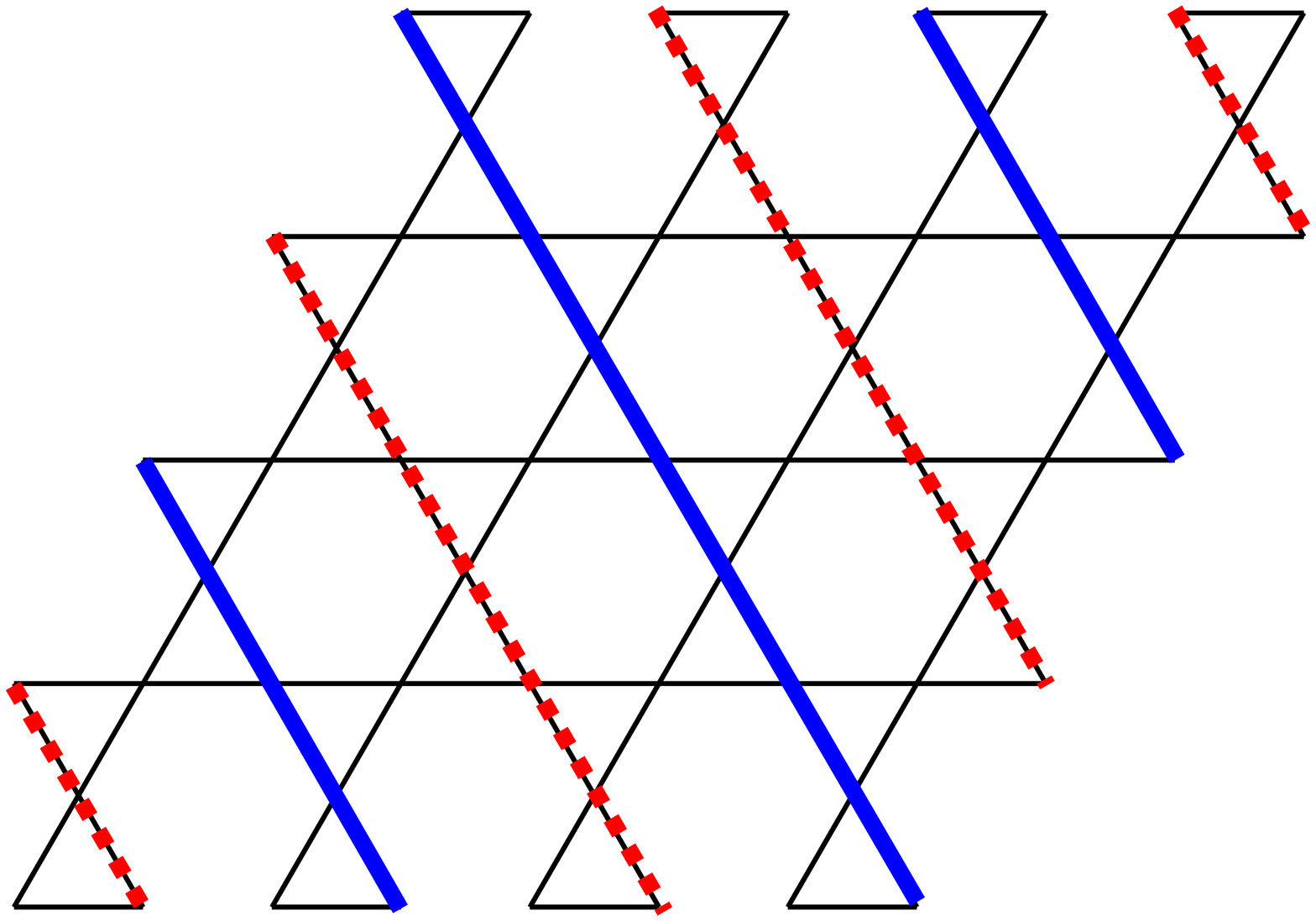}}
		\end{minipage}
		\caption{Self-stresses on the standard Kagome lattice: (a)-(c) show the three one-periodic self-stresses: each self-stress is constant on the solid line and zero elsewhere; (d)-(f) show the six two-periodic self-stresses. Each plot contains two self-stresses: one is constant on the solid line and zero elsewhere; another is constant on the dotted line and zero elsewhere.}
		\label{fig:self-stresses}
	\end{figure}
	
	\subsection{The effective Hooke's law}\label{subsec:effective}
	{\color{black} There is a huge literature on the effective behavior of spatially periodic mechanical systems. Some of it studies "cellular materials" (see e.g. \cite{cioranescu2012homogenization, gibson_ashby_1997}), and some of it studies lattices of springs connected at nodes where rotations are free (see e.g. \cite{guest2003determinacy, hutchinson2006structural, mao2018maxwell}). Our focus here is on the latter class of structures. The effective Hooke's law of a periodic linearly elastic structure can be understood in various ways, including asymptotic expansions, energy minimizations and the relationship between macroscopic stresses and strains (see e.g. \cite{cioranescu1999introduction, hill1963elastic}). For our limited purposes -- an understanding of GH modes -- we need to discuss the linearly elastic effective Hooke's law of a lattice system of springs. To make our paper self-contained, and since we are not aware of a convenient treatment elsewhere, we review this topic here\footnote{\color{black}We do \textit{not} claim that the linear effective behavior gives a good description of the mechanical response of a system like the Kagome lattice; but it is nevertheless needed for the study of GH modes}.}
	
	{\color{black}Making a choice, we shall emphasize the role of elastic energy minimization. By analogy with what is done for periodic elastic composites, for any symmetric strain $\xi \in \mathbb{R}^{2\times 2}_{\text{sym}}$, the effective linear elastic energy of a 2D lattice of springs at strain $\xi$ is the minimum average energy obtained by a displacement with average strain $\xi$. Since every such displacement can be expressed as $v(x) = \xi x + \varphi(x)$ with $\varphi(x)$ periodic, our starting point is the definition}
	\begin{align}\label{effective_energy}
		E_{\text{eff}}(\xi) &= \frac{1}{S}\min_{\substack{\varphi(x) \text{ is }\\ Q\text{-periodic}}} F(\xi, \varphi), &F(\xi, \varphi) &= \sum_{i \sim j} {\color{black} \frac{1}{2}} k_{ij} \bigg(l_{ij} \hat{b}_{ij}^T \xi \hat{b}_{ij} + \big\langle\varphi(x_i) - \varphi(x_j), \hat{b}_{ij} \big \rangle \bigg)^2,
	\end{align}
	where $S$ is the area of the unit cell $Q$ and $k_{ij}, l_{ij},\hat{b}_{ij}$ are defined near \eqref{nonlinear_elastic_energy} and \eqref{first_order_bond_extension}. Here the $Q$-periodic function $\varphi(x)$ shares the same period $Q$ as the lattice. The squared term in $F(\xi,  \varphi)$ is the first-order spring extension $e_{ij}$ for the displacement $v(x) = \xi x + \varphi(x)$ using \eqref{first_order_bond_extension}, i.e.
	\begin{align}\label{eqn:first_order_bond_effective}
		e_{ij} = \big\langle \xi(x_i - x_j) + \varphi(x_i) - \varphi(x_j), \hat{b}_{ij}\big\rangle = l_{ij} \hat{b}_{ij}^T \xi \hat{b}_{ij} + \big\langle \varphi(x_i) - \varphi(x_j), \hat{b}_{ij}\big\rangle.
	\end{align}
	Therefore, the objective function $F(\xi, \varphi) = \sum_{i \sim j} {\color{black}\frac{1}{2}} k_{ij} e_{ij}^2$ is the linear elastic energy for the displacement $v(x) = \xi x + \varphi(x)$, summed over all springs associated with the unit cell $Q$. We observe that the objective function is quadratic and convex for a given $\xi \in \mathbb{R}^{2\times 2}_{\text{sym}}$; hence, an optimal solution $\varphi_\xi(x)$ must exist. However, the optimal solution is not unique because the objective function is not strictly convex. We can, however, choose a special $\varphi^*_\xi(x)$ such that it is an optimal solution, uniquely determined by $\xi$, and linear in $\xi$ (the $\varphi^*_\xi(x)$ is, in fact, the minimum norm solution; see Appendix \eqref{eqn:unique-varphi} for the exact formula for $\varphi^*_\xi(x)$). We stick to the notation that $\varphi_\xi^*(x)$ is the optimal solution with minimum norm and $\varphi_\xi(x)$ is any optimal solution.
	
	The effective linear elastic energy $E_{\text{eff}}(\xi)$ written as a minimization problem {\color{black}is essentially the same as the definition in \cite{hutchinson2006structural}}, where $E_{\text{eff}}(\xi)$ is the average linear elastic energy for a displacement $\widetilde{v}_{\xi}(x)=\xi x + \widetilde{\varphi}_\xi(x)$ such that the tension caused by $\widetilde{v}_{\xi}(x)$ is a self-stress. To see why the two definitions match, we observe that our optimal $\varphi^*_\xi(x)$ in \eqref{effective_energy} yields a self-stress
	\begin{align}\label{special_tension_main}
		t^*_{\xi, ij}&= k_{ij} \bigg(l_{ij} \hat{b}_{ij}^T \xi \hat{b}_{ij} + \big\langle \varphi^*_\xi(x_i) - \varphi^*_\xi(x_j), \hat{b}_{ij}\big\rangle\bigg).
	\end{align}
	This comes from the optimality of $\varphi^*_\xi(x)$ (in fact, any optimal $\varphi_\xi(x)$ also yields a self-stress), and is conveniently shown using linear algebra as we do in Appendix \ref{appendix-a}. Therefore, by taking $\widetilde{\varphi}_\xi(x) = \varphi^*_\xi(x)$, we achieve that our definition of $E_{\text{eff}}(\xi)$ as a minimization problem is the same as the traditional definition. We also show in Appendix \ref{appendix-a} that the effective linear elastic energy $E_{\text{eff}}(\xi)$ is quadratic in $\xi$ and independent of the size of the periodic unit cell (see Lemma \ref{quadratic_law} and Proposition \ref{prop:independence-of-size}). Hence, the effective linear elastic energy $E_{\text{eff}}(\xi)$ has the form
	\begin{align}\label{eqn:effective_energy_quadratic_main}
		E_{\text{eff}}(\xi) &= {\color{black} \frac{1}{2}} \langle A_{\text{eff}}\xi, \xi \rangle,
	\end{align}
	where $A_{\text{eff}}$ is the effective tensor. This 4th order tensor $A_{\text{eff}}$ is called the effective Hooke's law. 
	
	The physical meaning of the effective Hooke's law is that when a lattice material achieves a strain $\xi$ on the macroscopic scale, it generates a macroscopic stress $\bar{\sigma} = A_{\text{eff}}\xi$ and the overall elastic energy is $E_{\text{eff}}(\xi) = \langle \bar{\sigma}, \xi\rangle$. In fact, on the microscopic scale, the macroscopic strain $\xi$ is achieved by the special displacement $v(x) = \xi x + \varphi^*_{\xi}(x)$, where $\varphi^*_\xi(x)$ is the optimal solution in \eqref{effective_energy}; and the macroscopic stress $\bar{\sigma}$ is locally achieved by the self-stress $t^*_{\xi, ij}$ in \eqref{special_tension_main} (see Lemma \ref{quadratic_law} in Appendix \ref{appendix-a}), since
	\begin{align}\label{eqn:macroscopic_stress_main}
		\bar{\sigma} = A_\text{eff}\xi = \frac{1}{S} \sum_{i \sim j} t^*_{\xi, ij} l_{ij} \hat{b}_{ij} \otimes \hat{b}_{ij}.
	\end{align}
	It is easy to observe that $\bar{\sigma} = A_\text{eff}\xi$ is symmetric and depends linearly in $\xi$, since $t_{\xi, ij}$ in \eqref{special_tension_main} depends linearly in $\xi$. The image space $\Ima(A_{\text{eff}}) = \{A_{\text{eff}}\xi| \xi \in \mathbb{R}_{\text{sym}}^{2\times 2}\}$ consists of the macroscopic stresses that can be achieved by the lattice material. If a lattice material sustains all macroscopic stresses, i.e. $\dim \Ima(A_{\text{eff}}) = 3$, then we call it \textit{non-degenerate}. Non-degeneracy is also known as the \textit{static stability}. With the effective energy in the form of \eqref{eqn:effective_energy_quadratic_main} and  \eqref{eqn:macroscopic_stress_main}, we can prove the Guest-Hutchinson Theorem by a simple counting argument.
	
	\begin{proof}[Proof of Theorem \ref{thm:GH}]
		If $A_{\text{eff}}$ is non-degenerate, then $\Ima(A_{\text{eff}})$ is three-dimensional. The linear relationship between $A_\text{eff}\xi$ and $t^*_{\xi, ij}$ in \eqref{eqn:macroscopic_stress_main} indicates that there must exist three linearly independent self-stresses associated to a basis of $\Ima (A_{\text{eff})}$. Therefore, the space of self-stresses is at least three-dimensional, i.e. $\dim(\ker(C^T)) \geq 3$. Since a Maxwell lattice has a square compatibility matrix $C$, the null space of $C$ is at least three-dimensional because $\dim(\ker(C))  = \dim(\ker(C^T)) \geq 3$. So besides the two translations, there must exist a GH mode.
	\end{proof}
	
	\section{Periodic mechanisms and GH modes}\label{sec:periodic-mechanism-GH-modes}
	We explore the idea that a periodic mechanism reveals at least one of the following degeneracies: (1) a macroscopic degeneracy, in the form of a non-trivial null vector of $\xi \in \ker(A_{\text{eff}})$; (2) a microscopic degeneracy, in the form of a GH mode. We also review the one-periodic mechanism of the Kagome lattice and use it as a transparent example to illustrate how it reveals a null vector of $A_{\text{eff}}$ for the twisted Kagome lattice and a GH mode for the standard Kagome lattice. The relationship between periodic mechanisms and GH modes also raises an interesting question: are GH modes always linearizations of some periodic mechanisms? For the one-periodic standard Kagome lattice, the answer is yes since the only GH mode is the linearization of the one-periodic mechanism (see section \ref{one-periodic-mechanism}). However, for the two-periodic standard Kagome lattice, the answer is no, as we will discuss in section \ref{two-periodic-GH-modes}.
	\subsection{GH modes and infinitesimal versions of periodic mechanisms}\label{subsec:GH-mechanism}
	In section \ref{sec:preliminary}, we have studied the linear elastic mechanics of a lattice, involving small displacements $v(x)$. From now on, we switch gears to consider mechanisms, i.e. large deformations that have zero nonlinear elastic energy. Our notation reflects this distinction by using $v(x)$ for linear displacements and $u(x)$ for nonlinear deformations. We focus primarily on periodic mechanisms. A \textit{periodic mechanism} of a lattice material $u(x,t) \in \mathbb{R}^2$ is a smooth deformation parameterized by $t$ that preserves the lengths of all the springs and transforms the reference lattice into a different periodic structure (a new lattice) for all $t \in [-t_0, t_0]$. We emphasize that a periodic mechanism $u(x,t)$ is typically not a periodic function of $x$. It deforms the reference lattice to a different lattice that might have a different unit cell. Therefore, a periodic mechanism $u(x,t)$ has a macroscopic deformation gradient $F(t) \in \mathbb{R}^{2\times 2}$ that deforms the unit cell of the reference lattice to that of the deformed lattice at time $t$. In other words, if $\bm{v}_1, \bm{v}_2 \in \mathbb{R}^2$ are primitive vectors of the reference lattice, then $F(t) \bm{v}_1$ and $F(t) \bm{v}_2$ are primitive vectors of the deformed lattice at time $t$. We can write the periodic mechanism in the form $u(x,t) = F(t) \cdot x + \varphi(x,t)$, where $\varphi(x,t) \in \mathbb{R}^2$ is periodic in $x$ for all $t$. The periodicity of $\varphi(x,t)$ depends on the periodic mechanism $u(x,t)$. If $u(x,t)$ is $N$-periodic, then $\varphi(x,t)$ is $N$-periodic in $x$. An example of a one-periodic mechanism is shown in Figure \ref{fig:one-periodic-mechanism}(a). We observe that the periodic structure on the left is deformed into a different periodic structure on the right; there is a macroscopic deformation gradient that transforms the two primitive vectors $\bm{v}_1, \bm{v}_2$ to $\bm{v}_1^{\text{def}}, \bm{v}_2^{\text{def}}$. Translations and rotations are trivial periodic mechanisms, so we consider two periodic mechanisms the same if they differ only by translation and rotation, i.e. $u_1(x,t) = F(t) x + \varphi(x,t)$ and $u_2(x,t) = R(t){\color{black}[F(t)x + \varphi(x,t)] }+ d(t)$ are equivalent in our notation, where $R(t) \in SO(2)$ and $d(t)\in \mathbb{R}^2$ is a translation. By polar decomposition, we can always take $F(t) = R(t) S(t)$, where $R(t) \in SO(2)$ and $S(t)$ is symmetric. Replacing $u(x,t)$ by $R^{-1}(t)u(x,t)$, we can assume that the macroscopic deformation gradient $F(t)$ is always symmetric.
	
	The infinitesimal version of a periodic mechanism around the reference lattice at $t=0$ contains two parts: the infinitesimal macroscopic deformation $\color{black}\frac{d F}{d t} \Big|_{t=0} = \dot{F}(0)$ and the infinitesimal periodic oscillation $\color{black}\frac{\partial \varphi(\cdot,t)}{\partial t}\Big|_{t = 0} = \dot{\varphi}(\cdot, 0)$. The following proposition explains when $\dot{F}(0)$ is a non-trivial null vector of $A_{\text{eff}}$ and when $\dot{\varphi}(x,0)$ is a GH mode.
	\begin{proposition}\label{degeneracy}
		Consider a periodic mechanism $u(x,t) = F(t) \cdot x + \varphi(x,t)$ with $F(t)$ symmetric. If $\dot{F}(0) = 0$ and $\dot{\varphi}(\cdot, 0)$ is not a translation, then $\dot{\varphi}(\cdot, 0)$ is a GH mode for the reference lattice; if $\dot{F}(0) \neq 0$, then $\dot{F}(0)$ is a non-trivial null vector for the effective tensor $A_{\text{eff}}$, i.e. $A_{\text{eff}} \dot{F}(0) = 0$.
	\end{proposition}
	\begin{proof}
		A periodic mechanism preserves the lengths of all springs for any $t \in [-t_0, t_0]$, i.e. $\big|u(x_i, t) - u(x_j,t)\big|^2 = \big|u(x_i,0) - u(x_j,0)\big|^2$ for all connected $x_i, x_j$ and $t \in [-t_0, t_0]$ . {\color{black}{Taking the time derivative}} and evaluating at $t=0$ gives
		\begin{align}\label{infinitesimal_mechanism}
			\big \langle \dot{F}(0)\cdot (x_i - x_j) + \dot{\varphi}(x_i,0) - \dot{\varphi}(x_j,0), u(x_i, 0) - u(x_j, 0) \big \rangle = 0,
		\end{align}
		where $u(x_i, 0) - u(x_j, 0) = x_i - x_j = l_{ij} \hat{b}_{ij}$ is parallel to the spring direction $\hat{b}_{ij}$. This indicates that when $\dot{F}(0)$ vanishes, the periodic $\dot{\varphi}(\cdot, 0)$ corresponds to a null vector of the compatibility matrix of the reference lattice at $t=0$, i.e. $\big \langle \dot{\varphi}(x_i,0) - \dot{\varphi}(x_j,0), \hat{b}_{ij} \big \rangle = 0$. By the assumption that $\dot{\varphi}(x,0)$ is not a translation, it must be a GH mode.
		
		When $\dot{F}(0) \neq 0$, it is a non-trivial null vector of $A_{\text{eff}}$. To see why, we observe that $\big \langle \dot{F}(0) \cdot (x_i - x_j), \hat{b}_{ij} \big \rangle = l_{ij} \hat{b}_{ij}^T \dot{F}(0)\hat{b}_{ij}$. Hence from \eqref{infinitesimal_mechanism}, we get
		\begin{align*}
			l_{ij} \hat{b}_{ij}^T \dot{F}(0) \hat{b}_{ij} + \big \langle \dot{\varphi}(x_i,0) - \dot{\varphi}(x_j,0), \hat{b}_{ij} \big \rangle = 0.
		\end{align*}
		This indicates that the first-order spring extension $e_{ij}$ in \eqref{eqn:first_order_bond_effective} for the infinitesimal deformation $u(x) = \dot{F}(0) x + \dot{\varphi}(x, 0)$ vanishes for all springs. Moreover, the effective linear elastic energy $E_{\text{eff}}(\xi)$ vanishes at $\xi = \dot{F}(0)$ because we can choose $\dot{\varphi}(x, 0)$ as the displacement $\varphi_\xi(x)$ in \eqref{effective_energy}. Thus, the macroscopic strain $\dot{F}(0)$ is a non-trivial null vector for the effective tensor $A_{\text{eff}}$.
	\end{proof}
	{\color{black}Proposition \ref{degeneracy} justifies our statement at the beginning of section \ref{sec:periodic-mechanism-GH-modes} that when a lattice has a periodic mechanism, either its linear elastic behavior is macroscopically degenerate (this occurs when $\dot{F}(0) \neq 0$) or else its linear elastic behavior is microscopically degenerate (in the sense that there is a GH mode $\dot{\varphi}(\cdot,0) \neq 0$)\footnote{\color{black}We do not exclude the case where a mechanism induces both macroscopic and microscopic degeneracy. This can happen, for example, in a $2\times 2$ periodic mechanism of a 2D square lattice.}.}
	
	However, Proposition \ref{degeneracy} does not tell us whether a GH mode comes from a periodic mechanism. The answer to this question is not trivial. In general, we shall show in section \ref{two-periodic-GH-modes} that for the two-periodic standard Kagome lattice, there are plenty of GH modes that do not come from mechanisms. But as we review in the following subsection, the space of GH modes for the one-periodic standard Kagome lattice is one-dimensional and its basis vector comes from a one-periodic mechanism. 
	
	\subsection{The one-periodic mechanism and some consequences}\label{one-periodic-mechanism}
	We revisit the well-known one-periodic mechanism of the Kagome lattice, which is a transparent example of Proposition \ref{degeneracy}. As we shall explain, the one-periodic mechanism reveals that (1) the standard Kagome lattice has a GH mode $\dot{\varphi}(x,0)$ because $\dot{F}(0)$ vanishes; (2) every GH mode of the one-periodic standard Kagome lattice is a multiple of $\dot{\varphi}(x,0)$, and therefore every GH mode comes from a scaled version of the one-periodic mechanism; (3) the twisted Kagome lattices are macroscopically degenerate w.r.t isotropic compressions and expansions, i.e. $A_{\text{eff}}I = 0$.
	
	Let us first review the one-periodic mechanism which deforms the standard Kagome lattice to a twisted Kagome lattice. For simplicity, we refer to the twisted Kagome lattice in Figure \ref{fig:one-periodic-mechanism}(a) by $L_{\theta}$, where $2\theta$ is the angle between the two triangles in its unit cell. The standard Kagome lattice corresponds to $\theta = \frac{\pi}{3}$. We get a \textit{one-parameter} one-periodic mechanism from the standard Kagome lattice to a twisted Kagome lattice $L_{\theta}$ by smoothly varying the angle between the two triangles in the unit cell. Geometrically, this one-periodic mechanism rotates the two triangles in the unit cell, which are shaded in Figure \ref{fig:one-periodic-mechanism}(a), by the same amount but in opposite directions. We denote this one-periodic mechanism as $u_{\frac{\pi}{3} \shortto \theta}(x)$, where $x$ are vertices of the standard Kagome lattice. This mechanism can be written as $u_{\frac{\pi}{3}\shortto \theta}(x) = F_\theta \cdot x + \varphi_\theta(x)$, where $F_\theta$ is the macroscopic deformation gradient and $\varphi_\theta(x)$ is the one-periodic oscillation (see Appendix \ref{appendix-b} for the explicit formulas for this one-periodic mechanism). Using the explicit representation of this one-periodic mechanism and the fact that the macroscopic deformation $F_\theta$ maps the primitive vectors of the standard Kagome lattice $\bm{v}_1, \bm{v}_2$  to the primitive vectors $\bm{v}_1^{\text{def}}, \bm{v}_2^{\text{def}}$ of the deformed lattice $L_{\theta}$ shown in Figure \ref{fig:one-periodic-mechanism}(a), we get the formulas of $\bm{v}_1^\text{def}$ and $\bm{v}_2^\text{def}$
	\begin{align*}
		&\bm{v}_1 = (2,0)^T  &\rightarrow&  &\bm{v}_1^{\text{def}} &= F_\theta \bm{v}_1 = \cos(\frac{\pi}{3} - \theta) (2,0)^T,\\
		&\bm{v}_2 = (1,\sqrt{3})^T  &\rightarrow&  &\bm{v}_2^{\text{def}} &= F_\theta \bm{v}_2 =\cos(\frac{\pi}{3} - \theta) (1, \sqrt{3})^T,
	\end{align*}
	Evidently, the macroscopic deformation $F_\theta$ is an isotropic compression
	\begin{align}\label{macroscopic-one-periodic}
		F_\theta = \cos(\frac{\pi}{3} - \theta) I.
	\end{align}
	
	To get the infinitesimal version of this mechanism around the standard Kagome lattice, we change $\theta = \frac{\pi}{3}+t$ and write the one-periodic mechanism as
	\begin{align}\label{local-one-periodic-mechanism}
		u_{\frac{\pi}{3} \shortto \frac{\pi}{3}+t}(x) = F(t)\cdot x + \varphi(x,t),
	\end{align}
	where $x$ are vertices in the standard Kagome lattice. Using \eqref{macroscopic-one-periodic}, the macroscopic deformation is $F(t) = \cos(t) I$. The infinitesimal macroscopic deformation vanishes at the standard Kagome lattice when $t=0$, i.e. $\dot{F}(0) = 0$. Therefore, Proposition \ref{degeneracy} tells us that $\dot{\varphi}(x,0)$ is a GH mode, and its character is shown in Figure \ref{fig:Kagome_GH}(a) ($\dot{\varphi}(x,0)$ is not a translation; see Appendix \ref{appendix-b} for its explicit formula).
	
	For the one-periodic standard Kagome lattice, every GH mode is a linearization of a one-periodic mechanism. In fact, we know that the space of one-periodic GH modes is one-dimensional from section \ref{sec:preliminary-GH}. So, the infinitesimal $\dot{\varphi}(x,0)$ spans the one-dimensional GH mode space. Moreover, every GH mode has the form  $k\dot{\varphi}(x,0)$ for some $k \in \mathbb{R}$, and comes from the scaled one-periodic mechanism $u_{\frac{\pi}{3} \shortto \frac{\pi}{3}-kt}(x)$. 
	
	\begin{remark}
		There is \textit{only a single} one-periodic mechanism for the Kagome lattice. The unit cell of the one-periodic standard Kagome lattice has only two triangles. If these triangles rotate with angles $\alpha,\beta$ as shown Figure \ref{fig:one-periodic-mechanism}(b), then there is a macroscopic rotation. To eliminate this rotation, we can choose $2\theta = \alpha + \beta$ and $\alpha = \beta = \theta$ so that the bisector is in the horizontal direction.
	\end{remark}
	
	\begin{figure}[!htb]
		\begin{minipage}[b]{.65\linewidth}
			\centering
			\subfloat[]{\label{}\includegraphics[width=\linewidth]{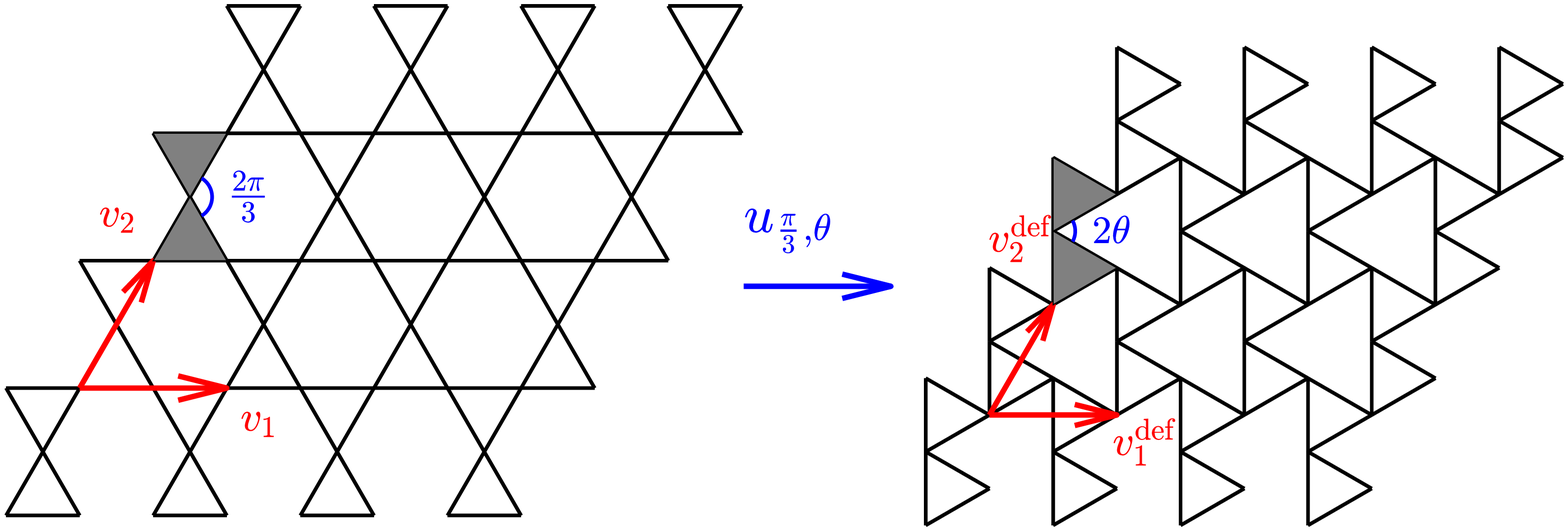}}
		\end{minipage}
		\hfill
		\begin{minipage}[b]{.28\linewidth}
			\centering
			\subfloat[]{\label{}\includegraphics[width=0.8\linewidth]{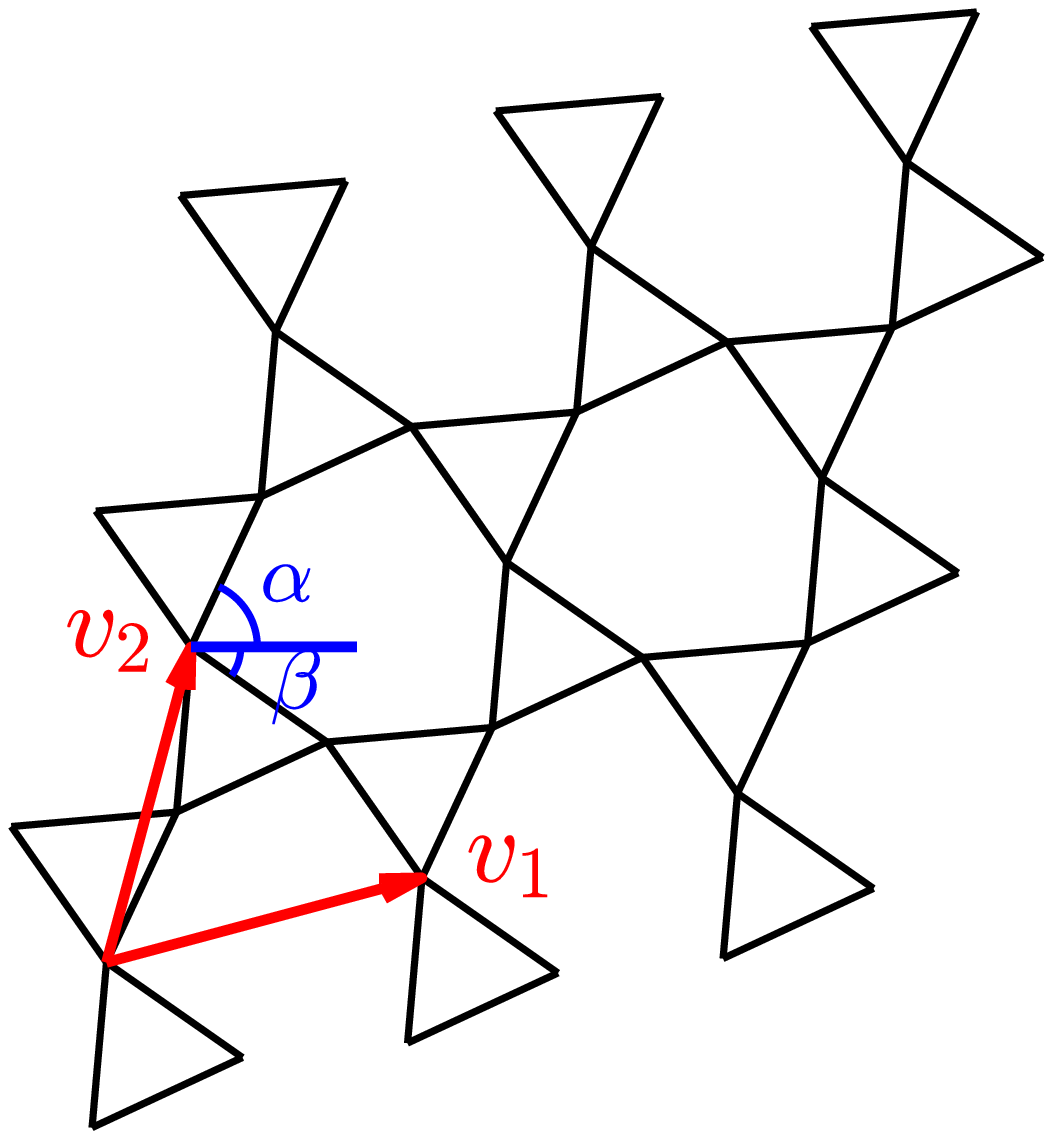}}
		\end{minipage}\par\medskip
		\centering
		\begin{minipage}[b]{.65\linewidth}
			\subfloat[]{\label{}\includegraphics[width=\linewidth]{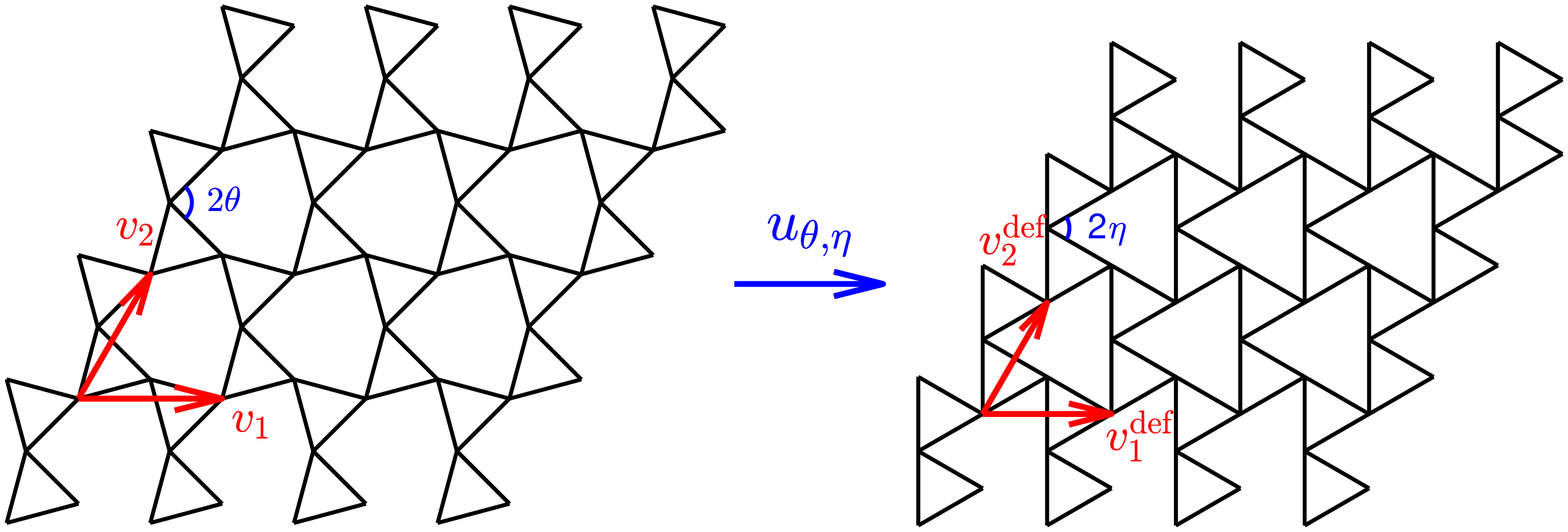}}
		\end{minipage}
		\caption{One-periodic mechanisms of the Kagome lattice: (a) the one-parameter mechanism $u_{\frac{\pi}{3}, \theta}(x)$ from the standard Kagome lattice to a twisted Kagome lattice $L_{\theta}$; (b) a rotated version of the twisted Kagome lattice: if two triangles in the unit cell rotate with different angles $\alpha \neq \beta$, then there is a macroscopic rotation; (c) the one-parameter mechanism $u_{\theta, \eta}(x)$ from the twisted Kagome lattice $L_{\theta}$ to a different twisted Kagome lattice$L_{\eta}$ with $\theta \neq \eta$.}
		\label{fig:one-periodic-mechanism}
	\end{figure}
	Thus far, we have been discussing mechanisms of the standard Kagome lattice. But since our mechanism takes the standard Kagome lattice to a twisted Kagome lattice, it also provides a mechanism for the twisted Kagome lattice by considering $u_{\theta\shortto \eta}(x) = F_{\theta \shortto \eta}\cdot x + \varphi_{\theta \shortto \eta}(x)$ taking a twisted Kagome lattice $L_{\theta}$ to a different twisted Kagome lattice $L_{\eta}$ in Figure \ref{fig:one-periodic-mechanism}(c), where $F_{\theta \shortto \eta}$ is the macroscopic deformation and $\varphi_{\theta \shortto \eta}(x)$ is the periodic oscillation. Notice that $x$ now ranges over the vertices of the twisted Kagome lattice $L_{\theta}$. The macroscopic deformation gradient $F_{\theta \shortto \eta}$  maps the primitive vectors $\bm{v}_1, \bm{v}_2$ of the twisted Kagome lattice $L_{\theta}$ to the primitive vectors $\bm{v}_1^{\text{def}}, \bm{v}_2^{\text{def}}$ of a different twisted Kagome lattice $L_{\eta}$
	\begin{align*}
		& \bm{v}_1 = \cos(\frac{\pi}{3} - \theta) (2,0)^T &\rightarrow&  &\bm{v}_1^{\text{def}} &= F_{\theta \shortto \eta} \bm{v}_1 = \cos(\frac{\pi}{3} - \eta) (2,0)^T,\\
		&\bm{v}_2 = \cos(\frac{\pi}{3} - \theta) (1, \sqrt{3})^T &\rightarrow&  &\bm{v}_2^{\text{def}} &= F_{\theta \shortto \eta} \bm{v}_2 =\cos(\frac{\pi}{3} - \eta) (1, \sqrt{3})^T,
	\end{align*}
	Thus, the macroscopic deformation gradient $F_{\theta \shortto \eta}$ is
	\begin{align}\label{macroscopic_twisted_mechanism}
		F_{\theta \shortto \eta} = \frac{\cos(\frac{\pi}{3} - \eta)}{\cos(\frac{\pi}{3} - \theta) } I.
	\end{align}
	
	For any twisted Kagome lattice $L_{\theta}$ with $\theta \neq \frac{\pi}{3}$, its effective tensor $A^{\theta}_{\text{eff}}$ vanishes at the identity matrix, i.e. $A_\text{eff}^\theta I = 0$ for any $\theta \neq \frac{\pi}{3}$. To see why, we observe that similarly to the standard Kagome lattice in \eqref{local-one-periodic-mechanism}, the one-periodic mechanism around a twisted Kagome lattice $L_{\theta}$ is 
	\begin{align}\label{mechanism_around_twisted_kagome}
		u_{\theta \shortto \theta + t}(x) = F(t) \cdot x + \varphi(x,t),
	\end{align}
	by choosing the deformed state $L_\eta$ as $\eta = \theta + t$. The macroscopic deformation in \eqref{mechanism_around_twisted_kagome} and its infinitesimal version become
	\begin{align*}
		F(t) &= \frac{\cos(\frac{\pi}{3} - \theta - t)}{\cos(\frac{\pi}{3} - \theta)} I, & \dot{F}(0) &= \frac{\sin(\frac{\pi}{3}-\theta)}{\cos(\frac{\pi}{3} - \theta)} I \neq 0.
	\end{align*}
	Using Proposition \ref{degeneracy}, we obtain that $\dot{F}(0) = c_\theta I$ is a multiple of the identity matrix and a null vector for the effective tensor $A_{\text{eff}}^{\theta}$ with
	\begin{align}\label{isotropic_degeneracy}
		A_{\text{eff}}^{\theta} \: c_\theta I = 0 \qquad \Leftrightarrow \qquad A_{\text{eff}}^{\theta} I = 0,
	\end{align}
	where $c_\theta = \frac{\sin(\theta-\frac{\pi}{3})}{\cos(\frac{\pi}{3} - \theta)}$. Thus, for any twisted Kagome lattice $L_{\theta}$, its effective tensor $A^\theta_{\text{eff}}$ vanishes at isotropic compression and expansion.
	
	\begin{figure}[!htb]
		\begin{minipage}[b]{.48\linewidth}
			\centering
			\subfloat[]{\label{}\includegraphics[width=.75\linewidth]{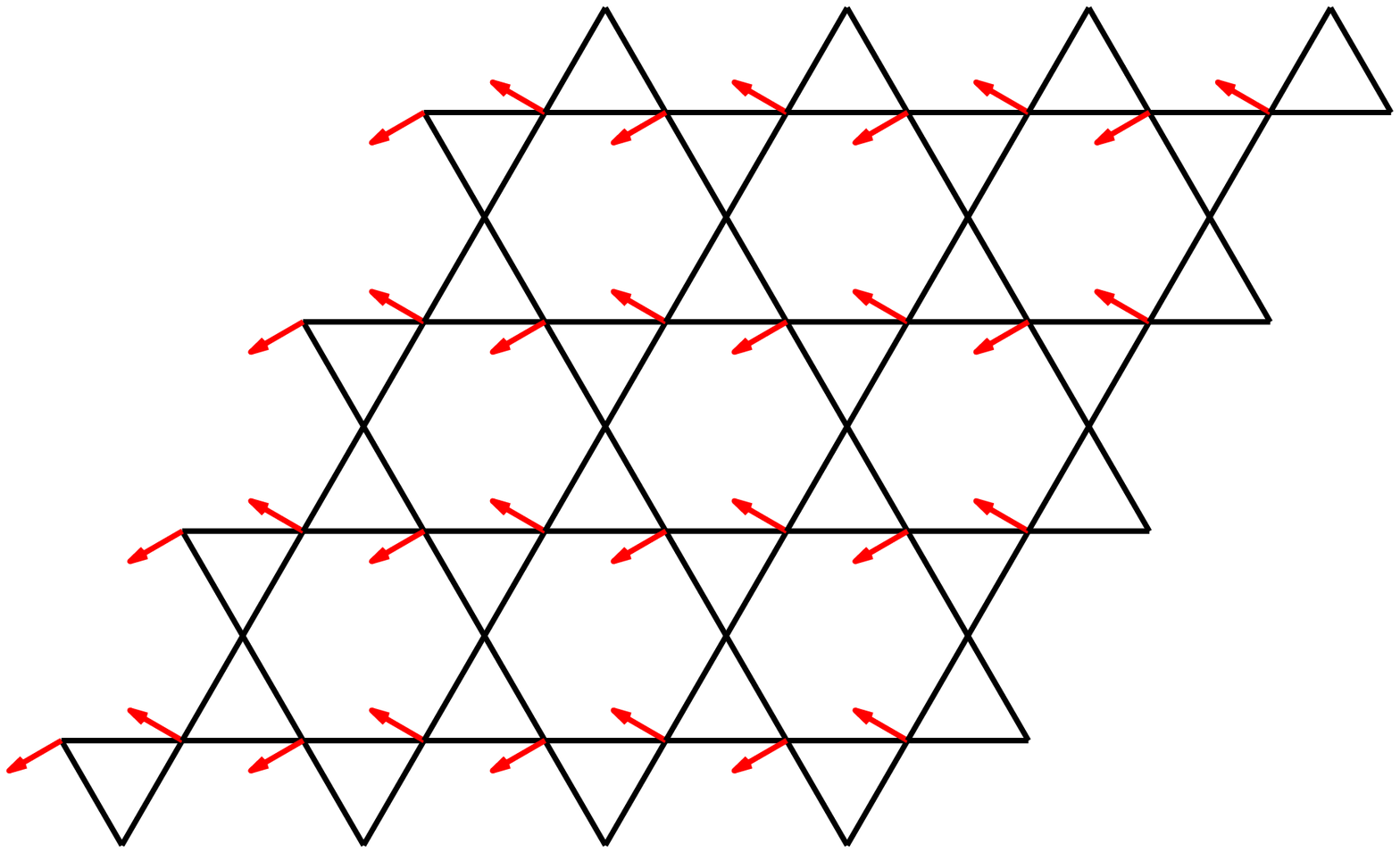}}
		\end{minipage}
		\begin{minipage}[b]{.48\linewidth}
			\centering
			\subfloat[]{\label{}\includegraphics[width=.6\linewidth]{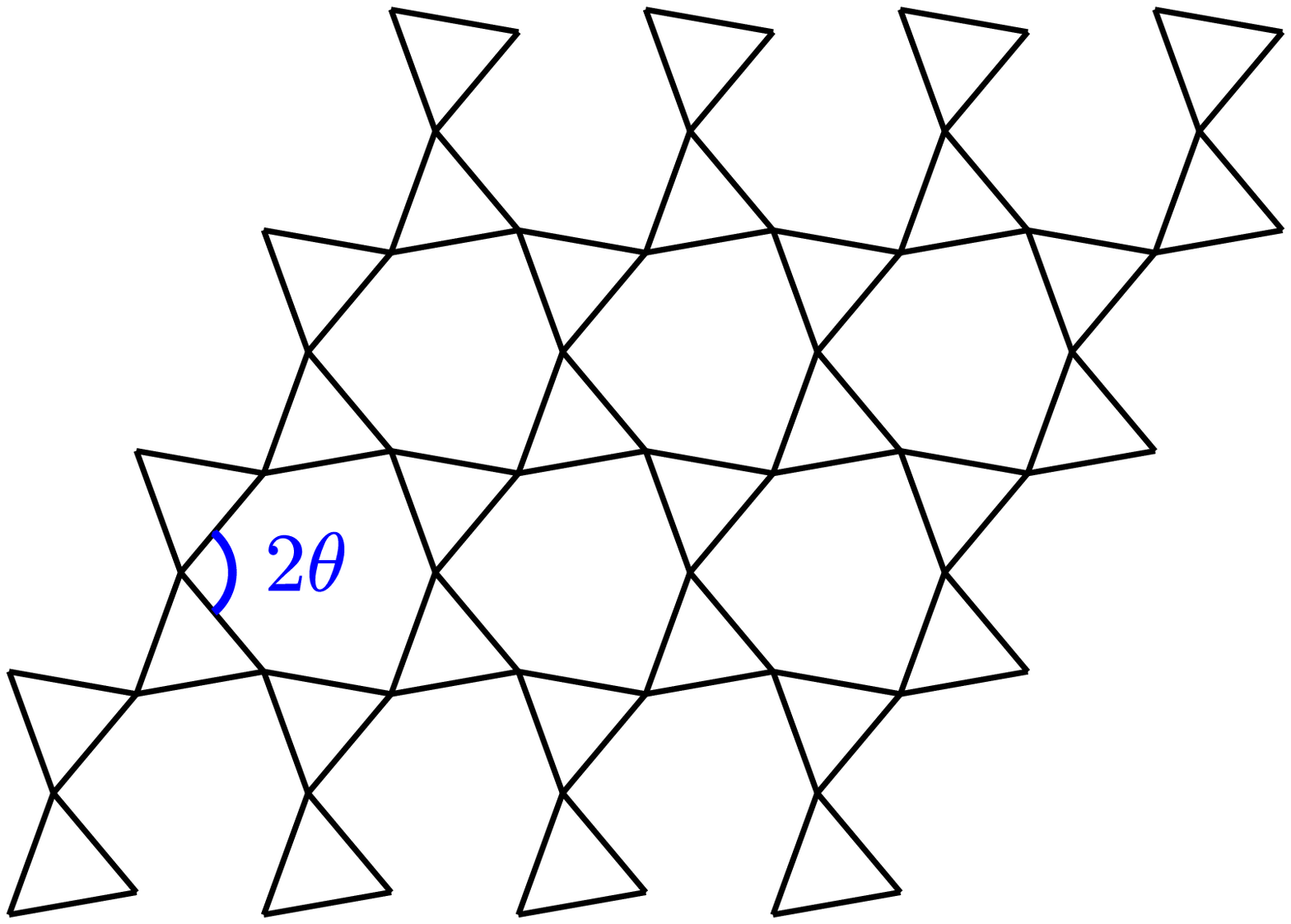}}
		\end{minipage}
		\caption{(a) The one-periodic GH mode on the standard Kagome lattice $\dot{\varphi}(x,0)$ as the infinitesimal version of the one-periodic mechanism; (b) the twisted Kagome lattice does not have any one-periodic GH modes. Instead, it has a macroscopic degeneracy for isotropic compression and expansion, i.e. $A_{\text{eff}}I = 0$.}
		\label{fig:Kagome_GH}
	\end{figure}
	
	\begin{remark}
		A geometric way to see that $\dot{F}(0)=0$ for the standard Kagome lattice but not for any twisted Kagome lattice is to note that $F(t)$ is always a multiple of identity, so $\dot{F}(0)$ controls how the size of the unit cell changes. The standard Kagome lattice has the largest unit cell since the mechanism can only shrink the area of each hexagon. Thus $\dot{F}(0) = 0$ for the standard Kagome lattice. For any twisted Kagome lattice $L_{\theta}$ with $\theta < \frac{\pi}{3}$, for example in Figure \ref{fig:Kagome_GH}(b), the area of each hexagon increases when we increase $\theta$ and decreases when we decrease $\theta$. The case where $\theta > \frac{\pi}{3}$ is similar, and the area of each hexagon increases and decreases by changing $\theta$ oppositely.
	\end{remark}
	
	\subsection{Discontinuity of the effective tensor $A_{\text{eff}}$ under the one-periodic mechanism}\label{subsec:discontinuity}
	So far, we have explained that the effective tensor $A_{\text{eff}}^{\theta}$ is non-degenerate for the standard Kagome lattice ($\theta = \frac{\pi}{3}$) and degenerate in the isotropic direction for all twisted Kagome lattices ($\theta \neq \frac{\pi}{3}$). If we view the effective tensor $A_{\text{eff}}^{\theta}$ as a tensor-valued function of $\theta$, then $A_{\text{eff}}^{\theta}$ is discontinuous at the standard Kagome lattice since $A_{\text{eff}}^{\frac{\pi}{3}} I \neq 0$ for the standard Kagome lattice but $A_{\text{eff}}^{\theta} I = 0$ in Equation \eqref{isotropic_degeneracy} for all the twisted Kagome lattices when $\theta \neq \frac{\pi}{3}$. If we believe that the displacement $\varphi^\theta_I(x)$ associated with the effective linear elastic energy evaluated at the identity matrix $E_{\text{eff}}^\theta(I)$ is continuous in $\theta$, then we cannot have a discontinuity $A_{\text{eff}}^{\theta}I$ at $\theta = \frac{\pi}{3}$. In fact, it is wrong to believe that the optimal solution $\varphi^\theta_I(x)$ is continuous in $\theta$. To explain this discontinuity, we shall calculate the optimal solution $\varphi^\theta_I(x)$ in \eqref{effective_energy} at the identity matrix $I$. The result shows that the optimal $\varphi^\theta_I(x)$ grows unbounded as $\theta$ approaches $\frac{\pi}{3}$. 
	
	To compute the optimal $\varphi^\theta_I(x)$, we need to use the one-periodic mechanism around the twisted Kagome lattice $L_{\theta}$ in Figure \ref{fig:one-periodic-mechanism}(c). By Proposition \ref{degeneracy}, the optimal $\varphi^\theta_I(x)$ for the minimization problem can be chosen as $\dot{\varphi}(x)$ as the infinitesimal version of the one-periodic mechanism in \eqref{mechanism_around_twisted_kagome}. Notice that we are finding the optimal solution $\varphi_I^\theta(x)$ for the identity matrix $I$, and the one-periodic mechanism $u_{\theta \shortto \theta-t}(x)$ in \eqref{mechanism_around_twisted_kagome} gives an optimal $\dot{\varphi}(x)$ associated to $\xi = c_\theta I$ in \eqref{isotropic_degeneracy} instead of $I$. To get the optimal $\varphi_I^\theta(x)$ associated to the identity matrix $I$, we need to scale the one-periodic mechanism $u_{\theta \shortto \theta-t}(x)$ to $u_{\theta \shortto \theta-\frac{1}{c_{\theta}}t}(x)$. We know this scaled one-periodic mechanism yields the optimal $\varphi_I^\theta(x)$ for the identity matrix. If we denote $\dot{\varphi}(x)$ as the associated optimal solution for $c_{\theta} I$ in Equation \eqref{isotropic_degeneracy}, then $\varphi_I^\theta(x) = \frac{1}{c_{\theta}} \dot{\varphi}(x)$. The scalar $c_{\theta} = \frac{\sin(\theta - \frac{\pi}{3})}{\cos(\frac{\pi}{3} - \theta)} \rightarrow 0$ as $\theta \rightarrow \frac{\pi}{3}$, so $\varphi_I^\theta(x)$ grows unbounded as $\theta \rightarrow \frac{\pi}{3}$ since $\frac{1}{c_\theta} \rightarrow \infty$.
	
	\section{Two-periodic mechanisms and GH modes of the Kagome lattices}\label{sec:two-periodic-mechanism}
	In this section, we present the analytic form of a \textit{three-parameter} two-periodic mechanism of the standard Kagome lattice, shown in Figure \ref{fig:two-periodic-mechanism}. We refer to this three-parameter mechanism from the standard Kagome lattice to a deformed two-periodic Kagome lattice by $u_{\theta_1, \theta_2, \theta_3}(x)$, where $x$ are vertices of the standard Kagome lattice. The three-parameter two-periodic mechanism provides a three-dimensional space of two-periodic GH modes. We will discuss the relation between two-periodic mechanisms and GH modes in subsection \ref{subsection:two-periodic-GH}.
	
	\begin{figure}[!htb]
		\centering
		\includegraphics[width=0.8\linewidth]{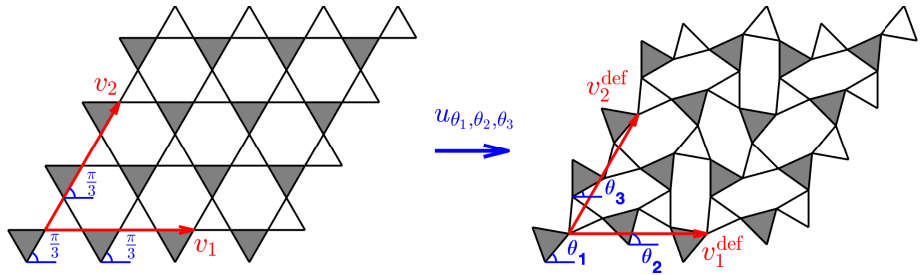}
		\caption{The two-periodic mechanism $u_{\theta_1, \theta_2, \theta_3}(x)$ around the standard Kagome lattice with $\theta_1 = \frac{5\pi}{18}, \theta_2 = \frac{5\pi}{12}$ and $\theta_3 = \frac{\pi}{6}$. The standard Kagome lattice on the left is the reference lattice and the two arrows $\bm{v}_1, \bm{v}_2$ are its primitive vectors. The two-periodic Kagome lattice on the right is the deformed lattice and the two red arrows $\bm{v}_1^{\text{def}}, \bm{v}_2^{\text{def}}$ are its primitive vectors. The macroscopic deformation gradient of $u_{\theta_1, \theta_2, \theta_3}(x)$ maps $\bm{v}_1, \bm{v}_2$ to $\bm{v}_1^{\text{def}}, \bm{v}_2^{\text{def}}$.}
		\label{fig:two-periodic-mechanism}
	\end{figure}
	
	Before we present the detailed construction of the two-periodic mechanism $u_{\theta_1, \theta_2, \theta_3}(x)$, let us first take a look at some geometric properties of $u_{\theta_1, \theta_2, \theta_3}(x)$: (1) this two-periodic mechanism also achieves an isotropic compression; and (2) all regular hexagons in the reference lattice are deformed to a special type of hexagon with three pairs of parallel edges. To explain the origin of these properties, we need to introduce some details about the two-periodic Kagome lattice deformed by the two-periodic mechanism. For simplicity, we call the deformed states of this two-periodic mechanism $L_{\theta_1, \theta_2, \theta_3}$ and fix the length of each equilateral triangle as 1. The unit cell of $L_{\theta_1, \theta_2, \theta_3}$ has 8 triangles, classified into 4 shaded triangles and 4 unshaded triangles in Figure \ref{fig:two-periodic-unit-cell}(a). The three degrees of freedom $\theta_1, \theta_2, \theta_3$ are the rotation angles for the three shaded triangles in Figure \ref{fig:two-periodic-unit-cell}(c). To achieve a two-periodic mechanism, the other five triangles in the unit cell have to rotate correspondingly as shown in Figure \ref{fig:two-periodic-unit-cell}(c), where $\theta_4$ is a function of $\theta_1, \theta_2, \theta_3$
	\begin{align}\label{eqn:theta_4}
		\theta_4 = \frac{\pi}{3} - \arcsin\left(\sin(\theta_1 - \frac{\pi}{3}) + \sin(\theta_2 - \frac{\pi}{3}) + \sin(\theta_3 - \frac{\pi}{3})\right). 
	\end{align}
	We will discuss the angle relations (shown in Figure \ref{fig:two-periodic-unit-cell}c) in section \ref{subsec:4-1}, and explain the origin of \eqref{eqn:theta_4} shortly.
	\begin{figure}[!htb]
		\begin{minipage}[b]{.33\linewidth}
			\centering
			\subfloat[]{\includegraphics[width=1.02\linewidth]{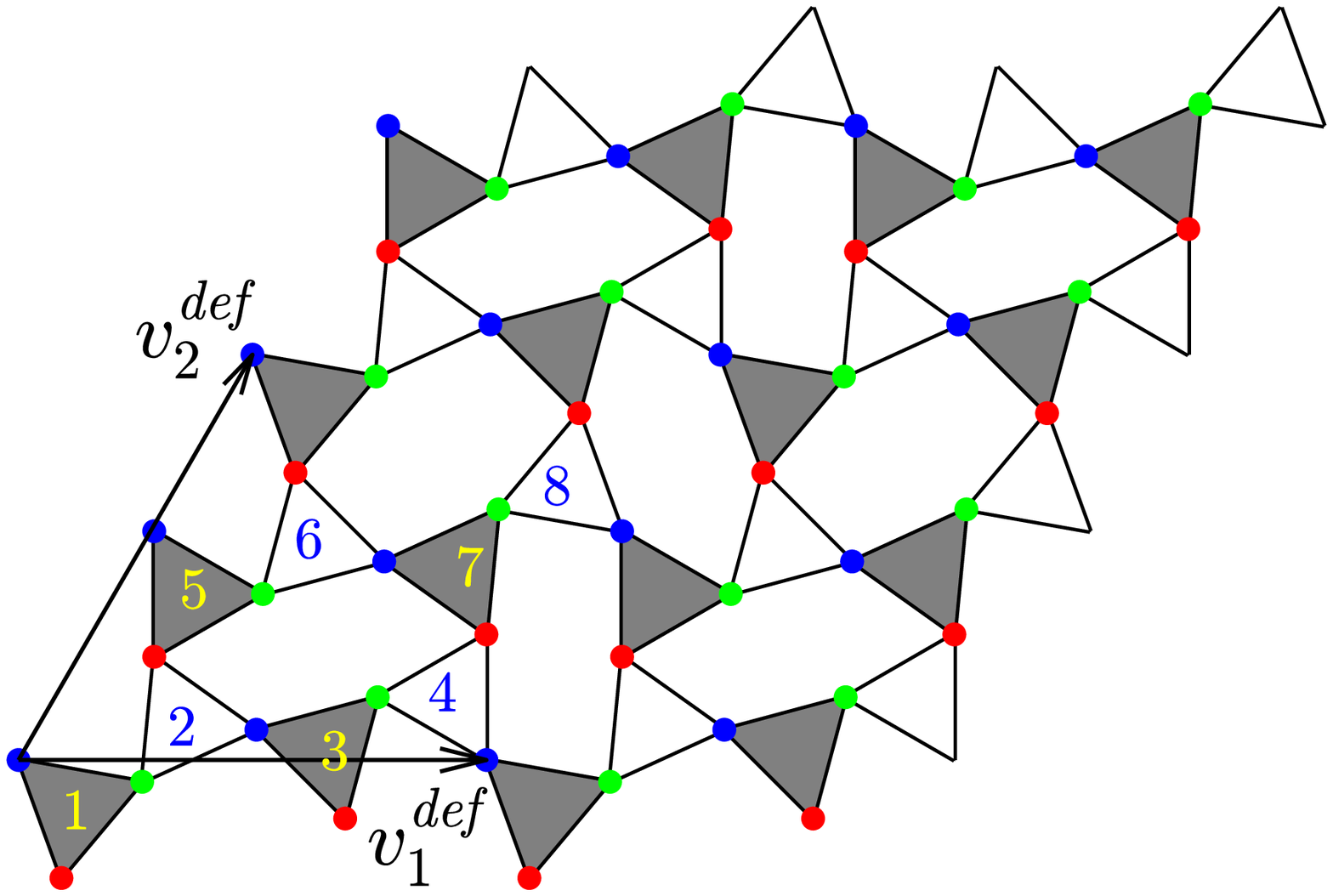}}
		\end{minipage}
		\begin{minipage}[b]{.33\linewidth}
			\subfloat[]{\includegraphics[width=\linewidth]{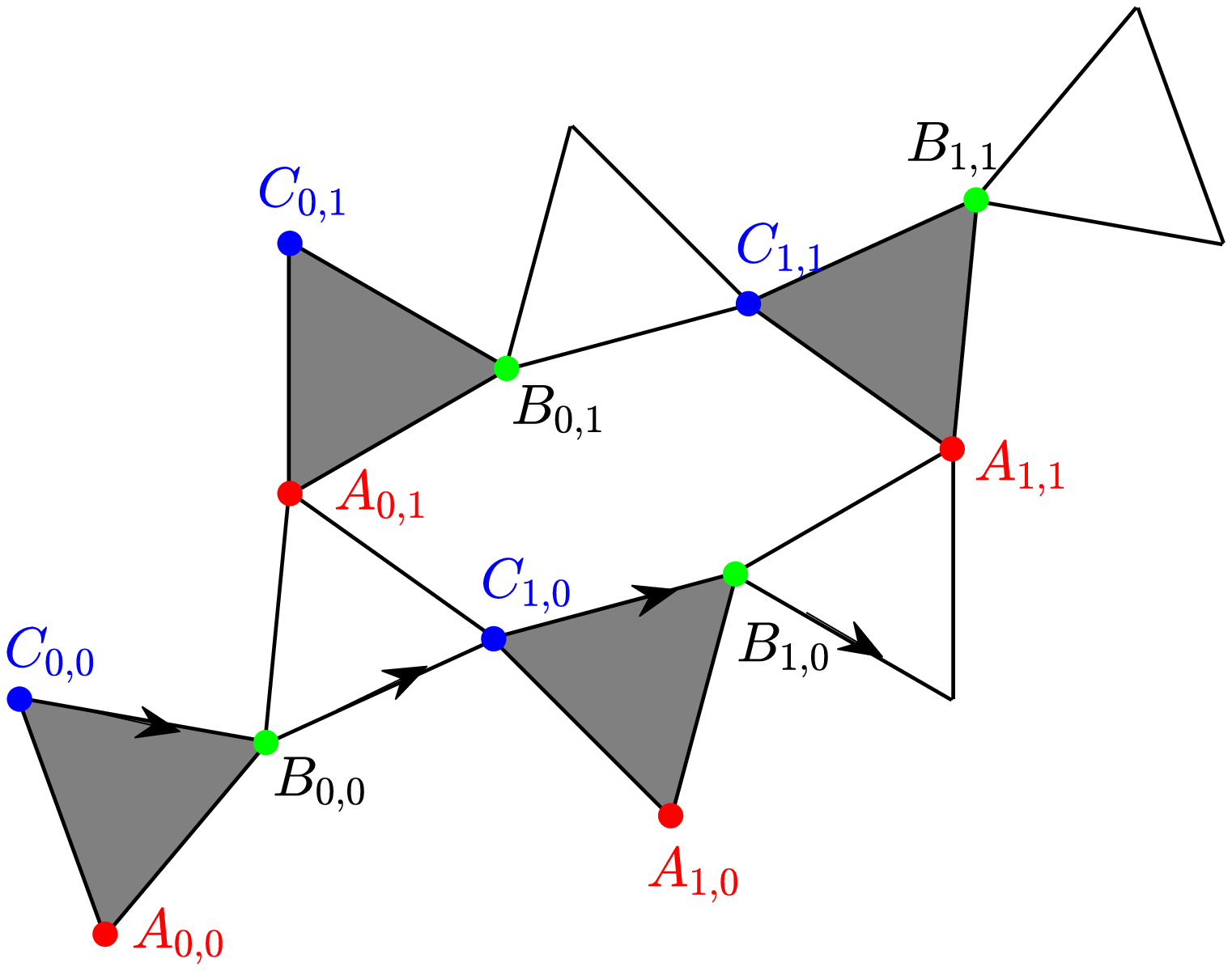}}
		\end{minipage}
		\begin{minipage}[b]{.33\linewidth}
			\centering
			\subfloat[]{\includegraphics[width=1.03\linewidth]{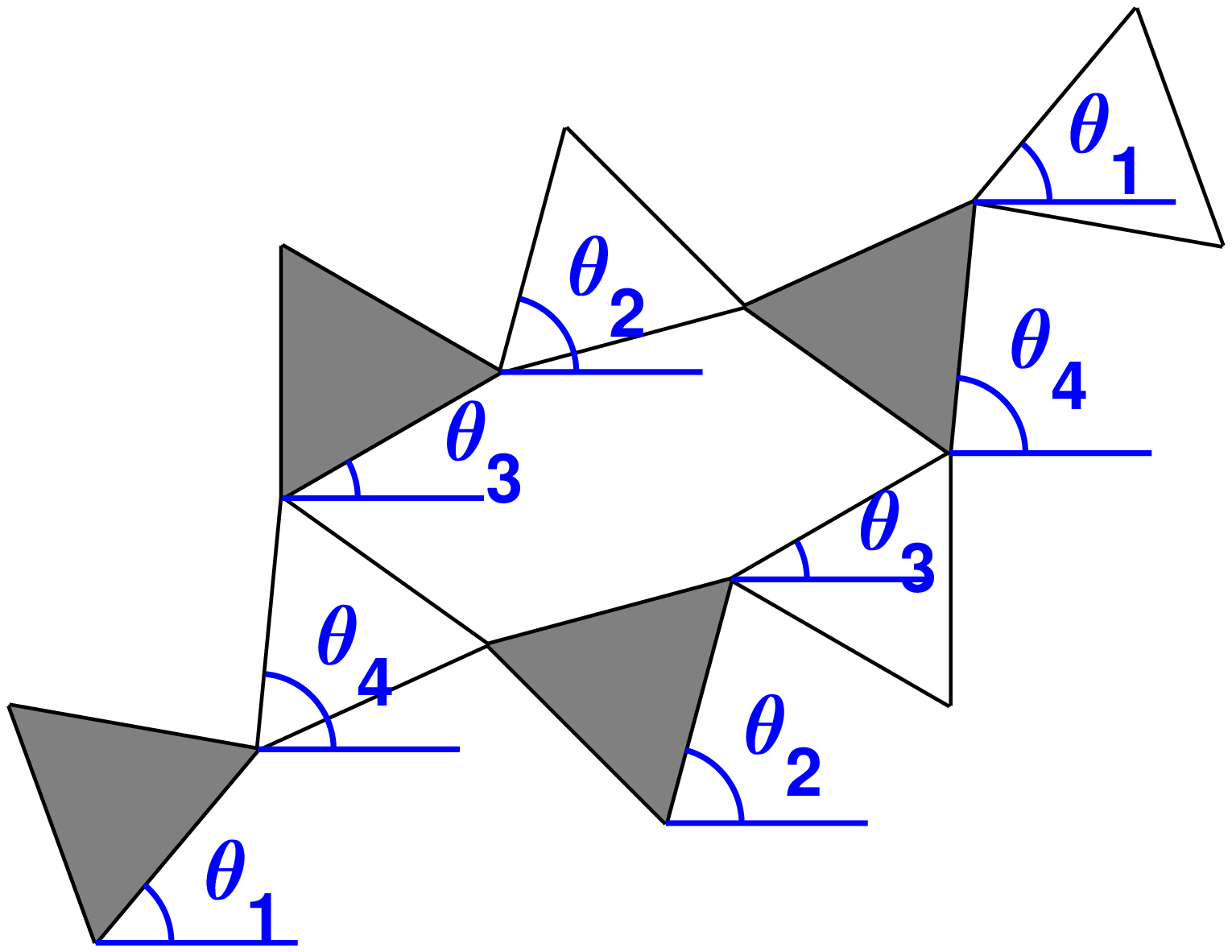}}
		\end{minipage}
		\caption{(a) The deformed two-periodic Kagome lattice $L_{\theta_1, \theta_2, \theta_3}$: its unit cell has 8 triangles, 4 shaded and 4 unshaded; (b) 12 vertices in the unit cell: each shaded triangle has one set of vertex $A,B,C$; (c) the corresponding rotation angles for the 8 triangles in the unit cell to achieve the three-parameter mechanism $u_{\theta_1, \theta_2, \theta_3}(x)$;.}
		\label{fig:two-periodic-unit-cell}
	\end{figure}
	
	In fact, once we know the rotation angles for all triangles in the unit cell, the structure of $L_{\theta_1, \theta_2, \theta_3}$ is fixed, i.e. the locations of every vertex in the unit cell and the two primitive vectors $\bm{v}_1^{\text{def}}, \bm{v}_2^\text{def}$ can be computed. For example, if we mark the vertices in the unit cell as shown in Figure \ref{fig:two-periodic-unit-cell}(b) and fix $A_{0,0}$ as the origin, then the locations of all vertices in the unit cell can be represented by rotation angles, e.g. $B_{0,0} = (\cos \theta_1, \sin \theta_1)$ and $C_{1,0} = B_{0,0} + \left(\cos(\theta_4-\frac{\pi}{3}), \sin(\theta_4-\frac{\pi}{3})\right)$ (see Appendix \ref{appendix-c} for a detailed expression of each vertex's location). We can also add up the four vectors connecting $B, C$ vertices in Figure \ref{fig:two-periodic-unit-cell}(c) to get $\bm{v}_1^\text{def}$
	\begin{align}
		\bm{v}_1^{\text{def}} &=\Big(\cos(\theta_1 - \frac{\pi}{3}) + \cos(\theta_2 - \frac{\pi}{3}) + \cos(\theta_3 - \frac{\pi}{3}) + \cos(\theta_4 - \frac{\pi}{3}), \nonumber \\ 
		& \qquad \sin(\theta_1 - \frac{\pi}{3}) + \sin(\theta_2 - \frac{\pi}{3}) + \sin(\theta_3 - \frac{\pi}{3}) + \sin(\theta_4 - \frac{\pi}{3}) \Big) \label{eqn:v_1} \\ 
		&= \Big(\cos(\theta_1 - \frac{\pi}{3}) + \cos(\theta_2 - \frac{\pi}{3}) + \cos(\theta_3 - \frac{\pi}{3}) + \cos(\theta_4 - \frac{\pi}{3}), 0\Big) \nonumber.
	\end{align}
	The vertical part of $\bm{v}_1^{\text{def}}$ vanishes because we deliberately choose $\theta_4$ in \eqref{eqn:theta_4} to fix $\bm{v}_1^{\text{def}}$ in the horizontal direction. Similarly, we can add up the four vectors connecting $A, B$ vertices to get $\bm{v}_2^\text{def}$
	\begin{align}
		\bm{v}_2^{\text{def}} &= \left(\cos\theta_1 + \cos\theta_2 + \cos\theta_3  + \cos\theta_4, \sin\theta_1 + \sin\theta_2  + \sin\theta_3  + \sin\theta_4\right). \label{eqn:v_2}
	\end{align}
	Evidently, from \eqref{eqn:v_1} and \eqref{eqn:v_2}, we observe that $\bm{v}_2^{\text{def}} = R_{\frac{\pi}{3}} \bm{v}_1^{\text{def}}$, where $R_{\frac{\pi}{3}}$ is the rotation matrix that rotates counterclockwise with angle $\frac{\pi}{3}$. The macroscopic deformation gradient $F_{\theta_1, \theta_2, \theta_3}$ for the two-periodic mechanism maps the  primitive vectors $\bm{v}_1, \bm{v}_2$ of the standard Kagome lattice in Figure \ref{fig:two-periodic-mechanism} to the two primitive vectors $\bm{v}_2^{\text{def}}, \bm{v}_2^{\text{def}}$ of the two-periodic Kagome lattice $L_{\theta_1, \theta_2,  \theta_3}$
	\begin{align*}
		&\bm{v}_1 = (4,0) &\rightarrow & &\bm{v}_1^{\text{def}} = F_{\theta_1, \theta_2, \theta_3} \bm{v}_1\\
		&\bm{v}_2 = (2,2\sqrt{3}) &\rightarrow & &\bm{v}_2^{\text{def}} = F_{\theta_1, \theta_2, \theta_3} \bm{v}_2.
	\end{align*}
	This yields that 
	\begin{align}
		F_{\theta_1, \theta_2, \theta_3} &=  c_{\theta_1, \theta_2, \theta_3} \: I ,\qquad c_{\theta_1, \theta_2, \theta_3} =  \frac{1}{4}\left(\cos(\theta_1-\frac{\pi}{3}) + \cos(\theta_2-\frac{\pi}{3}) + \cos(\theta_3-\frac{\pi}{3})  + \cos(\theta_4-\frac{\pi}{3})\right), \label{eqn:two_periodic_macroscopic}
	\end{align}
	where $\theta_4$ is a function of $\theta_1, \theta_2, \theta_3$ by \eqref{eqn:theta_4}. As expected, the two-periodic mechanism $u_{\theta_1, \theta_2, \theta_3}(x)$ achieves a macroscopic isotropic compression for every $\theta_1, \theta_2, \theta_3$.
	
	We mentioned earlier the geometric property of the two-periodic mechanism $u_{\theta_1, \theta_2, \theta_3}(x)$ that the deformed hexagon has its three pairs of opposite edges parallel. It can be seen from Figure \ref{fig:hexagons}(a) that the two solid edges are parallel to each other since the angles between the two edges and the horizontal direction are both $\theta_3$. Similarly, the other two pairs of edges are parallel to each other. We note that the deformed hexagon of the one-periodic mechanism $u_{\frac{\pi}{3} \shortto \theta}(x)$ does not have this property, e.g. see Figure \ref{fig:hexagons}(b). We thus see that the one-periodic mechanism $u_{\frac{\pi}{3} \shortto \theta}(x)$ cannot be derived from this two-periodic mechanism $u_{\theta_1, \theta_2, \theta_3}(x)$ by writing $\theta_1, \theta_2, \theta_3$ as functions of a single parameter $\theta$.
	\begin{figure}[!htb]
		\begin{minipage}{0.48\textwidth}
			\centering
			\subfloat[]{\includegraphics[width=0.75\linewidth]{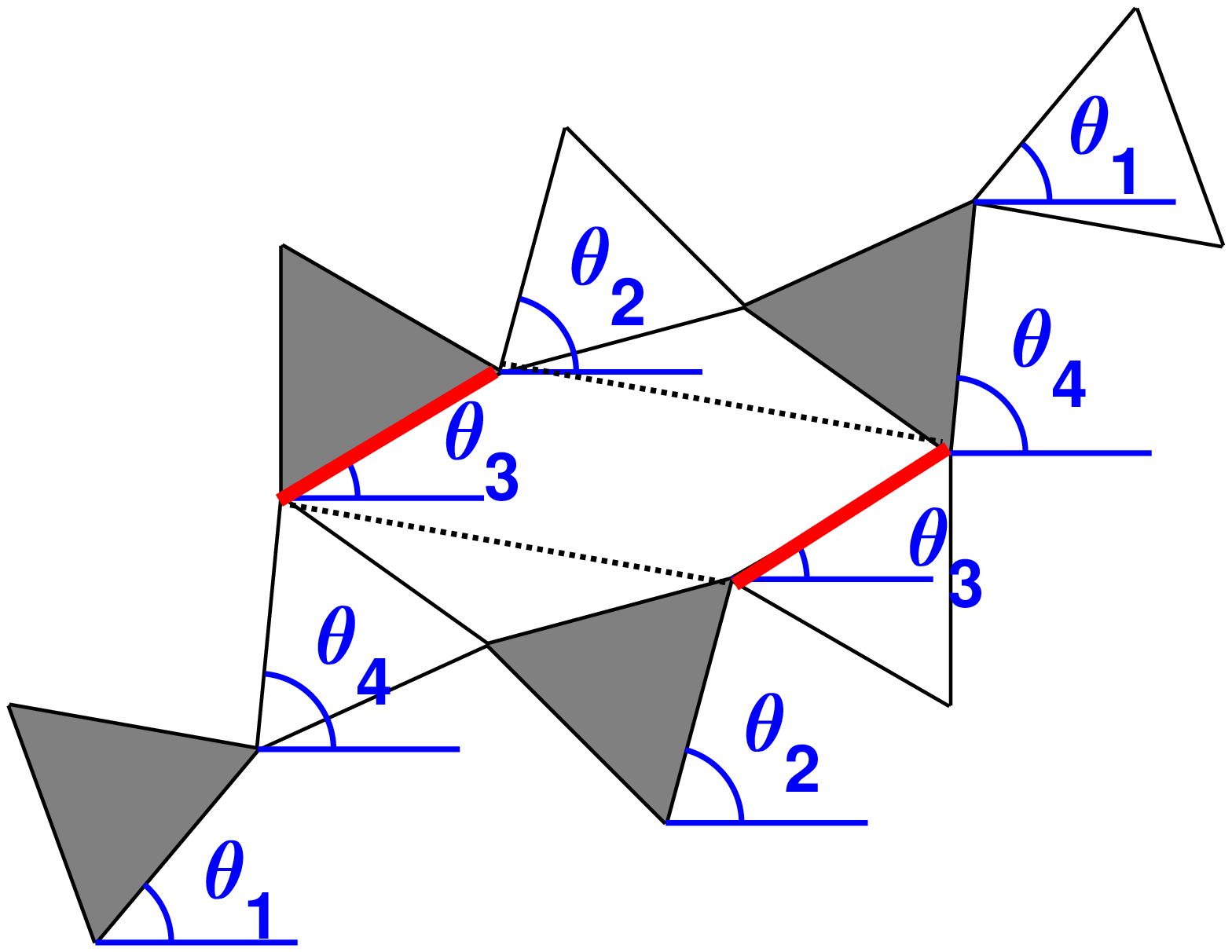}}
		\end{minipage}
		\begin{minipage}{0.48\textwidth}
			\centering
			\subfloat[]{\includegraphics[width=0.75\linewidth]{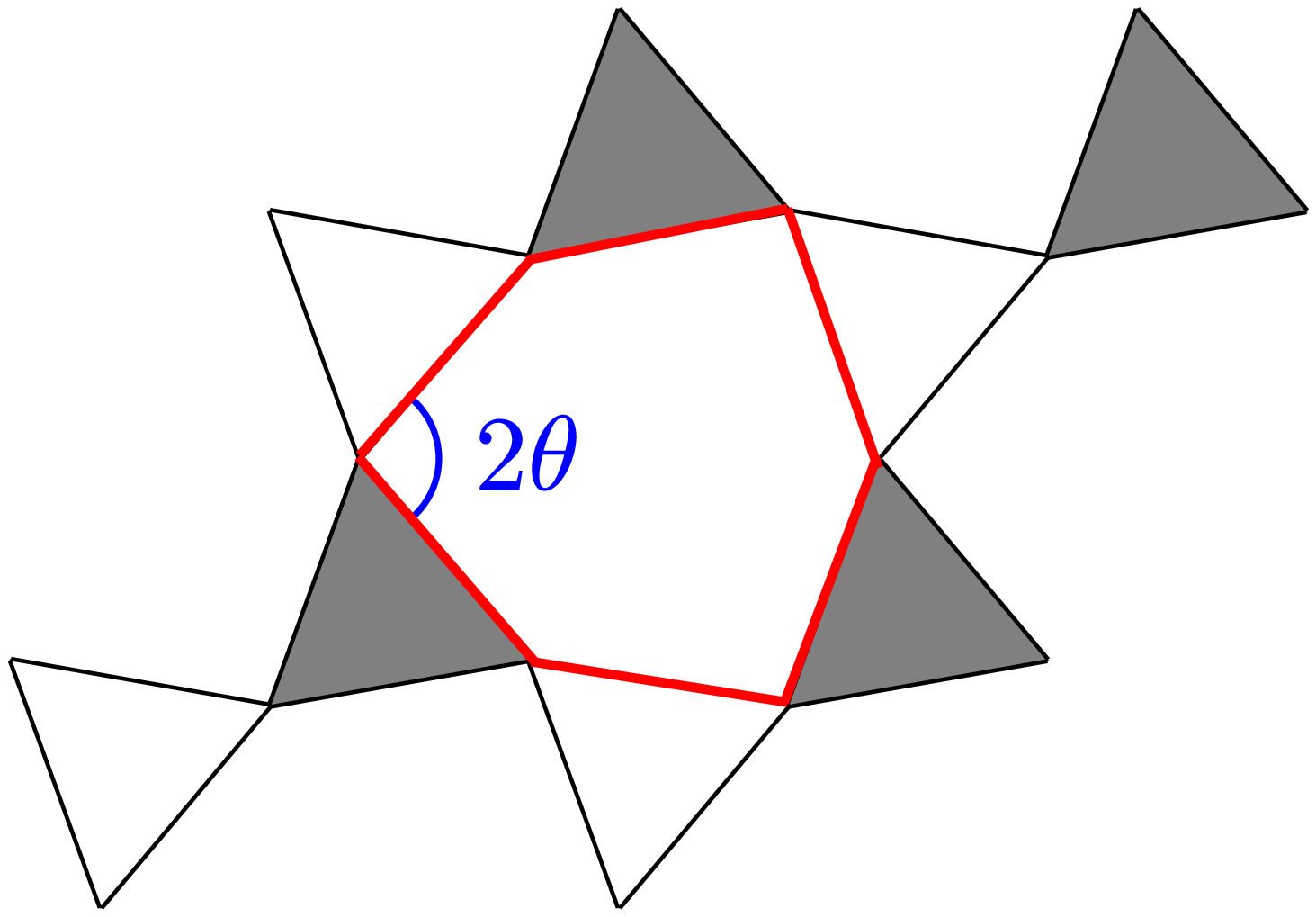}}
		\end{minipage}
		\caption{(a) One of the deformed hexagon under the  two-periodic mechanism $u_{\theta_1,\theta_2,\theta_3}(x)$: we can see clearly that three edges are parallel to each other; (b) the deformed hexagon under  the one-periodic mechanism $u_{\frac{\pi}{3}  \shortto \theta}(x)$.}
		\label{fig:hexagons}
	\end{figure}
	
	\subsection{Construction of the three-parameter two-periodic mechanism}\label{subsec:4-1}
	Let us discuss the details of the two-periodic mechanism $u_{\theta_1, \theta_2, \theta_3}(x)$ and why the rotation angles are related in the way shown in Figure \ref{fig:two-periodic-unit-cell}(c). A natural way to look for a two-periodic mechanism is to assign the 8 triangles in the unit cell with 8 different angles. As shown in \ref{fig:two-periodic-construction}(a), we name the rotation angles for the four shaded triangles as $\theta_1, \theta_2, \theta_3, \theta_4$ and the rotation angles for the four unshaded triangles as $\eta_1, \eta_2, \eta_3, \eta_4$. The direction of every edge is determined since we know how each triangle rotates. Take the pair of triangles in Figure \ref{fig:two-periodic-construction}(b) as an example. All six vectors are determined by the two angles $\theta_1, \eta_2$
	\begin{align}
		\bm{t}_1 &= \left(\cos \theta_1, \sin \theta_1\right)& \bm{t}_2 &= \left(\cos(\theta_1 + \frac{2\pi}{3}), \sin(\theta_1 + \frac{2\pi}{3})\right)& \bm{t}_3 &= \left(\cos(\theta_1 + \frac{4\pi}{3}), \sin(\theta_1 + \frac{4\pi}{3})\right) \nonumber\\
		\bm{t}_4 &= \left(\cos(\eta_1 - \frac{\pi}{3}), \sin(\eta_1 - \frac{\pi}{3})\right)& \bm{t}_5 &= \left(\cos(\eta_1 + \frac{\pi}{3}), \sin(\eta_1 + \frac{\pi}{3})\right)&  \bm{t}_6 &= \left(-\cos\eta_1, -\sin \eta_1\right).\label{eqn:vector-explicit-form}
	\end{align}
	
	\begin{figure}[!htb]
		\begin{minipage}[b]{.48\linewidth}
			\centering
			\subfloat[]{\includegraphics[width=0.7\linewidth]{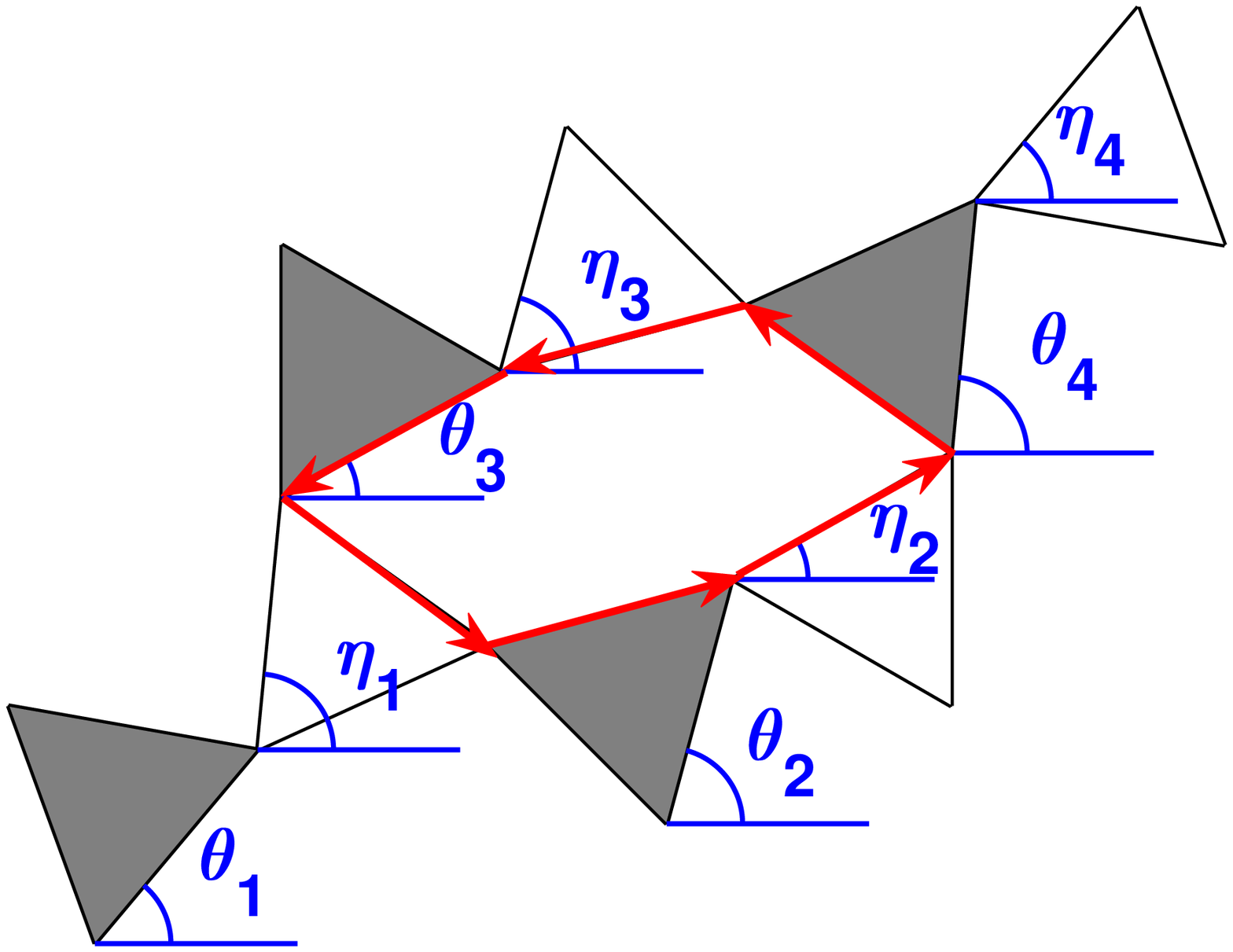}}
		\end{minipage}
		\begin{minipage}[b]{.48\linewidth}
			\centering
			\subfloat[]{\includegraphics[width=0.5\linewidth]{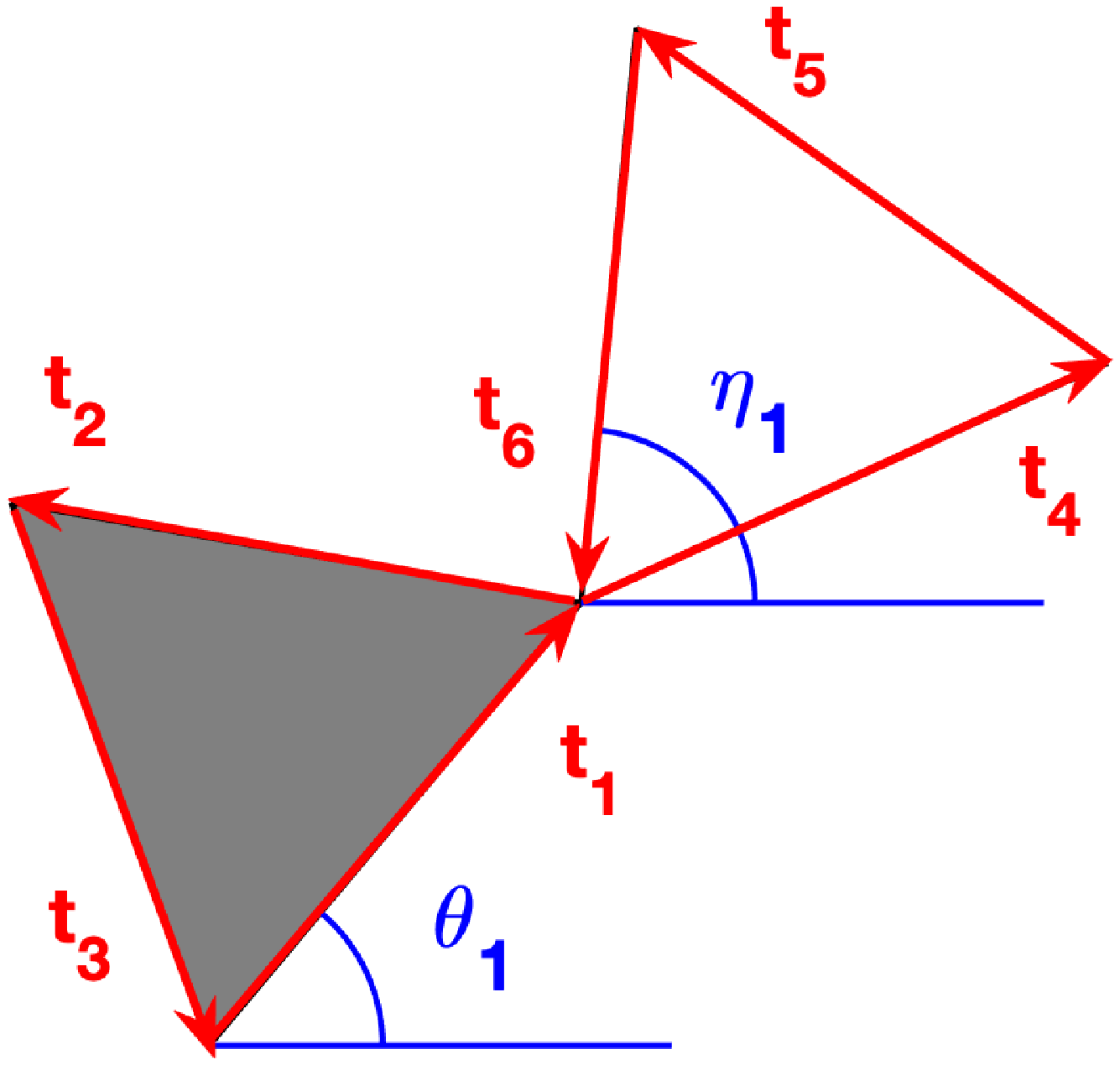}}
		\end{minipage}\par\medskip
		\begin{minipage}[b]{.48\linewidth}
			\subfloat[]{\includegraphics[width=0.7\linewidth]{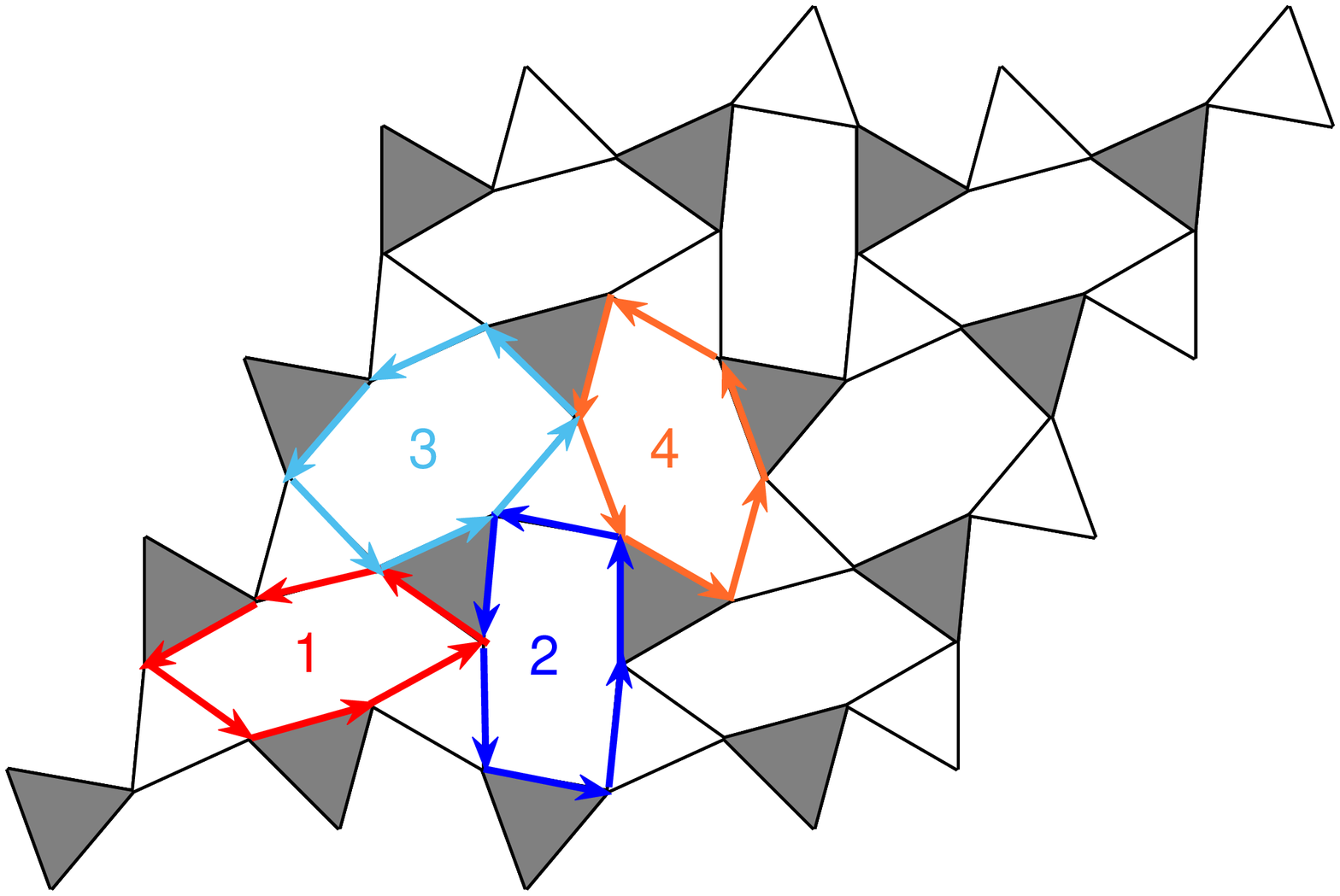}}
		\end{minipage}
		\begin{minipage}[b]{.48\linewidth}
			\subfloat[]{\includegraphics[width=0.7\linewidth]{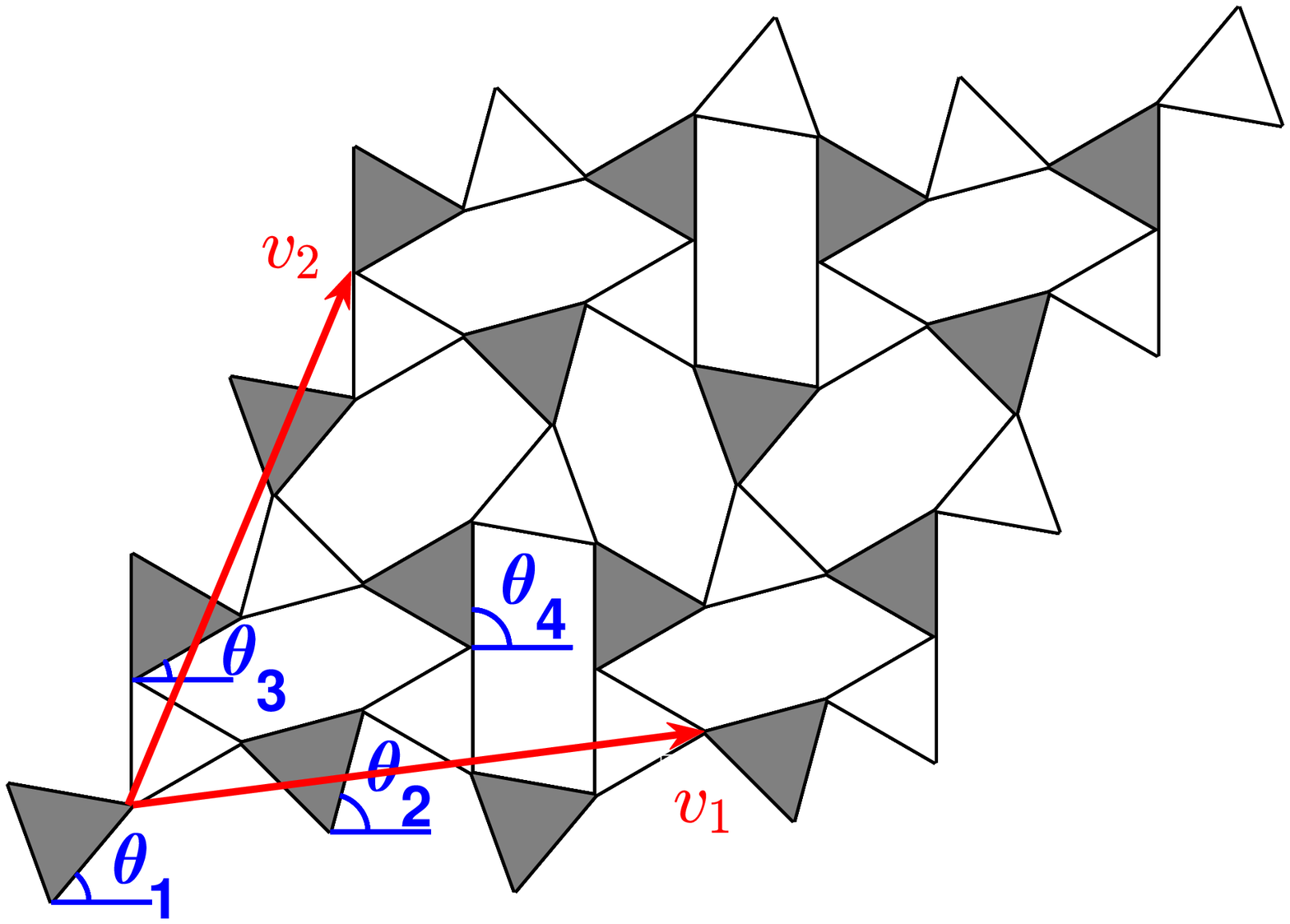}}
		\end{minipage}
		\caption{(a) The unit cell of a two-periodic Kagome lattice after the deformation $u_{\theta_1, \theta_2, \theta_3}(x)$; (b) a pair of triangles whose edge directions are determined by their rotation angles; (c) four deformed hexagons must close themselves to create a two-periodic unit cell; (d) a rotated version of the two-periodic Kagome lattice when $\theta_4$ does not follow \eqref{eqn:theta_4}.}
		\label{fig:two-periodic-construction}
	\end{figure}
	The 8 angles $\theta_1, \theta_2, \theta_3, \theta_4$ and $\eta_1, \eta_2, \eta_3, \eta_4$ must satisfy some constraints. In fact, the images of the 8 triangles in the unit cell under the two-periodic mechanism must form a lattice, i.e. there can be no gaps between their images. This requires that the sum of the six vectors in Figure \ref{fig:two-periodic-construction}(a) must vanish. Using the vector form of every edge in Equation \eqref{eqn:vector-explicit-form}, we get two constraints
	\begin{align}
		0 &= -\cos(\eta_1+\frac{\pi}{3}) - \cos(\theta_2+\frac{2\pi}{3}) + \cos \eta_2 + \cos(\theta_4 + \frac{\pi}{3}) - \cos(\eta_3 - \frac{\pi}{3}) - \cos \theta_3, \label{eqn:hexagon1_cos}\\
		0 &= -\sin(\eta_1+\frac{\pi}{3}) - \sin(\theta_2+\frac{2\pi}{3}) + \sin \eta_2 + \sin(\theta_4 + \frac{\pi}{3}) - \sin(\eta_3 - \frac{\pi}{3}) - \sin \theta_3.\label{eqn:hexagon1_sin}
	\end{align}
	These constraints assure that the image of a particular hexagon in  the reference lattice  is again a (deformed) hexagon. In other words, there is no gap when we connect the neighboring triangles around hexagon 1 in Figure \ref{fig:two-periodic-construction}(c). Similarly, the other three hexagons in Figure \ref{fig:two-periodic-construction}(c) give another six constraints
	\begin{align}
		0 &= -\cos(\eta_2+\frac{\pi}{3}) - \cos(\theta_1+\frac{2\pi}{3}) + \cos \eta_1 + \cos(\theta_3 + \frac{\pi}{3}) - \cos(\eta_4 - \frac{\pi}{3}) - \cos \theta_4, \label{eqn:hexagon2_cos}\\
		0 &= -\sin(\eta_2+\frac{\pi}{3}) - \sin(\theta_1+\frac{2\pi}{3}) + \sin \eta_1 + \sin(\theta_3 + \frac{\pi}{3}) - \sin(\eta_4 - \frac{\pi}{3}) - \sin \theta_4, \label{eqn:hexagon2_sin}\\
		0 &= -\cos(\eta_3+\frac{\pi}{3}) - \cos(\theta_4+\frac{2\pi}{3}) + \cos \eta_4 + \cos(\theta_2 + \frac{\pi}{3}) - \cos(\eta_1 - \frac{\pi}{3}) - \cos \theta_1, \label{eqn:hexagon3_cos}\\
		0 &= -\sin(\eta_3+\frac{\pi}{3}) - \sin(\theta_4+\frac{2\pi}{3}) + \sin \eta_4 + \sin(\theta_2 + \frac{\pi}{3}) - \sin(\eta_1 - \frac{\pi}{3}) - \sin \theta_1, \label{eqn:hexagon3_sin}\\
		0 &= -\cos(\eta_4+\frac{\pi}{3}) - \cos(\theta_3+\frac{2\pi}{3}) + \cos \eta_3 + \cos(\theta_1 + \frac{\pi}{3}) - \cos(\eta_2 - \frac{\pi}{3}) - \cos \theta_2,  \label{eqn:hexagon4_cos}\\
		0 &= -\sin(\eta_4+\frac{\pi}{3}) - \sin(\theta_3+\frac{2\pi}{3}) + \sin \eta_3 + \sin(\theta_1 + \frac{\pi}{3}) - \sin(\eta_2 - \frac{\pi}{3}) - \sin \theta_2 \label{eqn:hexagon4_sin}.
	\end{align}
	Once the 8 angles meet the above 8 constraints, all deformed hexagons are closed and we have no problem connecting all neighboring triangles. Thus,  \eqref{eqn:hexagon1_cos}-\eqref{eqn:hexagon4_sin} are the only constraints on $\theta_1, \theta_2, \theta_3,\theta_4$ and $\eta_1, \eta_2, \eta_3,  \eta_4$.
	
	Finding all solutions to the 8 nonlinear constraints is hard. However, they are obviously satisfied if we take
	\begin{align}
		\theta_1 = \eta_4 \quad \theta_2 = \eta_3 \quad \theta_3 = \eta_2 \quad \theta_4 = \eta_1. \label{eqn:angle-relation}
	\end{align}
	This special solution gives a two-periodic mechanism with four free angles $\theta_1, \theta_2, \theta_3, \theta_4$. When we freely rotate these four angles, there is in fact an overall rotation, e.g. see Figure \ref{fig:two-periodic-construction}(d), where the primitive vector $\bm{v}_1$ is not in the horizontal direction. To get rid of the macroscopic rotation, we can choose $\theta_4$ to keep the primitive vector $\bm{v}_1$ in the horizontal direction. This is indeed the constraint given by \ref{eqn:theta_4}. We have now fully explained the  three-parameter mechanism $u_{\theta_1, \theta_2, \theta_3}(x)$.
	
	\subsection{Remarks on the preceding calculation}
	It might seem surprisingly that the 8 angles $\theta_1, \dots, \theta_4$ and $\eta_1, \dots, \eta_4$ are subject to 8 constraints \eqref{eqn:hexagon1_cos}-\eqref{eqn:hexagon4_sin}, and yet we found a four-parameter family of solutions \eqref{eqn:angle-relation}. Actually, this is not the only surprise:
	\begin{itemize}
		\item The two-periodic extension of the one-periodic mechanism is another solution to \eqref{eqn:hexagon1_cos}-\eqref{eqn:hexagon4_sin}. It is not included in the four-parameter family of solutions \eqref{eqn:angle-relation}; rather, it corresponds to 
		\begin{align*}
			\theta_1 = \theta_2 = \theta_3 = \theta_4 = \theta, \qquad \eta_1 = \eta_2 = \eta_3 = \eta_4 = \eta.
		\end{align*}
		When $\theta \neq \frac{\pi}{3} - \eta$, there is a macroscopic rotation. In fact, if we choose $\theta, \eta$ freely, then this mechanism has a macroscopic rotation and associates to Figure \ref{fig:one-periodic-mechanism}(b) with $\theta = \alpha$ and $\eta = \frac{\pi}{3} - \beta$.
		
		\item Counting equations to estimate the number of free parameters is unreliable in this setting because the family of energy-free configurations is not a smooth manifold, as we will discuss later (see Remark \ref{rmk:singularity}).
		
		\item The 8 equations \eqref{eqn:hexagon1_cos}-\eqref{eqn:hexagon4_sin} are redundant -- they can be easily reduced to 6 nonlinear constraints on the 8 angles $\theta_1, \dots, \theta_4, \eta_1, \dots, \eta_4$.
	\end{itemize}
	This section dwells on the last bullet, offering an algebraic explanation first, then a geometric one. This material is not needed in the rest of the paper, so an impatient reader can safely skip to section \ref{subsection:two-periodic-GH}.
	
	Let us start with the short algebraic explanation first. We observe that for any choice of the 8 angles $\theta_1, \dots, \theta_4, \eta_1, \dots, \eta_4$, the sum of all $\cos$ parts on the right hand side of \eqref{eqn:hexagon1_cos}, \eqref{eqn:hexagon2_cos}, \eqref{eqn:hexagon3_cos} and \eqref{eqn:hexagon4_cos} always vanishes, as well as the the sum of all $\sin$ parts on the right hand side of \eqref{eqn:hexagon1_sin}, \eqref{eqn:hexagon2_sin}, \eqref{eqn:hexagon3_sin} and \eqref{eqn:hexagon4_sin}. This indicates that when hexagon 1,2 and 3 in Figure \ref{fig:two-periodic-construction}(d) are closed, hexagon 4 is automatically closed. Therefore, we only need six constraints \eqref{eqn:hexagon1_cos}-\eqref{eqn:hexagon3_cos} and \eqref{eqn:hexagon1_sin}-\eqref{eqn:hexagon3_sin} to achieve a compatible two-periodic unit cell.
	
	A geometric explanation of this reduction of constraints comes from an observation: closing a pair of hexagons is equivalent to make the average of every zigzag line in one direction the same. Let us use the horizontal direction as an illustrative example. If we close hexagon 1 and 2 in Figure \ref{fig:two-periodic-geometric-reason}(a), then the average of the two zigzag lines in the horizontal direction (marked as dashed lines in Figure \ref{fig:two-periodic-geometric-reason}(b)) must be the same. To see why, we translate the left two vectors in hexagon 1 to the right side of hexagon 2 (see Figure \ref{fig:two-periodic-geometric-reason}(a)). The sum of the 12 vectors in hexagon 1 and 2 becomes the sum of 8 vectors in Figure \ref{fig:two-periodic-geometric-reason}(b), since the related triangles are closed. When hexagon 1 and 2 are both closed, the sum of the 12 vectors in Figure \ref{fig:two-periodic-geometric-reason}(a) and the 8 vectors in Figure \ref{fig:two-periodic-geometric-reason}(b) are both zero. Therefore, the two dashed vectors in Figure \ref{fig:two-periodic-geometric-reason}(b) are of the same magnitude but in opposite direction. This indicates that the average of the two zigzag lines in the horizontal direction are the same.
	
	The preceding argument is also reversible: when we make the average of the zigzag lines in the horizontal direction the same, the sum of all vectors in hexagon 1 and 2 is zero.
	\begin{figure}[!htb]
		\begin{minipage}[b]{0.48\linewidth}
			\centering
			\subfloat[]{\includegraphics[width=0.7\linewidth]{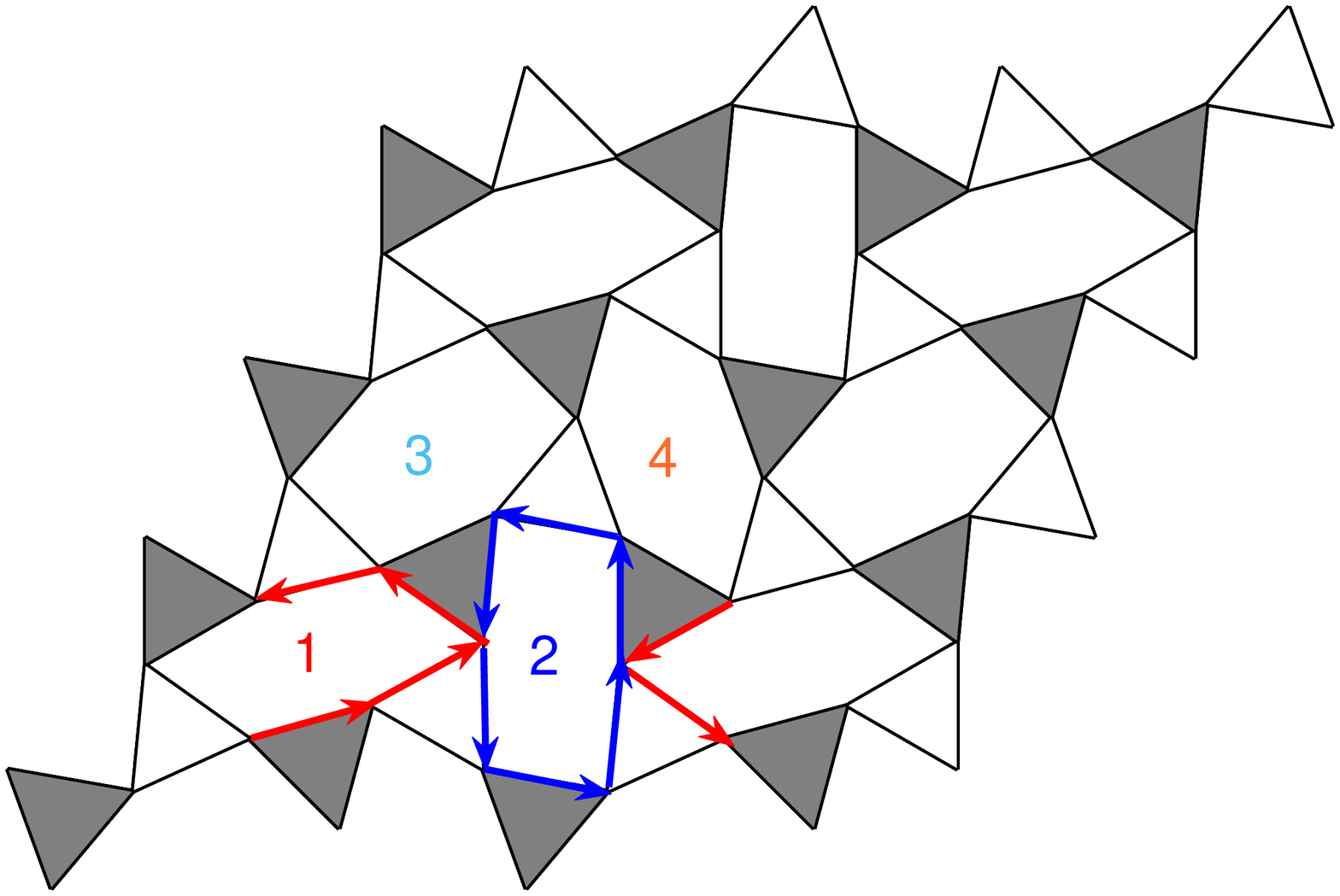}}
		\end{minipage}
		\begin{minipage}[b]{.48\linewidth}
			\centering
			\subfloat[]{\includegraphics[width=0.7\linewidth]{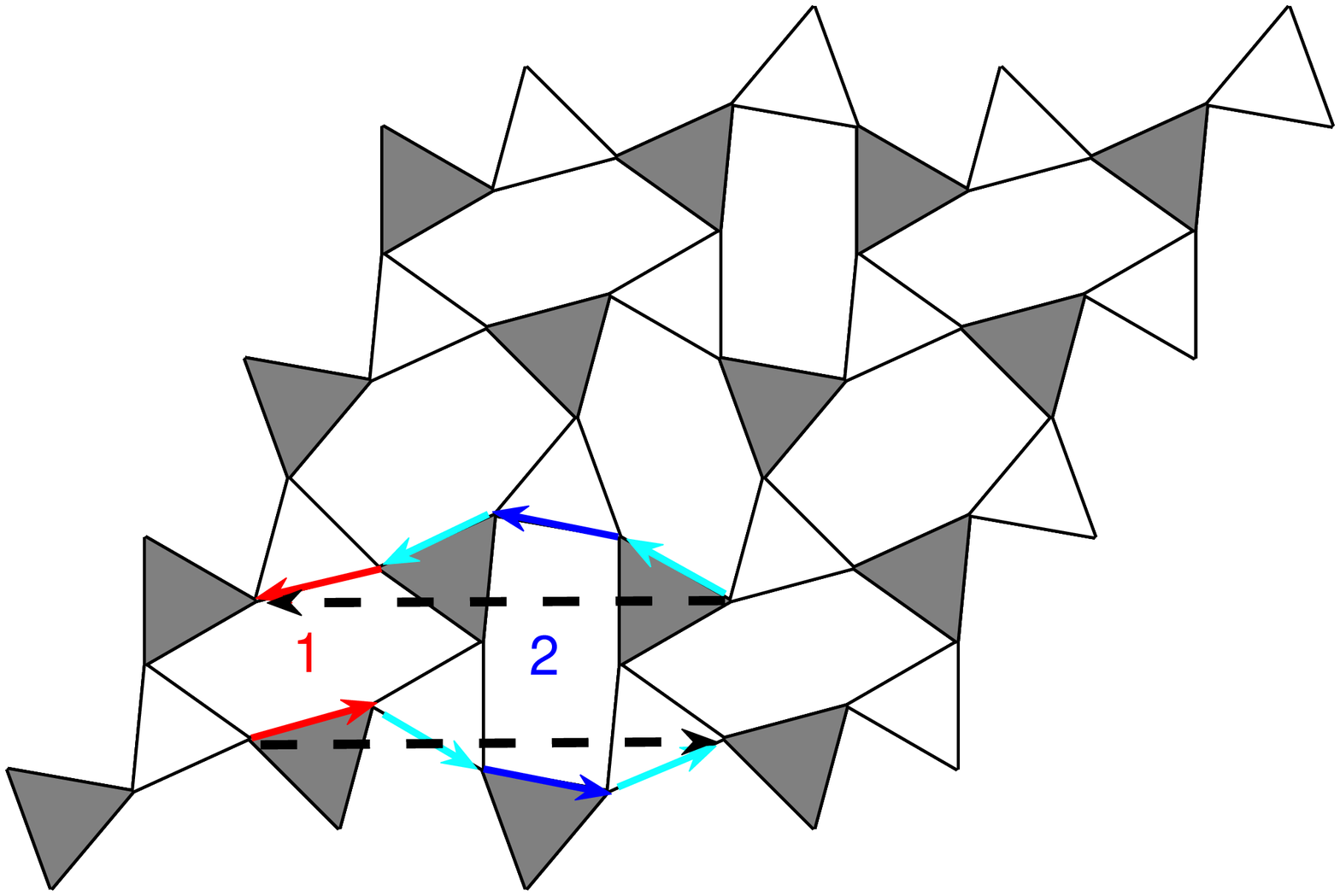}}
		\end{minipage}
		\caption{(a) Vectors in hexagon 1 and 2: we move the left two vectors in hexagon 1 to the right of hexagon 2; (b) the sum of the 12 vectors in (a) becomes the sum of the 8 vectors here: it indicates that the average of the two zigzag lines in the horizontal direction must be the same.}
		\label{fig:two-periodic-geometric-reason}
	\end{figure}
	Similarly, making the average of the zigzag lines in the 60 (120) degree direction is equivalent to making the sum of vectors in hexagon 1 and 3 (2 and 3) zero. The average constraints on every pair of parallel lines in the three lattice directions give 6 constraints. To see why they are equivalent to closing the 4 hexagons, we mark the sum of the 6 vectors in hexagon $i$ as $\vec{s}_i, i \in \{1,2,3,4\}$. The average constraints now become
	\begin{align*}
		\vec{s}_1 + \vec{s}_2 &= \vec{0} & \vec{s}_1 + \vec{s}_3 &= \vec{0} & \vec{s}_2 + \vec{s}_3 &= \vec{0}.
	\end{align*}
	This yields that $\vec{s}_1 = \vec{s}_2 = \vec{s}_3 = \vec{0}$ ($\vec{s}_4$ must also vanish since the average of the horizontal zigzag lines can be represented by both $\vec{s}_1+\vec{s}_2$ and $\vec{s}_3+\vec{s}_4$). Thus, closing the four hexagons is equivalent to the six constraints by taking average in the three lattice directions.
	
	\subsection{Relation between the two-periodic mechanisms and GH modes}\label{subsection:two-periodic-GH}
	We know from section \ref{sec:preliminary} that for the standard Kagome lattice, the space of two-periodic GH modes is four-dimensional. In fact, we can get four explicit GH modes as a basis of the space of two-periodic GH modes from linearizing two-periodic mechanisms. We will show that the three degrees of freedom in the two-periodic mechanism $u_{\theta_1,\theta_2,\theta_3}(x)$ in Figure \ref{fig:two-periodic-mechanism} yield three linearly independent two-periodic GH modes. Moreover, any linear combination of these three GH modes comes from a scaled version of this two-periodic mechanism $u_{\theta_1,\theta_2,\theta_3}(x)$. 
	
	Arguing as we did in section \ref{one-periodic-mechanism} for the one-periodic mechanism, we can get three two-periodic GH modes from the two-periodic mechanism $u_{\theta_1, \theta_2, \theta_3}(x)$, namely the three GH modes $\varphi_1^2(x), \varphi_2^2(x), \varphi_3^2(x)$ defined by
	\begin{align}
		\varphi_1^2(x) &= \frac{d}{d t}\Big|_{t=0} u_{\frac{\pi}{3}-t, \frac{\pi}{3}, \frac{\pi}{3}}, \qquad \varphi_2^2(x) = \frac{d}{d t}\Big|_{t=0} u_{\frac{\pi}{3}, \frac{\pi}{3}-t, \frac{\pi}{3}},\qquad \varphi_3^2(x) = \frac{d}{d t}\Big|_{t=0} u_{\frac{\pi}{3}, \frac{\pi}{3}, \frac{\pi}{3}-t}. \label{eqn:two-periodic-GH-def}
	\end{align}
	These three two-periodic GH modes are shown in Figure \ref{fig:two-periodic-standard-GH-modes} (see Appendix \ref{appendix-c} for their explicit expressions). Moreover, these three GH modes are linearly independent.
	\begin{figure}[!htb]
		\begin{minipage}[b]{.33\linewidth}
			\centering
			\subfloat[]{\includegraphics[width=\linewidth]{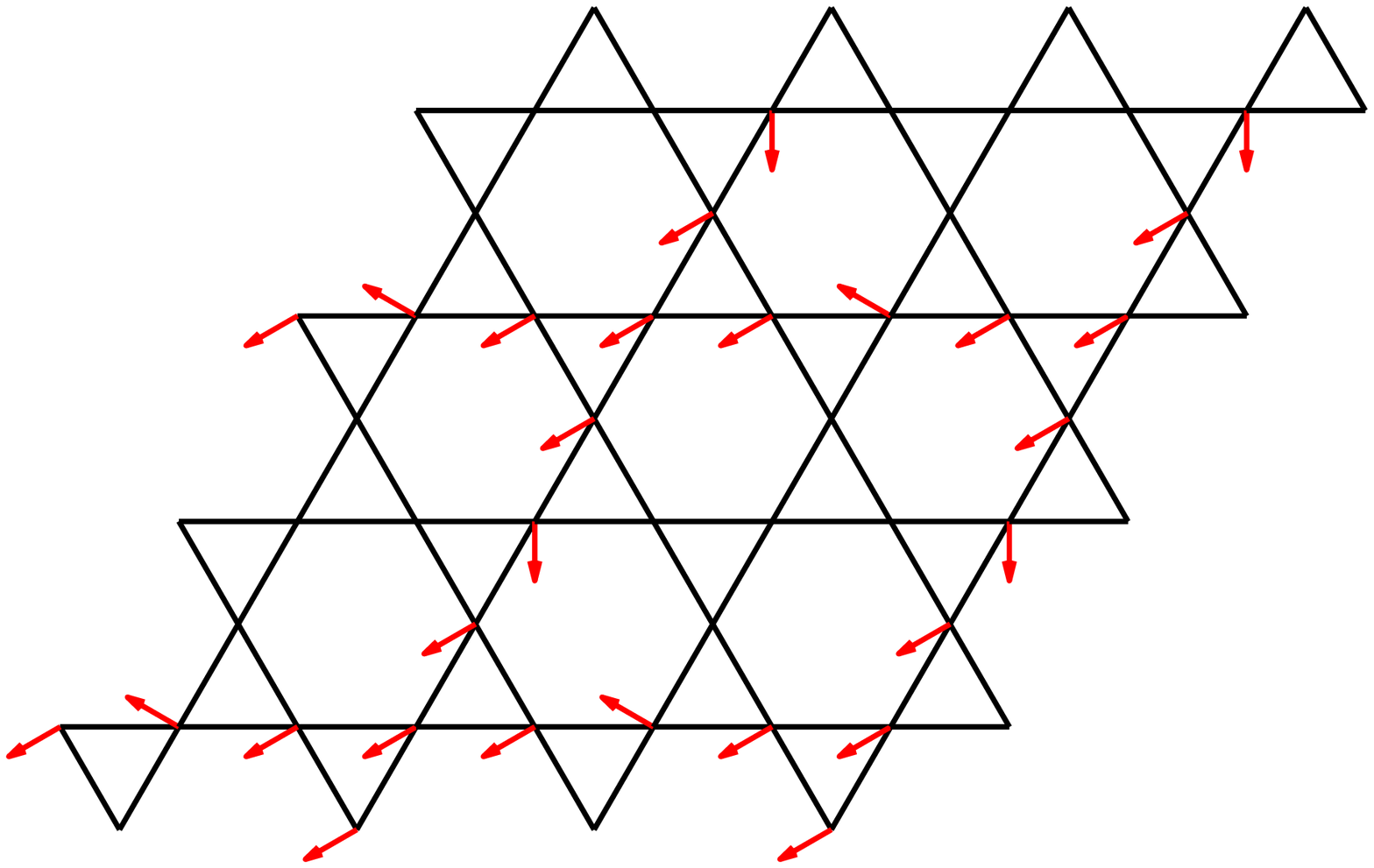}}
		\end{minipage}
		\begin{minipage}[b]{.33\linewidth}
			\subfloat[]{\includegraphics[width=\linewidth]{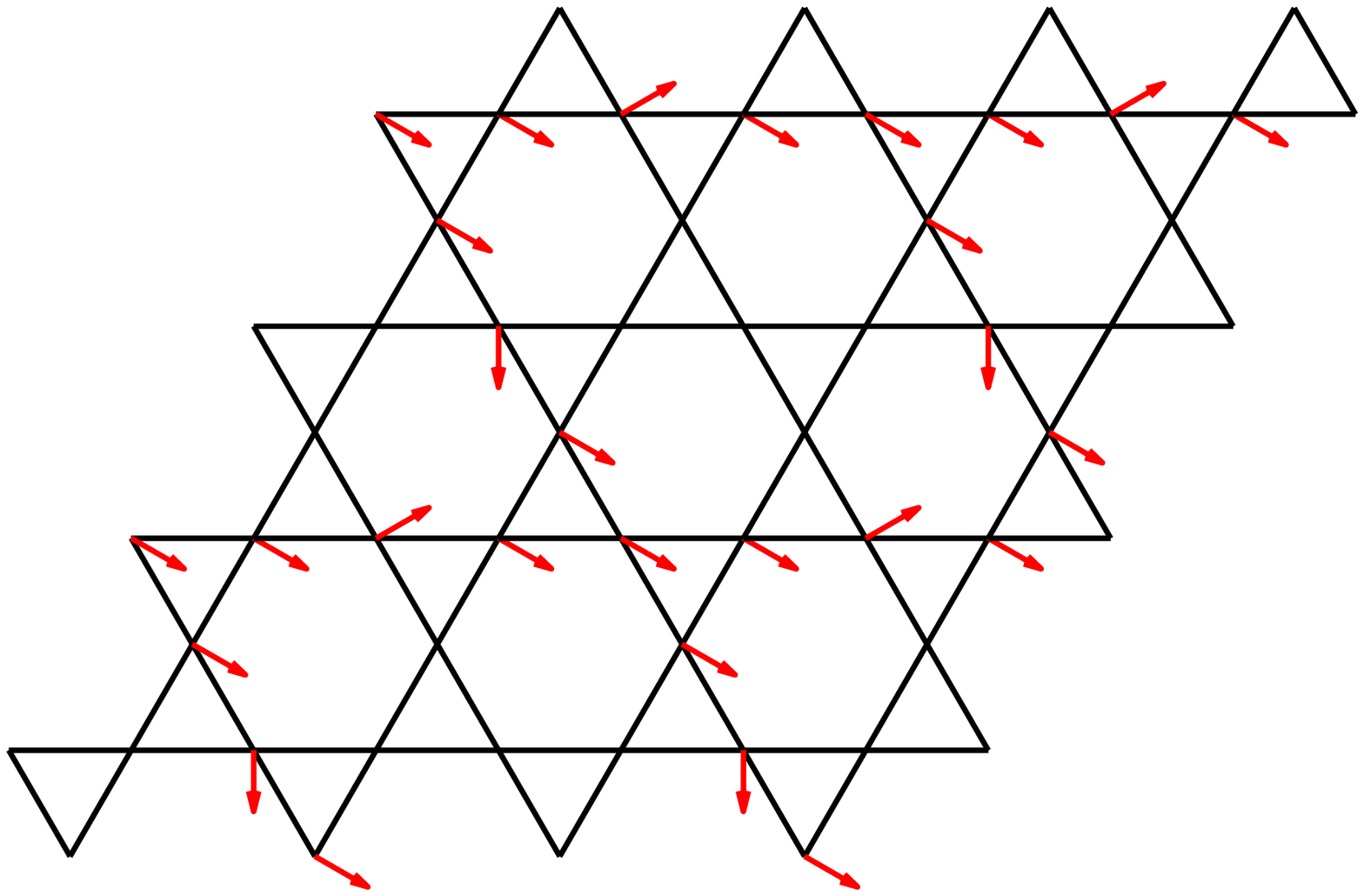}}
		\end{minipage}
		\begin{minipage}[b]{.33\linewidth}
			\centering
			\subfloat[]{\includegraphics[width=\linewidth]{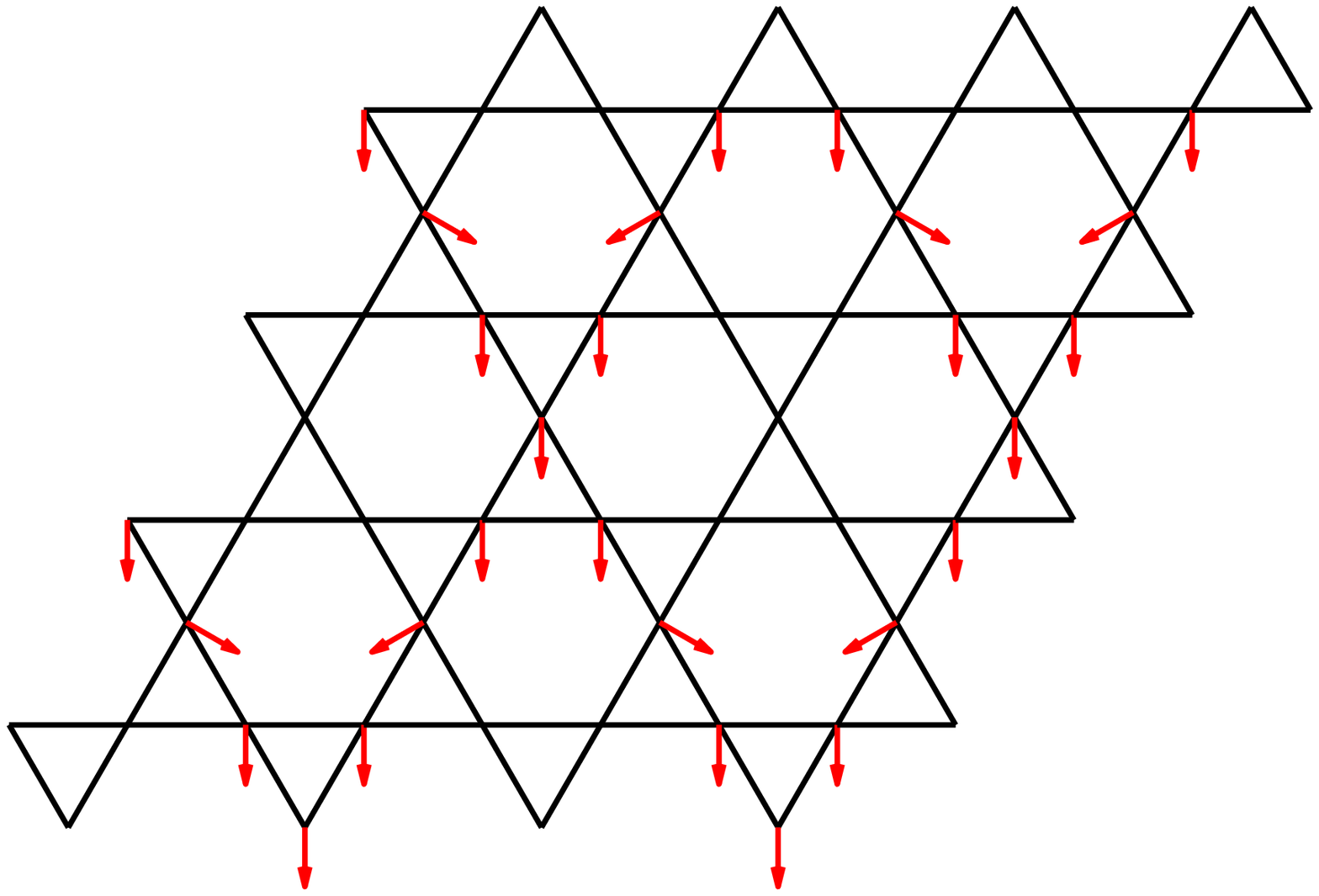}}
		\end{minipage}
		\caption{The three linearly independent two-periodic GH modes of the standard Kagome lattice: (a) the first GH mode $\varphi_1^2(x)$; (b) the second GH mode $\varphi_2^2(x)$; (c) the third GH mode $\varphi_3^2(x)$.}
		\label{fig:two-periodic-standard-GH-modes}
	\end{figure}
	
	We note that the one-periodic mechanism $u_{\frac{\pi}{3} \shortto \theta}(x)$ can also be viewed as a two-periodic mechanism on the standard Kagome lattice. So its infinitesimal version on the two-periodic unit cell is also a two-periodic GH mode. Let us name this particular two-periodic GH mode $\varphi_1^1(x)$. We expect (and it is true) that the one-periodic mechanism is not in the family of the two-periodic mechanism $u_{\theta_1, \theta_2, \theta_3}(x)$; the GH mode $\varphi_1^1(x)$ is not a linear combination of $\varphi_1^2(x), \varphi_2^2(x), \varphi_3^2(x)$. Thus, we have found a basis for the four-dimensional space of two-periodic GH modes.
	
	So far, we have shown that the four GH modes $\varphi_1^1(x), \varphi_1^2(x), \varphi_2^2(x), \varphi_3^2(x)$ form a basis for the space of two-periodic GH modes, and each comes from some mechanism. We also observe that any linear combination of $\varphi_1^2(x), \varphi_2^2(x), \varphi_3^2(x)$ comes from a two-periodic mechanism, i.e. for any $s_1, s_2, s_3 \in \mathbb{R}$,
	\begin{align*}
		s_1 \varphi_1^2(x)  + s_2 \varphi_2^2(x) + s_3 \varphi_3^2(x) &= \frac{d}{d t} \Big |_{t=0} u_{\frac{\pi}{3}-s_1t, \frac{\pi}{3}-s_2t, \frac{\pi}{3}-s_3t}.
	\end{align*}
	A remaining question is whether any other linear combination of these four GH modes comes from a two-periodic mechanism. The answer is no, as we will explain in section \ref{two-periodic-GH-modes}.
	
	\subsection{GH modes of the deformed two-periodic Kagome lattice $L_{\theta_1, \theta_2, \theta_3}$}
	In section \ref{sec:preliminary}, we mentioned the  Guest-Hutchinson theorem that all non-degenerate Maxwell lattices must have a GH mode. In section \ref{one-periodic-mechanism}, we saw that the standard Kagome lattice (which is non-degenerate) has a GH mode, and the twisted Kagome lattices (which are degenerate) do not have GH modes. This leads to the question whether the "reverse" of the Guest-Hutchinson theorem is true; that is: if a lattice has a degenerate linear elastic effective tensor $A_{\text{eff}}$, must there be no GH modes? In fact, such a result is not correct. The two-periodic degenerate Kagome lattice $L_{\theta_1, \theta_2, \theta_3}$ has GH modes. In fact, we will show that the space of GH modes is two-dimensional for all the two-periodic Kagome lattices $L_{\theta_1, \theta_2, \theta_3}$ except for the standard Kagome lattice. An example is shown in Figure \ref{fig:two-periodic-deformed-GH-modes}.
	
	To get started, we observe that among the family of two-periodic Kagome lattices $L_{\theta_1, \theta_2, \theta_3}$, only the standard Kagome lattice is non-degenerate. This is true because for other $L_{\theta_1, \theta_2, \theta_3}$, the compression ratio in \eqref{eqn:two_periodic_macroscopic} is less than 1. This means that we can change $\theta_1, \theta_2, \theta_3$ to further compress or expand it. Therefore, there are mechanisms around $L_{\theta_1, \theta_2, \theta_3}$ such that the infinitesimal version of its macroscopic deformation gradient does not vanish. By Proposition \ref{degeneracy}, the two-periodic Kagome lattice $L_{\theta_1, \theta_2, \theta_3}$ is degenerate. 
	
	However, there is a mechanism that does not change the macroscopic deformation gradient. Since the associated $\dot{F}(0)$ vanishes, there must exist a GH mode by Proposition \ref{degeneracy}. In fact, we shall prove that the space of GH modes is two-dimensional for the deformed two-periodic Kagome lattice $L_{\theta_1, \theta_2, \theta_3}$ by using the implicit function theorem as follows.
	\begin{figure}[H]
		\begin{minipage}[b]{.48\linewidth}
			\centering
			\subfloat[]{\includegraphics[width=0.85\linewidth]{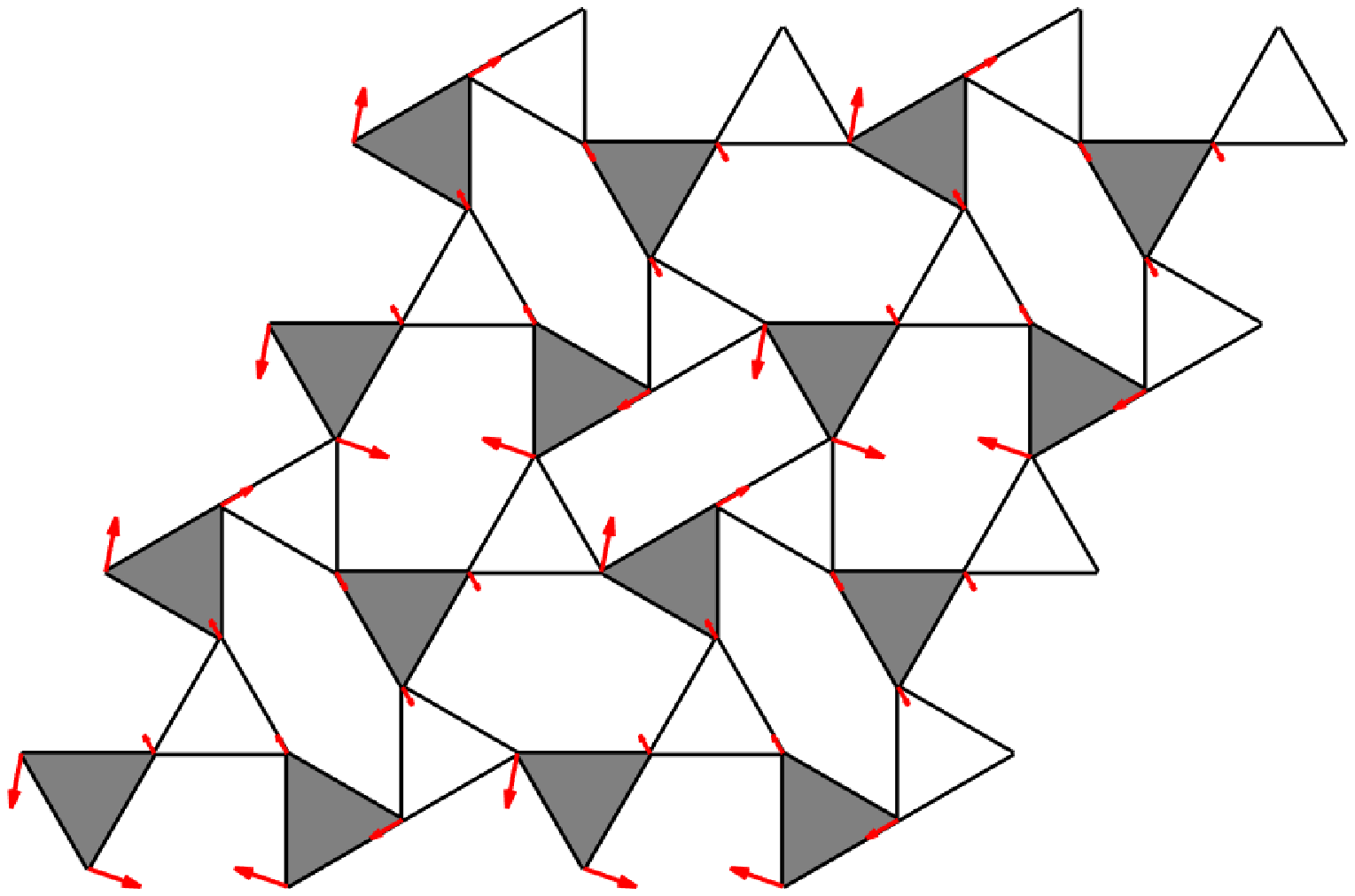}}
		\end{minipage}
		\begin{minipage}[b]{.48\linewidth}
			\centering
			\subfloat[]{\includegraphics[width=0.85\linewidth]{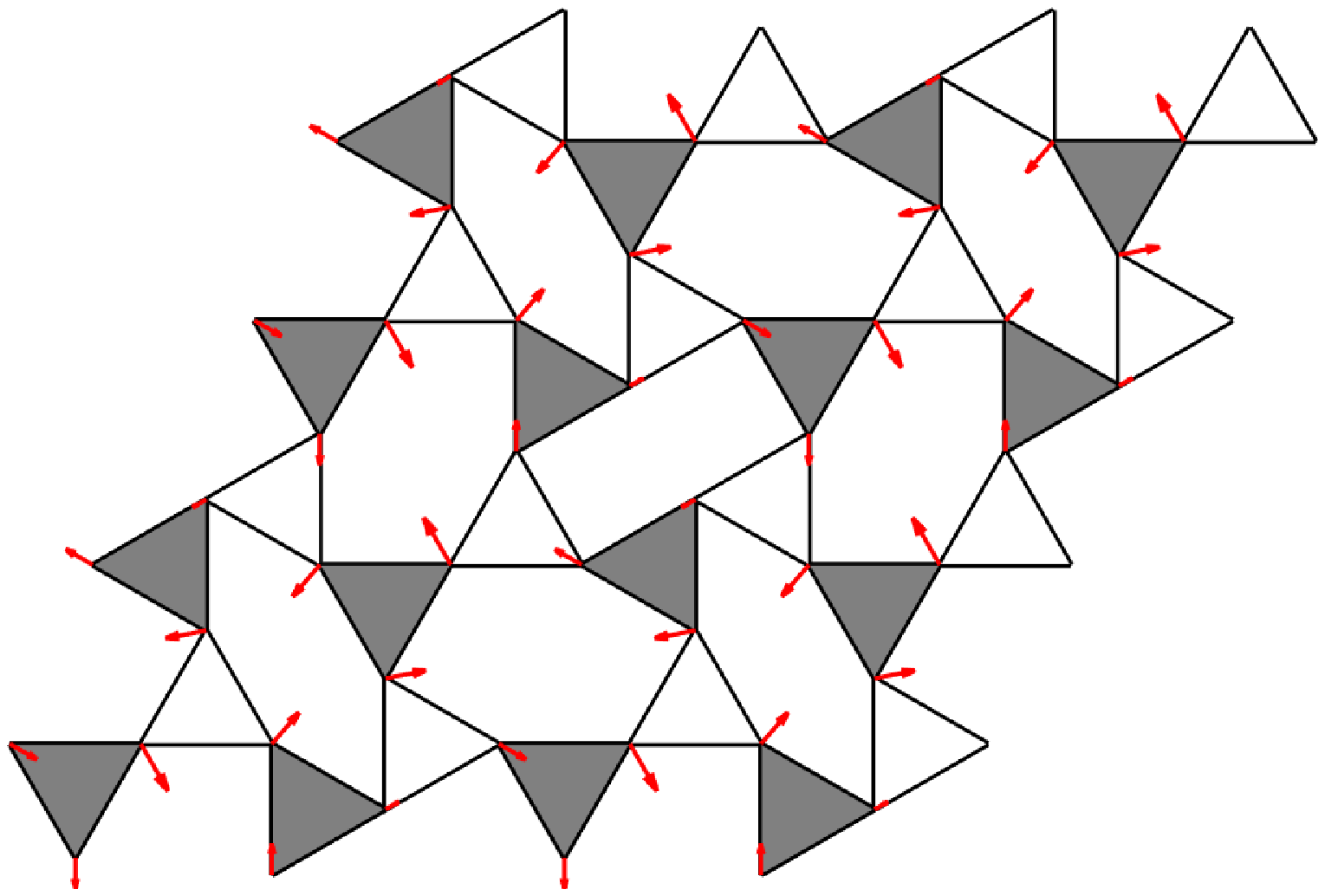}}
		\end{minipage}
		\caption{Two linearly independent GH modes of the deformed two-periodic Kagome lattice $L_{\theta_1, \theta_2, \theta_3}$ with $\theta_1 = \frac{\pi}{3}, \theta_2 = \frac{\pi}{6}$ and $\theta_3=\frac{\pi}{2}$. The arrows indicate the displacement vector of each GH mode at every vertex.}\label{fig:two-periodic-deformed-GH-modes}
	\end{figure}
	
	Consider the function $F(\theta_1, \theta_2, \theta_3, \theta_4)$, defined on the four angles $\theta_1, \theta_2, \theta_3, \theta_4$ that appeared in the construction of the two-periodic Kagome lattice:
	\begin{align*}
		& F(\theta_1, \theta_2, \theta_3, \theta_4) = \begin{pmatrix}
			\frac{1}{4} \big(\cos(\theta_1 - \frac{\pi}{3}) + \cos(\theta_2 - \frac{\pi}{3}) + \cos(\theta_3 - \frac{\pi}{3}) + \cos(\theta_4 - \frac{\pi}{3}) \big)\\
			\sin(\theta_1 - \frac{\pi}{3}) + \sin(\theta_2 - \frac{\pi}{3}) + \sin(\theta_3 - \frac{\pi}{3}) + \sin(\theta_4 - \frac{\pi}{3}).
		\end{pmatrix},
	\end{align*}
	The first component is the macroscopic compression ratio $c_{\theta_1, \theta_2, \theta_3, \theta_4}$ in \eqref{eqn:two_periodic_macroscopic}, and the second component will permit us to avoid macroscopic rotation, using \eqref{eqn:theta_4}. It suffices to show that the level set of $F(\theta_1, \theta_2, \theta_3, \theta_4) = (c, 0)^T$ is a differentiable 2-manifold when the two-periodic Kagome lattice $L_{\theta_1, \theta_2, \theta_3}$ is not the standard Kagome lattice. To use the implicit function theorem, we take the partial derivative to $F$ w.r.t. $\theta_1$ and $\theta_2$
	\begin{align*}
		&\partial_{\theta_1, \theta_2} F = \begin{pmatrix}
			-\frac{1}{4} \sin(\theta_1 - \frac{\pi}{3}) & -\frac{1}{4} \sin(\theta_2 - \frac{\pi}{3})\\
			\cos(\theta_1 - \frac{\pi}{3}) & \cos(\theta_2 - \frac{\pi}{3})
		\end{pmatrix}.
	\end{align*}
	Its determinant is $\frac{1}{4}\sin(\theta_2 - \theta_1)$. When $\theta_1 \neq \theta_2$, the Jacobian matrix $\partial_{\theta_1, \theta_2} F$ is invertible and the implicit function theorem guarantees that the level set $F(\theta_1, \theta_2, \theta_3, \theta_4) = (c, 0)^T$ is a differentiable two-dimensional manifold. Moving along any curve on this manifold, the macroscopic deformation gradient does not change. Thus, the two-dimensional tangent space is a subspace of the space of GH modes. 
	
	So far, we have seen that when a deformed two-periodic Kagome lattice $L_{\theta_1, \theta_2,\theta_3}$ has $\theta_1\neq\theta_2$, its space of GH modes is at least two-dimensional. In fact, using the same method, we can show that if the four angles $\theta_1, \theta_2,\theta_3, \theta_4$ are not the same, then the space of GH mode for this two-periodic Kagome lattice is two-dimensional. It remains to show that when the four angles are the same, it must be the standard Kagome lattice. This is true because when the four angles are the same and the second component of $F$ is zero, the four angles must be $\frac{\pi}{3}$, i.e. $L_{\theta_1, \theta_2,\theta_3}$ is the standard Kagome lattice. Therefore, for all two-periodic degenerate Kagome lattices $L_{\theta_1, \theta_2,\theta_3}$, their spaces of GH modes are at least two-dimensional. In fact, one can check that the space of GH modes of every degenerate $L_{\theta_1, \theta_2,\theta_3}$ is indeed two-dimensional by computing the null space of its compatibility matrix.
	
	\subsection{Special deformed two-periodic Kagome lattices}\label{subsec:special-two-periodic}
	There are two special \textit{one-parameter} two-periodic Kagome lattices in the family of deformed two-periodic Kagome lattices $L_{\theta_1, \theta_2, \theta_3}$, obtained by choosing special angle relations between $\theta_1, \theta_2, \theta_3$. The first one is a two-by-one periodic Kagome lattice, i.e. the unit cell of this Kagome lattice only contains four triangles instead of eight. We can achieve it from $L_{\theta_1, \theta_2, \theta_3}$ by choosing the three angles as a function of one parameter $\gamma$:
	\begin{align*}
		\theta_1 &= \gamma & \theta_2&= \gamma\\
		\theta_3 &= \frac{2\pi}{3}-\gamma & \theta_4&= \frac{2\pi}{3}-\gamma.
	\end{align*}
	We name this two-by-one periodic Kagome lattice $L^{2,1}_{\gamma}$; it is shown in Figure \ref{fig:two-periodic-special}(a). Its primitive vectors $\bm{v}_1$ and $\bm{v}_2$ are
	\begin{align*}
		\bm{v}_1 &= \cos(\frac{\pi}{3} - \gamma)(2,0),\\
		\bm{v}_2 &= \cos(\frac{\pi}{3} - \gamma)(2, 2\sqrt{3}),
	\end{align*}
	where $\big|\bm{v}_1\big| = \frac{1}{2} \big|\bm{v}_2\big|$ since $L_{\gamma}$ is two-by-one periodic. The standard Kagome lattice is achieved when $\gamma = \frac{\pi}{3}$.
	\begin{figure}[!htb]
		\begin{minipage}[b]{.48\linewidth}
			\centering
			\subfloat[]{\includegraphics[width=0.7\linewidth]{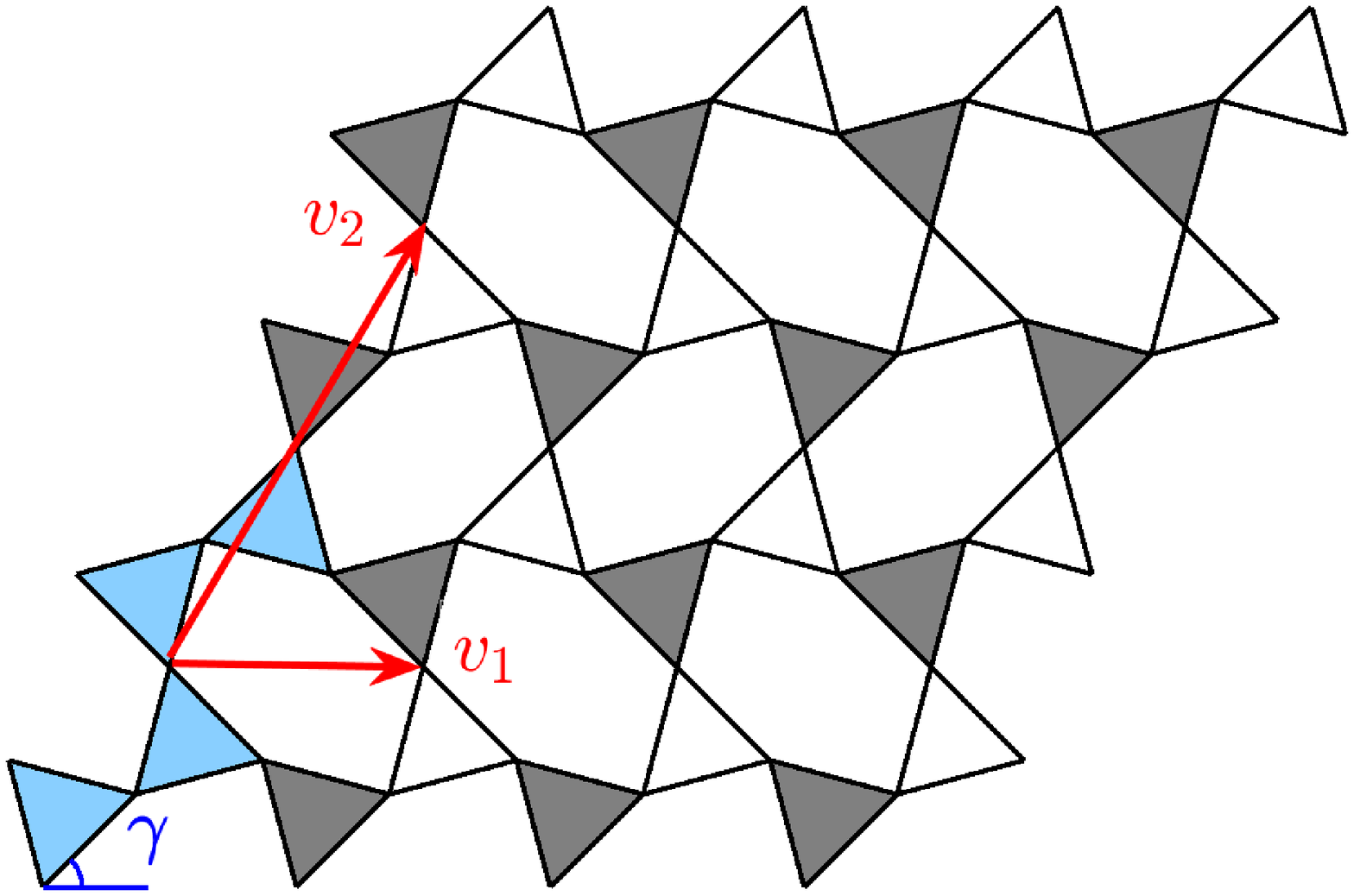}}
		\end{minipage}
		\begin{minipage}[b]{.48\linewidth}
			\centering
			\subfloat[]{\includegraphics[width=0.72\linewidth]{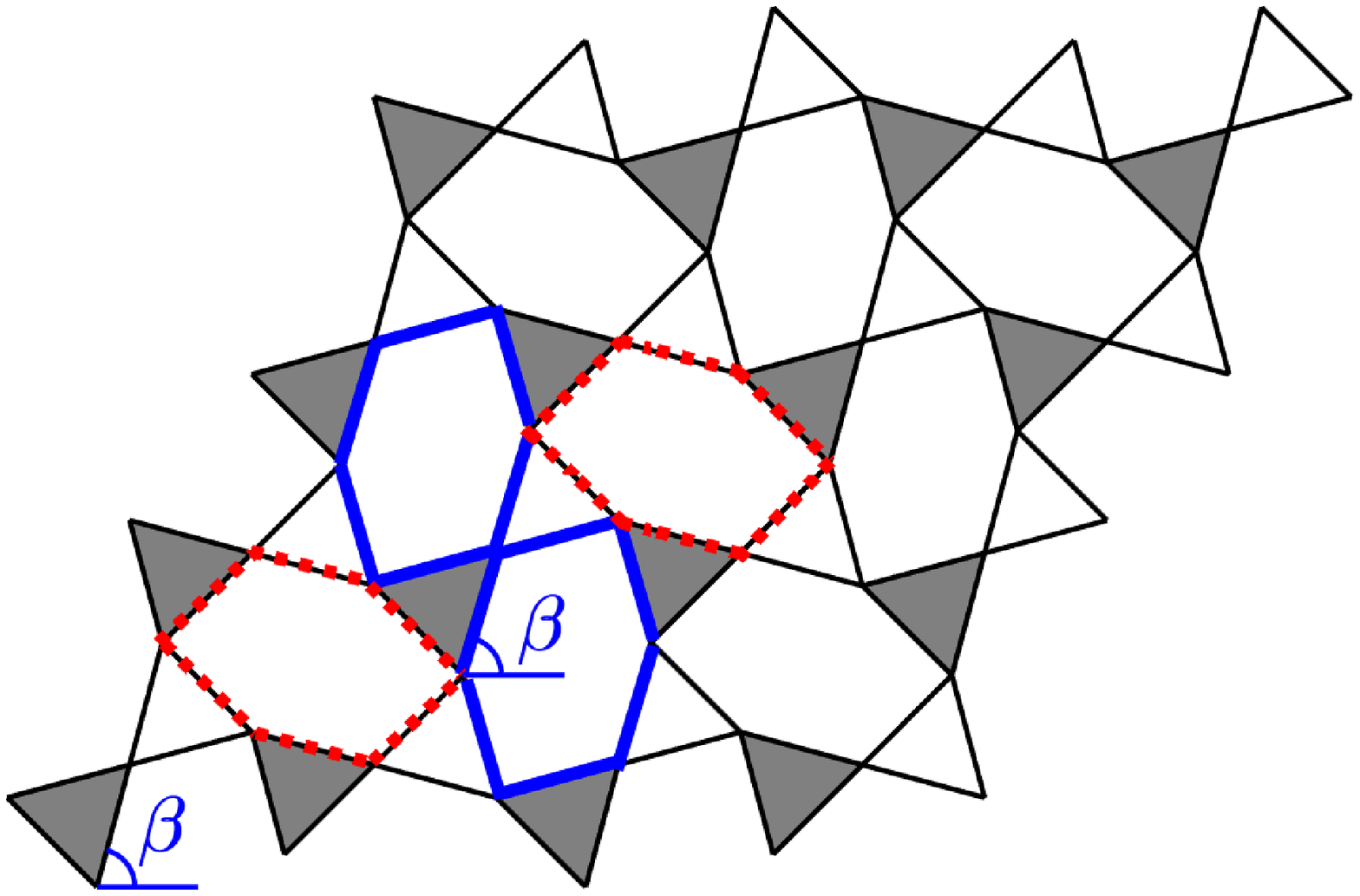}}
		\end{minipage}
		\caption{(a) The two-by-one periodic Kagome lattice $L^{2,1}_{\gamma}$ with $\gamma = \frac{\pi}{4}$: the unit cell consists of four lightly shaded triangles; (b) the special two-periodic Kagome lattice $L^{2,2}_{\beta}$ with $\beta = \frac{5\pi}{12}$: the two hexagons with dotted edges are identical to each other, same for the two hexagons with solid dashed lines .}
		\label{fig:two-periodic-special}
	\end{figure}  
	
	Another special two-periodic Kagome lattice is obtained from the two-periodic mechanism $L_{\theta_1, \theta_2, \theta_3}$ by choosing $\theta_1. \theta_2, \theta_3$ as functions of one parameter $\beta$
	\begin{align*}
		\theta_1 &= \beta & \theta_2&= \frac{2\pi}{3}  - \beta\\
		\theta_3  &= \frac{2\pi}{3}  - \beta & \theta_4 &= \beta.
	\end{align*}
	We name this two-periodic Kagome lattice $L^{2,2}_{\beta}$. One example is shown in Figure \ref{fig:two-periodic-special}(b); as one sees in the figure, the four hexagons are two pairs of identical hexagons. The two primitive vectors $\bm{v}_1, \bm{v}_2$ are
	\begin{align*}
		\bm{v}_1 &= \cos(\frac{\pi}{3} - \beta)(4,0),\\
		\bm{v}_2 &= \cos(\frac{\pi}{3} - \beta)(2, 2\sqrt{3});
	\end{align*}
	here the two primitive vectors are of the same magnitude and the angle between them is 60 degrees.
	\begin{remark}
		These one-parameter two-periodic families of lattices come from two special two-periodic mechanisms of the standard Kagome lattice obtained by smoothly changing $\gamma$ and $\beta$ away from $\frac{\pi}{3}$. It can be checked from the primitive vectors that the macroscopic deformation gradients for the two special mechanisms are
		\begin{align*}
			F^{2,1}_{\gamma} &= \cos(\frac{\pi}{3} - \gamma) I & F^{2,2}_{\beta} &= \cos(\frac{\pi}{3} - \beta) I,
		\end{align*}
		where $F^{2,1}_{\gamma}$ is associated with the two-by-one periodic mechanism and $F^{2,2}_{\beta}$ is associated with the special two-periodic mechanism. Evidently, the two macroscopic deformation gradients are the same as the macroscopic deformation gradient of the one-periodic mechanism $u_{\frac{\pi}{3} \shortto \theta}$ in \eqref{macroscopic-one-periodic} when we choose $\gamma = \beta = \theta$.
	\end{remark}
	
	\section{Which GH modes are linearizations of periodic mechanisms?}\label{sec:necessary-condition}
	{\color{black}In this section, we focus on the question which GH modes come from periodic mechanisms. This is related to the question in finite structures which first-order flexes come from fully nonlinear flexes. It is well-known that there is a necessary condition for a first-order flex to come from a nonlinear flex, i.e. the second-order stress test (see e.g. \cite{connelly1996second}). We shall present a similar necessary condition for GH modes to come from periodic mechanisms (see section \ref{subsec:necessary-condition}). The rest of this section contains applications of the necessary condition to two-periodic GH modes and Fleck-Hutchinson modes to see whether they come from mechanisms.}
	
	{\color{black}Let us briefly explain why not every two-periodic GH mode from the standard Kagome lattice comes from a mechanism.} The necessary condition comes from a geometric observation about the two-periodic mechanism in Figure \ref{fig:two-periodic-overview}. Though straight lines in the reference lattice are deformed into zigzag lines, they still have to experience the same macroscopic contraction on each pair of parallel lines to fit the macroscopic deformation gradient. This leads to, as we shall explain, a constraint on a GH mode $\varphi_1(x)$ for it to come from a mechanism: roughly speaking, the "quadratic part" of the GH mode $\varphi_1(x)$ must average to the same amount on every parallel line in the three lattice directions. This constraint, which we refer to as the \textit{consistency condition}, is explained in sections \ref{subsec:necessary-condition} and \ref{subsection:consistency-condition}. Then we apply this consistency condition in sections \ref{two-periodic-GH-modes} and \ref{subsection:fleck-hutchinson} to understand (a) which linear combinations of our explicit GH modes $\varphi_1^1(x), \varphi_1^2(x), \varphi_2^2(x), \varphi_3^2(x)$ come from mechanisms; and (b) which Fleck-Hutchinson modes come from mechanisms. 
	\begin{figure}[!htb]
		\begin{minipage}[b]{.32\linewidth}
			\centering
			\subfloat[]{\includegraphics[width=0.9\linewidth]{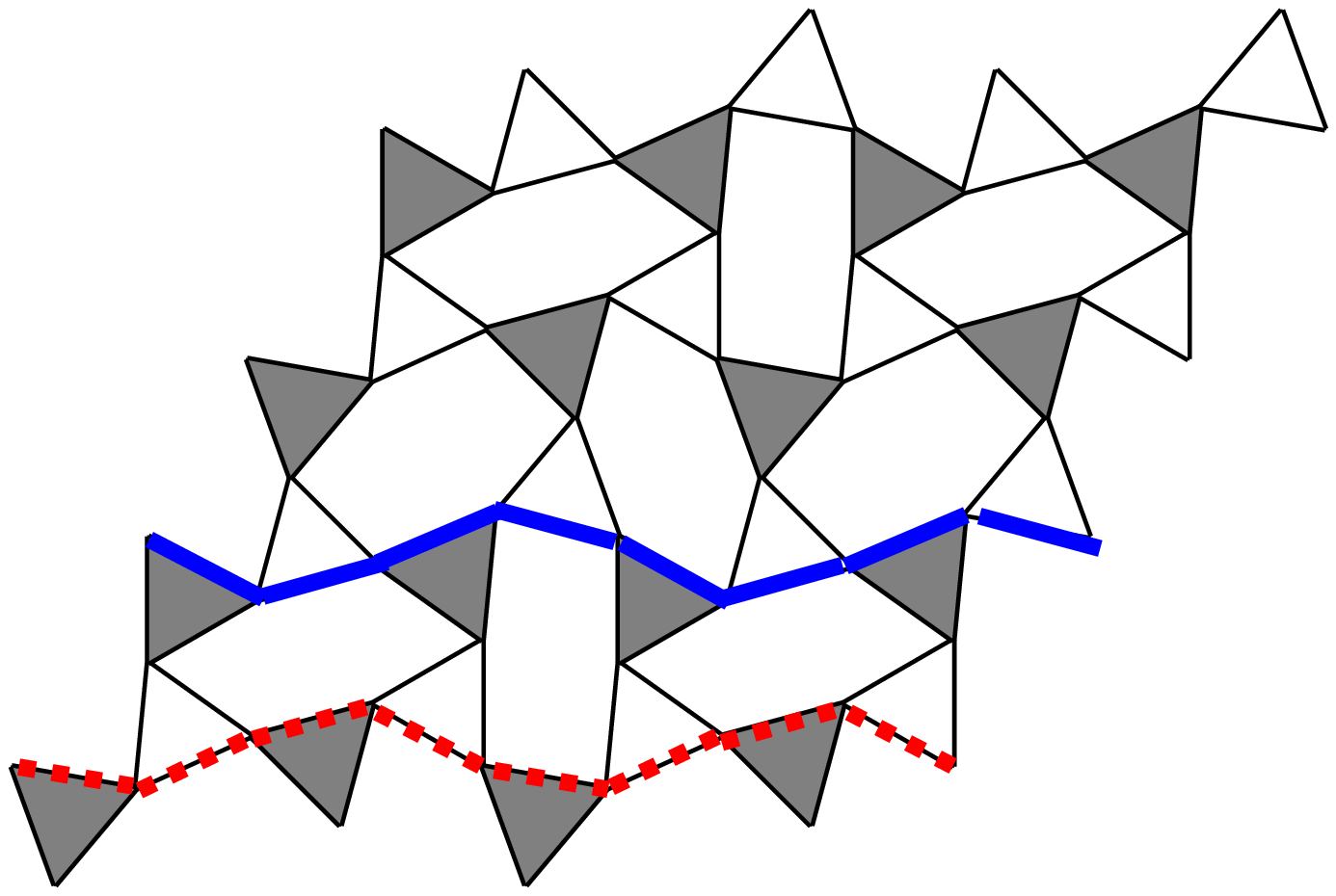}}
		\end{minipage}
		\begin{minipage}[b]{.32\linewidth}
			\centering
			\subfloat[]{\includegraphics[width=0.9\linewidth]{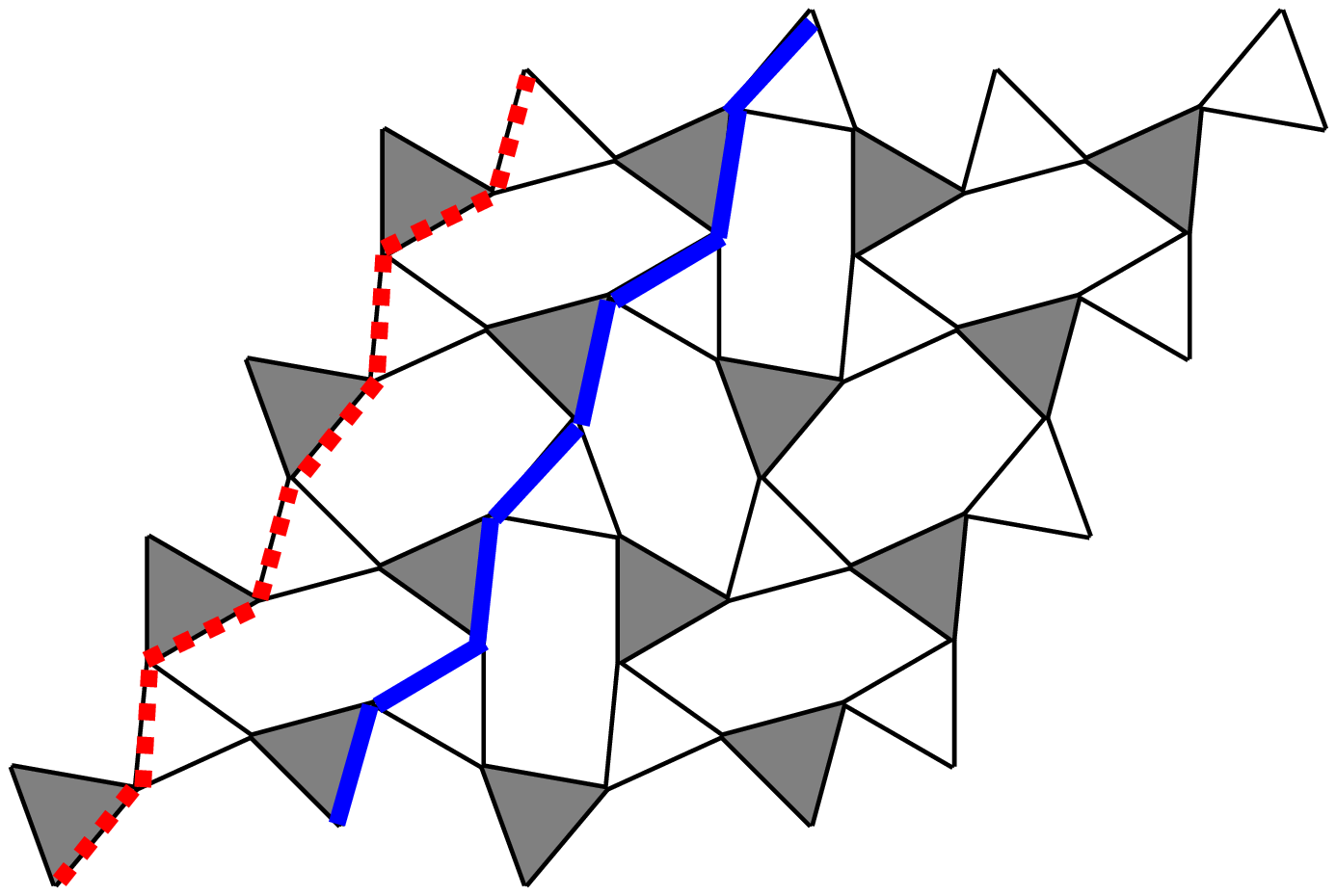}}
		\end{minipage}
		\begin{minipage}[b]{.32\linewidth}
			\centering
			\subfloat[]{\includegraphics[width=0.9\linewidth]{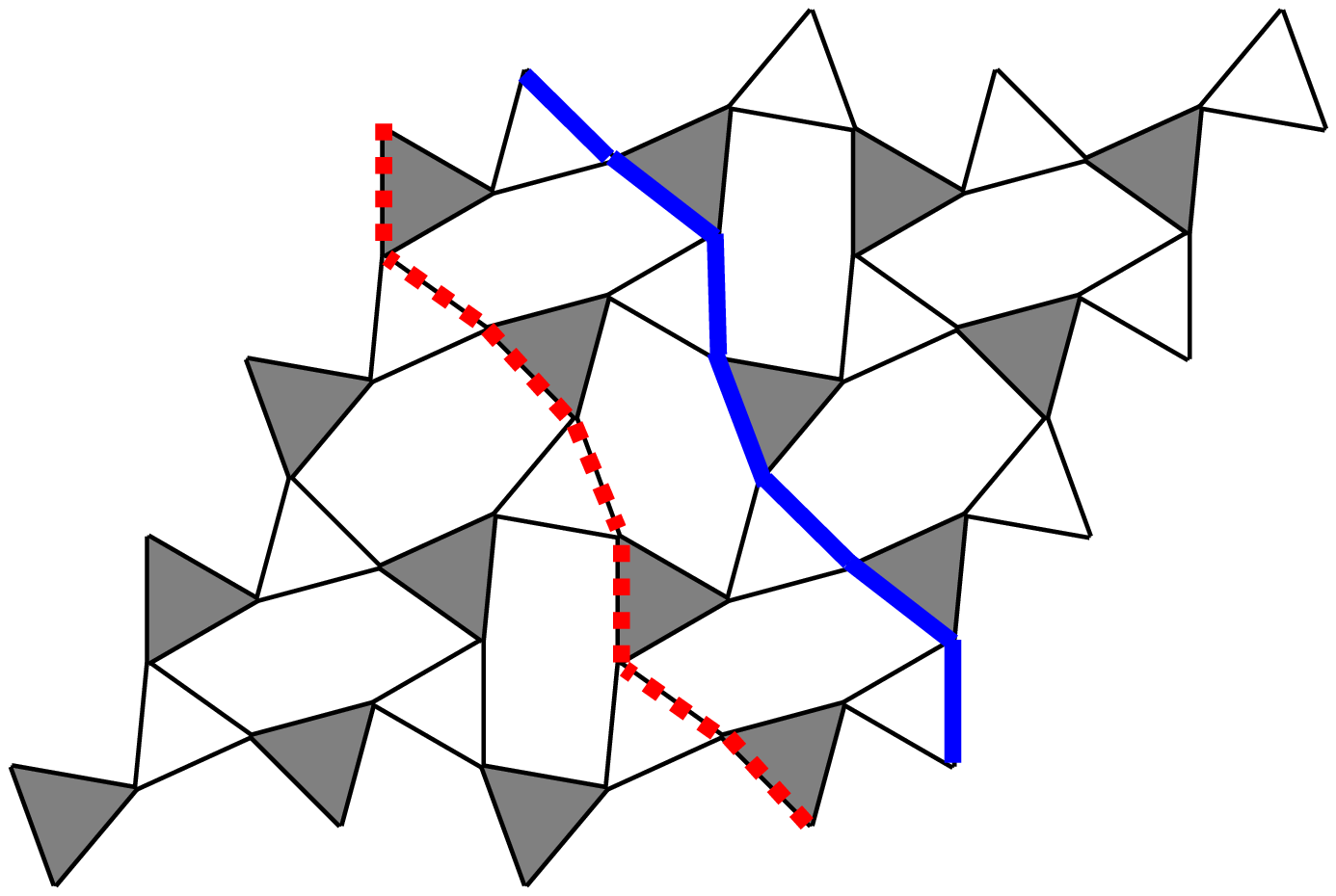}}
		\end{minipage}
		\caption{Zigzag lines in the two-periodic Kagome lattice: (a) two zigzag lines in the horizontal direction; (b) two zigzag lines in the 60 degree direction; (c) two zigzag lines in the 120 degree direction. The two zigzag lines in all three figures must average to the same amount.}
		\label{fig:two-periodic-overview}
	\end{figure}
	
	\subsection{A necessary condition}\label{subsec:necessary-condition}
	Before we discuss the consistency condition, we first derive a necessary condition for a GH mode to come from a mechanism on a general lattice. We need this necessary condition because not every lattice has straight lines of springs like the standard Kagome lattice. Our consistency condition is obtained by specializing {\color{black}our} necessary condition to the standard Kagome lattice.
	
	The necessary condition comes from expanding the potential mechanism to second order. If a GH mode $\varphi_1(x)$ comes from a $N$-periodic smooth mechanism $u(x,t)$, then $u(x,t) = F(t) x + \varphi(x,t)$ with $\dot{\varphi}(x, 0) = \varphi_1(x)$ and $\dot{F}(0) = 0$. The Taylor expansion of $u(x,t)$ around $t=0$ to second order gives
	\begin{align}
		u(x,t) &= x + t \big(\xi_1 x + \varphi_1(x)\big) + t^2 \big(\xi_2 x + \varphi_2(x)\big) + O(t^3), \label{eqn:taylor-expansion}
	\end{align}
	where $u(x,0) = x$ is the reference lattice and $\xi_1 = \dot{F}(0) = 0$. The symmetric $\xi_2 \in \mathbb{R}^{2\times 2}_{\text{sym}}$ is the macroscopic strain at second order and $\varphi_2(x)$ is the second order periodic oscillation.
	
	The existence of $\xi_2$ and $\varphi_2(x)$ gives a constraint on the quadratic part of the GH mode $\varphi_1(x)$, since the mechanism $u(x,t)$ must preserve the length of every spring to second order. On the spring connecting $x_i, x_j$, the mechanism $u(x,t)$ preserves the length of the spring for any $t$, so
	\begin{align}
		\big|x_i-x_j\big|^2 &= \big|u(x_i, t) - u(x_j,t)\big|^2 \label{eqn:taylor-expansion-spring}\\
		&= \big|x_i - x_j\big|^2 + 2t\big\langle x_i - x_j, \varphi_1(x_i) - \varphi_1(x_j)\big \rangle \nonumber\\
		& + t^2 \bigg[\big|\varphi_1(x_i) - \varphi_1(x_j)\big|^2 + 2\big\langle x_i - x_j, \xi_2 (x_i - x_j)\big\rangle + 2\big\langle x_i - x_j, \varphi_2(x_i) - \varphi_2(x_j)\big\rangle \bigg] + O(t^3) \nonumber.
	\end{align}
	The first order term $\big\langle x_i - x_j, \varphi_1(x_i) - \varphi_1(x_j)\big \rangle$ automatically vanishes because $\varphi_1(x)$ is a GH mode (see \eqref{first_order_bond_extension}). At second order, we have
	\begin{align}
		0&=\big|\varphi_1(x_i) - \varphi_1(x_j)\big|^2 + 2\big\langle x_i - x_j, \xi_2 (x_i - x_j)\big\rangle + 2\big\langle x_i - x_j, \varphi_2(x_i) - \varphi_2(x_j)\big\rangle, \nonumber \\
		&= \big|\varphi_1(x_i) - \varphi_1(x_j)\big|^2 + 2l_{ij}^2 \hat{b}_{ij}^T\xi_2 \hat{b}_{ij} + 2 l_{ij} \big\langle \hat{b}_{ij}, \varphi_2(x_i) - \varphi_2(x_j)\big\rangle, \label{eqn:necessary_condition_original}
	\end{align}
	where $l_{ij} =  \big|x_i - x_j\big|$ and $\hat{b}_{ij} = \frac{x_i - x_j}{|x_i - x_j|}$. The last term in vector form is $2 L C \bm{\varphi}_2$, where $L$ is the diagonal matrix with diagonal entries $l_{ij}$ and $C$ is the $N$-periodic compatibility matrix. Writing the first two terms in the vector form, we achieve our necessary condition from \eqref{eqn:necessary_condition_original}:
	\begin{align}
		\bm{e}_{\varphi_1} + \bm{d}_{\xi_2} = L C\bm{\varphi}_2, \label{eqn:necessary_condition}
	\end{align}
	where each entry of $\bm{e}_{\varphi_1}$ is a quadratic term $|\varphi_1(x_i) - \varphi_1(x_j)|^2$ and each entry of $\bm{d}_{\xi_2}$ is $2l_{ij}^2 \hat{b}_{ij}^T\xi_2 \hat{b}_{ij}$.
	
	\subsection{The necessary condition for the standard Kagome lattice}\label{subsection:consistency-condition}
	Now we use our understanding of the standard Kagome lattice -- especially our understanding of its self-stresses -- to specialize the necessary condition \eqref{eqn:necessary_condition} to this setting. While the details are specific to the standard Kagome lattice, our calculation is similar in spirit to the second-order stress test for bars in \cite{connelly1996second} -- a condition that must be satisfied by a first-order flex of a bar framework if it comes from a mechanism.
	
	Without loss of generality, we assume that the length of each spring in the reference lattice is 1. The necessary condition in \eqref{eqn:necessary_condition} tells that if a GH mode $\varphi_1(x)$ comes from a $N$-periodic mechanism, then there is a symmetric matrix $\xi_2$ such that  $\bm{e}_{\varphi_1}+ \bm{d}_{\xi_2} \in \Ima(LC) = \Ima(C)$ ($L = I$ because we assume the lengths of all springs are equal to 1). By the Fredholm Alternative,  it is equivalent to require that $\color{black}\bm{e}_{\varphi_1} + \bm{d}_{\xi_2}$ be orthogonal to all $N$-periodic self-stresses. Thus, the equivalent necessary condition is: there exists a symmetric matrix $\xi_2$ such that for all $\bm{s} \in \ker(C^T)$, 
	\begin{align}
		\langle \bm{e}_{\varphi_1}+ \bm{d}_{\xi_2}, s \rangle = 0. \label{eqn:necessary_stresses}
	\end{align}
	For simplicity, we discuss the case $N=2$; it will be clear that the same method can be applied to in the $N$-periodic case, for any $N$. We apply our equivalent necessary condition in \eqref{eqn:necessary_stresses} with the six explicit two-periodic self-stresses in Figure \ref{fig:self-stresses}(d)-(f). We first plug the two horizontal self-stresses in Figure \ref{fig:self-stresses}(d) into \eqref{eqn:necessary_stresses} and get
	\begin{align}
		\sum_{x_i, x_j \text{ on the solid line}}\big|\varphi_1(x_i) - \varphi_1(x_j)\big|^2 + 2\big\langle x_i - x_j, \xi_2 (x_i - x_j)\big\rangle = 0, \label{eqn:red}\\
		\sum_{x_i, x_j \text{ on the dotted line}}\big|\varphi_1(x_i) - \varphi_1(x_j)\big|^2 + 2\big\langle x_i - x_j, \xi_2 (x_i - x_j)\big\rangle = 0. \label{eqn:orange}
	\end{align}
	Notice that every vector $x_i - x_j$ is a unit vector in the horizontal direction. Therefore, the sum of $2\big\langle x_i - x_j, \xi_2 (x_i - x_j)\big\rangle$ over the two lines in the horizontal direction are the same. This indicates that the quadratic part of $\varphi_1(x)$ must sum up to the same amount over the two lines in the horizontal direction
	\begin{align}
		\sum_{x_i, x_j \text{ on the solid line}}\big|\varphi_1(x_i) - \varphi_1(x_j)\big|^2= \sum_{x_i, x_j \text{ on the dotted line}}\big|\varphi_1(x_i) - \varphi_1(x_j)\big|^2. \label{eqn:consistency_horizontal}
	\end{align}
	Similarly, we get two more conditions in the 60 and 120 degree directions by plugging the four self-stresses in Figure \ref{fig:self-stresses}(e) and (f) into \eqref{eqn:necessary_stresses}. The three conditions say that the \textit{quadratic part} of $\varphi_1(x)$ must sum up to the same amount on every parallel line in the three lattice directions. We call these three conditions the consistency condition. 
	
	\begin{remark}
		When $N > 2$, the consistency condition still requires that the quadratic part of a given GH mode $\varphi_1(x)$ must sum to the same amount {\color{black}on each of the parallel lines in one of the three lattice directions}. For example, on the horizontal direction, there are in total $N$ parallel lines (each line consists of $2N$ springs). The consistency condition requires that the quadratic part of $\varphi_1(x)$ {\color{black}on each of the $N$ horizontal lines} sum up to the same amount. The other two directions follow the same rule.
	\end{remark}
	
	So far, the consistency condition seems weaker than the necessary condition in \eqref{eqn:necessary_condition}. However, the consistency condition is indeed equivalent to the necessary condition on the standard Kagome lattice for any periodicity. This is true because we can determine $\xi_2$ by adding up the quadratic part of $\varphi_1(x)$ in the three lattice directions, when a GH mode $\varphi_1(x)$ satisfies the consistency condition. Let us {\color{black}again} take $N=2$ as an example. If a two-periodic GH mode $\varphi_1(x)$ satisfies the consistency condition, then its quadratic part must satisfy \eqref{eqn:red}-\eqref{eqn:orange}. By summing up the two equations \eqref{eqn:red}-\eqref{eqn:orange}, we get
	\begin{align}
		\sum_{(x_i-x_j) \parallel \; (1,0)}\big|\varphi_1(x_i) - \varphi_1(x_j)\big|^2 &= \sum_{x_i, x_j \text{ on the solid line}}\big|\varphi_1(x_i) - \varphi_1(x_j)\big|^2 + \sum_{x_i, x_j \text{ on the dotted line}}\big|\varphi_1(x_i) - \varphi_1(x_j)\big|^2 \nonumber \\
		&= -2*8\begin{pmatrix} 1 & 0 \end{pmatrix} \xi_2 \begin{pmatrix} 1\\ 0 \end{pmatrix}, \label{eqn:horizontal}
	\end{align}
	where 8 is the number of springs in the horizontal direction in the two-periodic unit cell. Similarly for the 60 and 120 degree directions, we get
	\begin{align}
		& \sum_{(x_i-x_j) \parallel \; \left(\frac{1}{2},\frac{\sqrt{3}}{2}\right)}\big|\varphi_1(x_i) - \varphi_1(x_j)\big|^2 = -2 *8\begin{pmatrix} \frac{1}{2} & \frac{\sqrt{3}}{2} \end{pmatrix} \xi_2 \begin{pmatrix} \frac{1}{2} \\ \frac{\sqrt{3}}{2} \end{pmatrix}, \label{eqn:60}\\
		& \sum_{(x_i-x_j) \parallel \; \left(-\frac{1}{2},\frac{\sqrt{3}}{2}\right)} \big|\varphi_1(x_i) - \varphi_1(x_j)\big|^2 = -2 *8 \begin{pmatrix} -\frac{1}{2} & \frac{\sqrt{3}}{2} \end{pmatrix} \xi_2 \begin{pmatrix} -\frac{1}{2} \\ \frac{\sqrt{3}}{2} \end{pmatrix}. \label{eqn:120}
	\end{align}
	Since $\xi_2$ is a symmetric matrix with three degrees of freedom, it is fully determined by the averaged value of the quadratic part of $\varphi_1(x)$ in the three lattice directions. For $N\neq 2$, equations \eqref{eqn:horizontal}-\eqref{eqn:120} still hold but the number 8 is replaced by $2N^2$ since there are $2N^2$ springs in each lattice direction. Thus, we have shown that the consistency condition is indeed equivalent to the necessary condition in section \ref{subsec:necessary-condition} for the standard Kagome lattice for any periodicity.
	
	\subsection{A complete understanding of the two-periodic GH modes}\label{two-periodic-GH-modes}
	We apply the consistency condition to a linear combination of the four explicit two-periodic GH modes $\varphi_1^1(x), \varphi_1^2(x), \varphi_2^2(x), \varphi_2^3(x)$. We already know that $\varphi_1^1$ comes from a one-periodic mechanism and any linear combination of $\varphi_1^2, \varphi_2^2, \varphi_3^2$ comes from a two-periodic mechanism. So the remaining question is whether any other linear combinations of these GH modes ever satisfy the consistency condition. The answer is no.
	
	To see why the consistency condition fails, it is equivalent to check when a sum of two GH modes that both satisfy the consistency condition still satisfies the consistency condition.  Let us write any two-periodic GH mode as $\psi= \psi^1 + \psi^2$, where $\psi^1 \in \text{span}\{\varphi^1_1\}$ and $\psi^2 \in \text{span}\{\varphi^2_1, \varphi^2_2, \varphi^2_3\}$. The two GH modes $\psi^1, \psi^2$ both satisfy the consistency condition. If the sum $\psi$ satisfies the consistency condition, then its quadratic part should sum up to the same amount in every lattice directions. For example, in the horizontal direction, we get
	\begin{align*}
		& \sum_{x_i, x_j \text{ on the solid line}} \big|\psi(x_i) - \psi(x_j)\big|^2 = \sum_{x_i, x_j \text{ on the dotted line}} \big|\psi(x_i) - \psi(x_j)\big|^2,
	\end{align*}
	Expanding it in terms of $\psi^1, \psi^2$, we get
	\begin{align*}
		& \sum_{x_i, x_j \text{ on the solid line}} \big|\psi^1(x_i) - \psi^1(x_j)\big|^2 + \big|\psi^2(x_i) - \psi^2(x_j)\big|^2 + 2\big\langle \psi^1(x_i) - \psi^1(x_j), \psi^2(x_i) - \psi^2(x_j) \big\rangle\\
		= & \sum_{x_i, x_j \text{ on the dotted line}} \big|\psi^1(x_i) - \psi^1(x_j)\big|^2 + \big|\psi^2(x_i) - \psi^2(x_j)\big|^2 + 2\big\langle \psi^1(x_i) - \psi^1(x_j), \psi^2(x_i) - \psi^2(x_j) \big\rangle.
	\end{align*}
	Since the two GH modes $\psi^1$ and $\psi^2$ satisfy the consistency condition, the first two terms are already matched. This leads to a constraint on the cross term
	\begin{align}
		& \sum_{x_i, x_j \text{ on the solid line}} \big\langle \psi^1(x_i) - \psi^1(x_j), \psi^2(x_i) - \psi^2(x_j) \big\rangle
		= \sum_{x_i, x_j \text{ on the dotted line}} \big\langle \psi^1(x_i) - \psi^1(x_j), \psi^2(x_i) - \psi^2(x_j) \big\rangle. \label{eqn:two-periodic-horizontal}
	\end{align}
	Similarly, we get two more constraints for the cross terms on the 60 and 120 degree directions. We observe that these constraints are bilinear in $\psi_1$ and $\psi_2$ (that is, linear in either $\psi_1$ or $\psi_2$ if the other is held fixed).
	
	We would like to write down the consistency condition as a linear system for the coefficients of a GH mode in terms of our explicit basis. WLOG, we can choose $\psi = \varphi^1_1 + a_1 \varphi^2_1 + a_2 \varphi^2_2 + a_3 \varphi^2_3$. By plugging $\psi^1 = \varphi^1_1$ and $\psi^2 = a_1 \varphi^2_1 + a_2 \varphi^2_2 + a_3 \varphi^2_3$ into \eqref{eqn:two-periodic-horizontal}, the constraint becomes a linear system in terms of $a_1, a_2, a_3$, and it is
	\begin{align*}
		0 = \: & a_1 \bigg[\sum_{x_i, x_j \text{ on the solid line}} \big\langle \varphi^1_1(x_i) - \varphi^1_1(x_j), \varphi^2_1(x_i) - \varphi^2_1(x_j) \big\rangle - \sum_{x_i, x_j \text{ on the dotted line}} \big\langle \varphi^1_1(x_i) - \varphi^1_1(x_j), \varphi^2_1(x_i) - \varphi^2_1(x_j) \big\rangle \bigg]\\
		+ \: & a_2 \bigg[\sum_{x_i, x_j \text{ on the solid line}} \big\langle \varphi^1_1(x_i) - \varphi^1_1(x_j), \varphi^2_2(x_i) - \varphi^2_2(x_j) \big\rangle - \sum_{x_i, x_j \text{ on the dotted line}} \big\langle \varphi^1_1(x_i) - \varphi^1_1(x_j), \varphi^2_2(x_i) - \varphi^2_2(x_j) \big\rangle \bigg]\\
		+ \: & a_3 \bigg[\sum_{x_i, x_j \text{ on solid line}} \big\langle \varphi^1_1(x_i) - \varphi^1_1(x_j), \varphi^2_3(x_i) - \varphi^2_3(x_j) \big\rangle - \sum_{x_i, x_j \text{ on the dotted line}} \big\langle \varphi^1(x_i) - \varphi^1(x_j), \varphi^2_3(x_i) - \varphi^2_3(x_j) \big\rangle \bigg].
	\end{align*}
	We also get two more linear constraints on $a_1, a_2, a_3$ from the 60 and 120 degree direction. Using the explicit forms of the four GH modes $\varphi_1^1, \varphi_1^2, \varphi_2^2, \varphi_2^2$ (see Appendix \ref{appendix-c}), we {\color{black}find that} the linear system for $a_1, a_2, a_3$ is
	\begin{align*}
		\begin{bmatrix}
			4  &  4    &    0\\
			4    &     0  &  4\\
			0 &  -4  &  -4
		\end{bmatrix} \begin{bmatrix}
			a_1\\
			a_2\\
			a_3
		\end{bmatrix} = \begin{bmatrix}
			0\\
			0\\
			0
		\end{bmatrix}.
	\end{align*} 
	Clearly, the system only has the zero solution. This indicates that among all GH modes in the form of $\psi = \varphi_1^1 + a_1 \varphi_1^2 + a_2 \varphi_2^2 + \varphi_3^2$, only $\varphi_1^1$ satisfies the consistency condition. Thus, a non-trivial linear combination of the one-periodic GH mode and a two-periodic GH mode does not come from a two-periodic mechanism. The set of two-periodic GH modes that come from two-periodic mechanisms is shown in Figure \ref{fig:tangent-cone}. It is the union of a line generated by $\varphi^1_1(x)$ and a three-dimensional subspace generated by $\varphi_1^2, \varphi_2^2,\varphi_3^2$. A GH mode outside this set does not come from a two-periodic mechanism.
	\begin{figure}[!htb]
		\centering
		\includegraphics[width=0.5\linewidth]{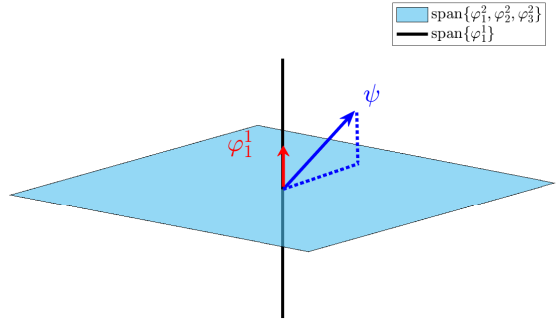}
		\caption{The set of two-periodic GH modes that come from two-periodic mechanisms, also known as the tangent cone, is a union of two subspaces. One subspace is generated by the GH mode $\psi_1^1$ from the one-periodic mechanism $u_{\frac{\pi}{3} \shortto \theta}$ and another one is generated by the three GH modes $\varphi_1^2, \varphi_2^2, \varphi_3^2$ from the two-periodic mechanism $u_{\theta_1, \theta_2, \theta_3}(x)$. Any GH mode outside this tangent cone, for example $\psi$ in the figure, does not come from a two-periodic mechanism.}\label{fig:tangent-cone}
	\end{figure}
	
	\begin{remark}\label{rmk:singularity}
		The preceding discussion was geometric. A different, more abstract explanation why some two-periodic GH modes do not come from mechanisms starts by considering the zero level set of the elastic energy on the two-periodic unit cell. The image of every two-periodic mechanism lies in this zero level set. The infinitesimal version of a two-periodic mechanism is thus a tangent vector of a curve on the level set. The space of two-periodic GH modes is the tangent space. If this zero level set is a smooth manifold, where the implicit function theorem applies, then any vector in the tangent space is a tangent vector of a curve on the level set. However, the zero level set is not a smooth manifold at the two-periodic standard Kagome lattice, since not every two-periodic GH mode comes from a mechanism. This means the two-periodic standard Kagome lattice is a \textit{singular point} (see e.g. Lecture 20 in \cite{harris2013algebraic}). The real set of all two-periodic GH modes that come from mechanisms, defined as the \textit{tangent cone}, is a union of a line and a three-dimensional subspace, which is not a vector space. A GH mode outside the tangent cone does not come from a two-periodic mechanism.
	\end{remark}
	
	\begin{remark}
		{\color{black}In a lattice with a line of springs, a GH mode must take a straight line to a zigzag line. The existence of such GH modes is very intuitive, since nodal displacements normal to the line do not stretch the springs in the linear elastic approximation. But this observation does not help us see which GH modes come from mechanisms.
			
			Do all $k$-periodic GH modes come from $k$-periodic mechanisms? This question is related to asking whether the set of $k$-periodic mechanisms is singular. For the standard Kagome lattice, all one-periodic GH modes come from one-periodic mechanisms; but not all two-periodic GH modes come from two-periodic mechanisms. As discussed in Remark \ref{rmk:singularity}, this shows that the set of two-periodic mechanisms is singular.}
	\end{remark}
	
	\subsection{Applying the consistency condition to Fleck-Hutchinson modes}\label{subsection:fleck-hutchinson}
	Fleck and Hutchinson \cite{hutchinson2006structural} found a special class of GH modes by studying a linear elasticity problem in the unit cell with a Bloch-type boundary condition. We call these special GH modes Fleck-Hutchinson modes. The Fleck-Hutchinson modes provide a basis for the space of $N$-periodic GH modes. It is natural to ask whether the special features of a Fleck-Hutchinson mode assure that it comes from a mechanism. The answer is no: in fact, for the standard Kagome lattice, there are very few examples of Fleck-Hutchinson modes that come from mechanisms. This section justifies the preceding statement, as an application of our consistency condition in section \ref{subsection:consistency-condition}. In particular, we shall show that the $N$-by-one periodic Fleck-Hutchinson modes almost never come from a mechanism (see below for a more complete and precise summary of this section's results). 
	
	Evidently, considering individual Fleck-Hutchinson modes is not an efficient means of finding nonlinear mechanisms. This raises the question what \emph{other} tool might be used to find mechanisms. For the $N$-by-one mechanisms of the standard Kagome lattice, we shall offer an approach based on layering in section \ref{subsec:non-periodic}. 
	
	We start by reviewing some properties of the Fleck-Hutchinson modes (see Appendix \ref{appendix-e-1} for a detailed review).
	\begin{itemize}
		\item For any $N \geq 2$, these $N$-periodic GH modes of the standard Kagome lattice are obtained by considering complex displacements $d(x)$ with vanishing linear strain on the unit cell with a Bloch-type boundary condition (see \eqref{eqn:bloch-appendix})
		\begin{align}
			d(\bm{j} + \bm{x}) &= d(\bm{j}) \exp(2\pi i \bm{x}\cdot \bm{w}), \label{eqn:bloch-main}
		\end{align}
		where $\bm{j}$ are vertices in the unit cell of the standard Kagome lattice, i.e. $\bm{j} = A,B,C$ in Figure \ref{fig:standardkagomespring}(b). The vector $\bm{w}$ is the Bloch wave number. It is in the form of $\bm{w} = w_1 \bm{a}_1 + w_2 \bm{a}_2$, where $\bm{a}_1, \bm{a}_2$ are primitive vectors in the Brillouin zone (see their explicit values around \eqref{eqn:brillouin}).
		
		\item There are three types of $\bm{w}$ that give $N$-periodic GH modes: (1) $w_1 = w_2 = \frac{s}{N}$; (2) $w_1 = 0, w_2 = \frac{s}{N}$; and (3) $w_2 = 0, w_1 = \frac{s}{N}$. In all three cases, the integer $s$ can be chosen in the range $0 \leq s \leq \lfloor \frac{N}{2} \rfloor$. We shall focus here only on the GH modes associated to the first case $w_1 = w_2 = \frac{s}{N}$, since the other two cases are related to these by symmetry (see Remark \ref{rmk:remark-symmetry}).
		
		\item We achieve a one-dimensional family of complex displacements $d(x)$ by solving the relevant linear system with the Bloch-type boundary condition in \eqref{eqn:bloch-main} and with Bloch wave number $w_1 = w_2 = \frac{s}{N}$. Taking the real and imaginary parts of this unique $d(x)$ gives two real-valued GH modes $u_1^{s,N}(x)$ and $u_2^{s,N}(x)$.
		
		\item The special GH modes $u_1^{s,N}(x)$ and $u_2^{s,N}(x)$ are actually $N$-by-1 periodic for any $s$ in the range $0 \leq s \leq \lfloor \frac{N}{2} \rfloor$. As explained in Appendix \ref{appendix-e-1}, they have period 1 in the horizontal direction and period $N$ in the 60 degree direction (see Figure \ref{fig:three-by-one-consistency}(a)-(b) for an example with $N=3$). Moreover, their values can be written down explicitly: there are a total of $3N$ vertices in the $N$-by-one periodic unit cell, and we refer to them as $A_{0,k}, B_{0,k},C_{0,k}$ with $k = 0,1,\dots, N-1$ (see Figure \ref{fig:three-by-one-consistency}(c) for the case $N=3$); the exact values of $u_1^s(x)$ and $u_2^s(x)$ on the $3N$ vertices are then
		\begin{align}
			u_1^{s,N}(A_{0,k}) &= (0,0)^T & u_1^{s,N}(B_{0,k}) &= \cos(\frac{2k\pi}{N})\begin{pmatrix}
				\frac{\sqrt{3}}{2}, -\frac{1}{2}
			\end{pmatrix}^T & u_1^{s,N}(C_{0,k}) &= \cos(\frac{2ks\pi}{N})\begin{pmatrix}
				\frac{\sqrt{3}}{2}, \frac{1}{2}
			\end{pmatrix}^T \label{eqn:u_1}\\
			u_2^{s,N}(A_{0,k}) &= (0,0)^T & u_2^{s,N}(B_{0,k}) &= \sin(\frac{2k\pi}{N})\begin{pmatrix}
				\frac{\sqrt{3}}{2}, -\frac{1}{2}
			\end{pmatrix}^T & u_2^{s,N}(C_{0,k}) &= \sin(\frac{2ks\pi}{N})\begin{pmatrix}
				\frac{\sqrt{3}}{2},\frac{1}{2}
			\end{pmatrix}^T,\label{eqn:u_2}
		\end{align}
		where $0 \leq s \leq \lfloor \frac{N}{2} \rfloor$.
		
		\item By varying $s$, we get $N$ linearly independent Fleck-Hutchinson modes of the form $u_1^{s,N}(x)$ or $u_2^{s,N}(x)$. In fact, these $N$ Fleck-Hutchinson modes provide a basis for the space of $N$-by-one periodic GH modes (see Proposition \ref{app-d:prop-basis}).
		
		\item The preceding arguments apply equally to the other two families of Fleck-Hutchinson modes mentioned in the second bullet. All three families contain the $1$-periodic GH mode, but aside from this they are linearly independent. As a result, the three families taken together provide a basis for the entire $(3N-2)$-dimensional  space of $N$-periodic GH modes (see Remark \ref{app-d-remark}). 
	\end{itemize}
	
	From the last two bullets, we know that there must be some linear combinations of Fleck-Hutchinson modes that come from mechanisms, since there are $N$-periodic mechanisms (see e.g. \cite{kapko2009collapse} and sections \ref{sec:two-periodic-mechanism} and \ref{subsec:non-periodic} of this paper) and their infinitesimal versions are linear combinations of Fleck-Hutchinson modes. But our focus here is different: we want to know whether these special basis elements themselves come from mechanisms. The answer is mostly negative. In fact, we shall show that:
	\begin{enumerate}[(1)]
		\item When $s=0$, $u_1^{0,N}(x)$ is one-periodic and comes from the one-periodic mechanism, as shown in Figure \ref{fig:one-periodic-mechanism}(a); $u_2^{0,N}(x)$ vanishes.
		
		\item When $s \geq 1$ and $N$ is odd, neither $u_1^{s,N}(x)$ nor $u_2^{s,N}(x)$ comes from a mechanism; when $s \geq 1$ and $N$ is even, the same conclusion applies except for $u_1^{\frac{N}{2},N}(x)$, which does come from a mechanism. This special Fleck-Hutchinson mode $u_1^{\frac{N}{2},N}(x)$ is two-by-one periodic, and it comes from the two-by-one periodic mechanism shown in Figure \ref{fig:two-periodic-special}(a).
		
		\item When $s \geq 1$, a linear combination of the two Fleck-Hutchinson modes with the same Bloch-type boundary condition, i.e. $t_1u_1^{s,N}(x) + t_2u_2^{s,N}(x)$ with the same $s$, almost never comes from a mechanism: (1) when $\frac{N}{4} \notin \mathbb{Z}$, a non-zero $t_1u_1^{s,N}(x) + t_2u_2^{s,N}(x)$ never comes from a mechanism; (2) when $\frac{N}{4} \in \mathbb{Z}$, a non-zero $t_1u_1^{s,N}(x) + t_2u_2^{s,N}(x)$ comes from a mechanism if and only if $s = \frac{N}{4}$ and $|t_1| = |t_2|$. Moreover, these special linear combinations are four-by-one periodic, and they come from four-by-one periodic mechanisms.
	\end{enumerate}
	
	The rest of this section is devoted to proving these assertions. Assertion (1) is straightforward: it is easy to check that when $s=0$, $u^{0,N}_1(x)$ comes from the one-periodic mechanism and $u^{0,N}_2(x)$ vanishes. Turning to assertion (2): the proof uses the consistency condition associated with the horizontal lines of the Kagome lattice.
	\begin{figure}[!htb]
		\begin{minipage}[b]{.48\linewidth}
			\centering
			\subfloat[]{\includegraphics[width=0.8\linewidth]{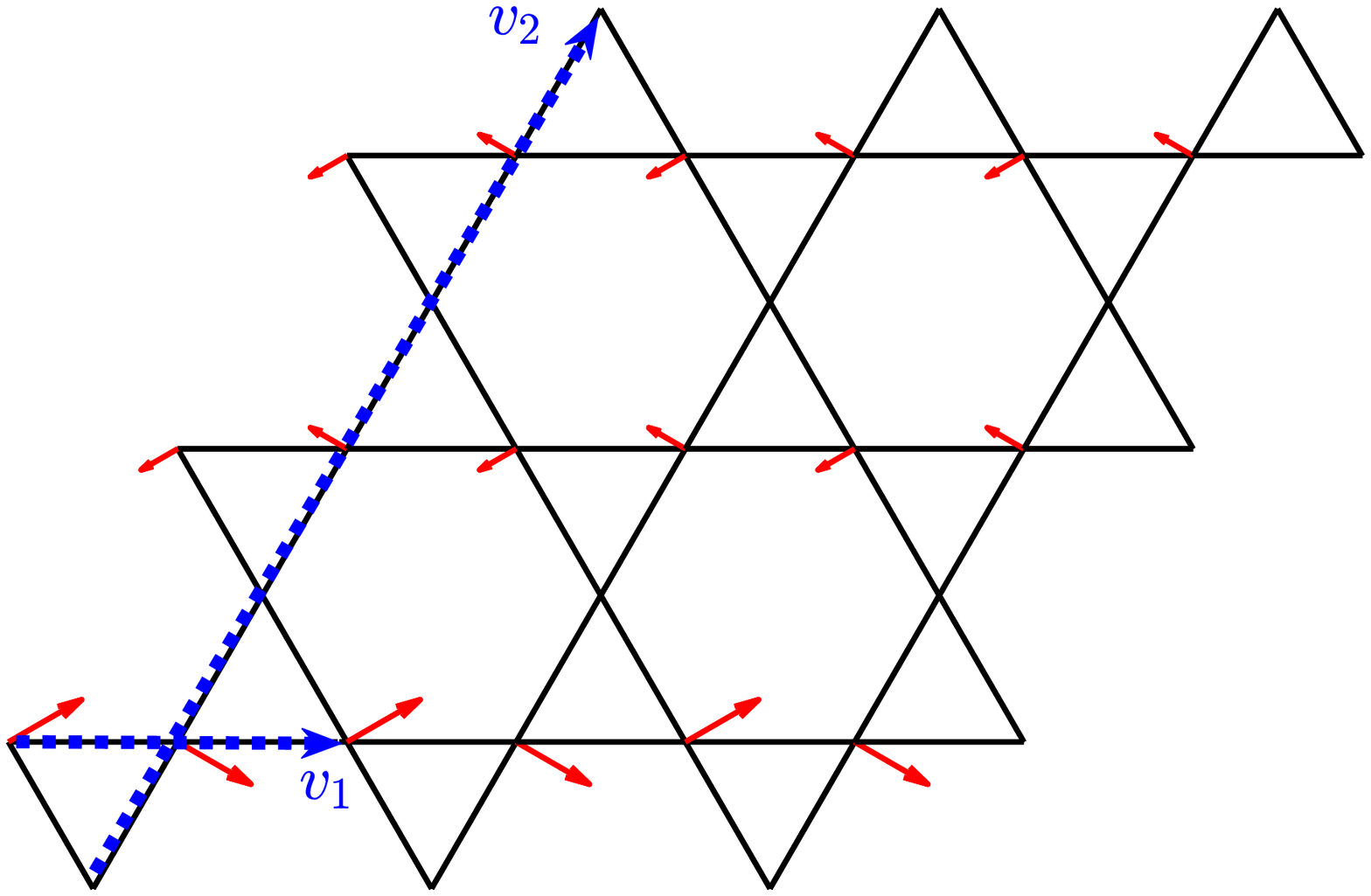}}
		\end{minipage}
		\begin{minipage}[b]{.48\linewidth}
			\centering
			\subfloat[]{\includegraphics[width=0.8\linewidth]{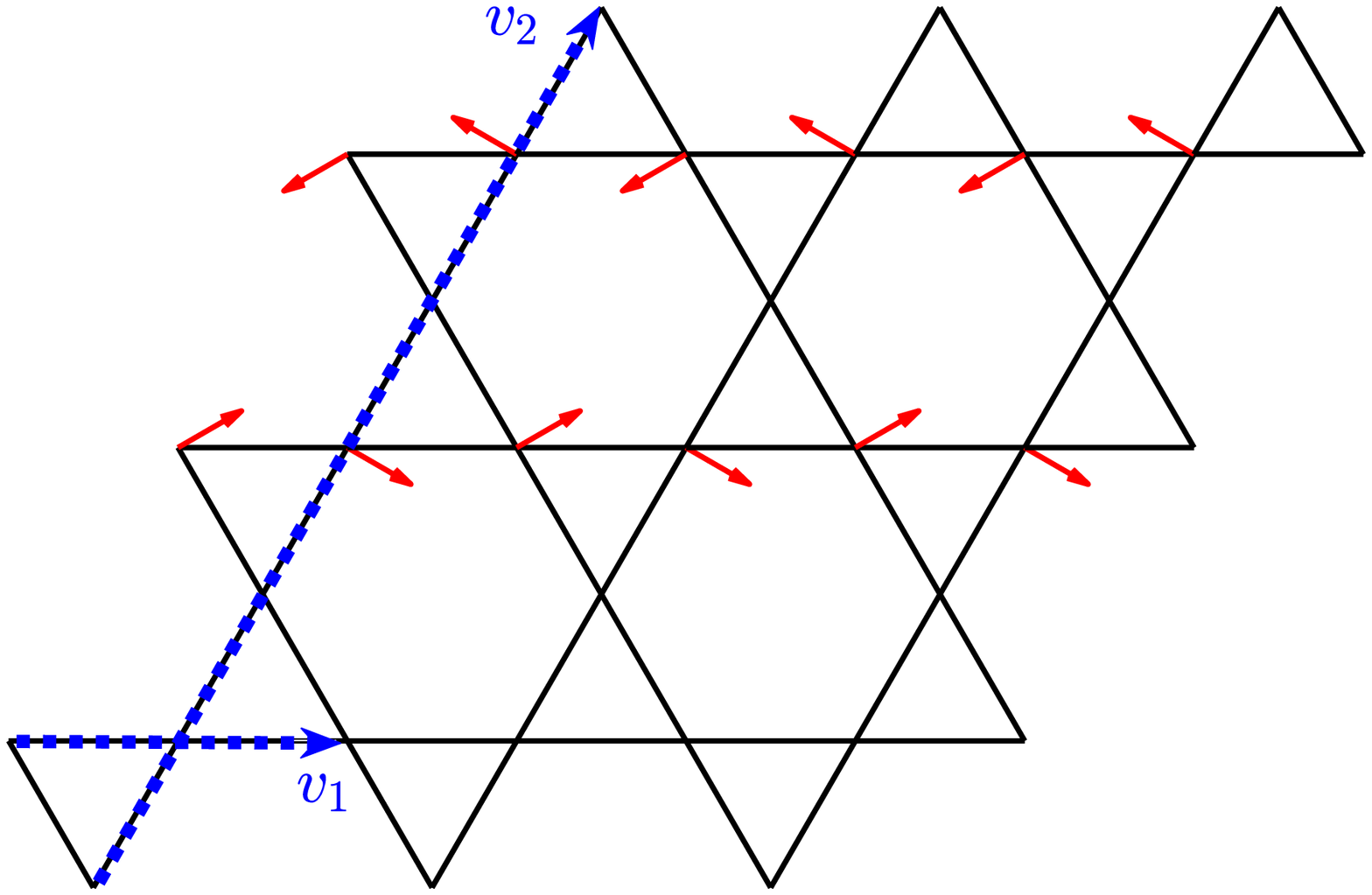}}
		\end{minipage}\par \medskip
		\centering
		\begin{minipage}[b]{.48\linewidth}
			\centering
			\subfloat[]{\includegraphics[width=0.75\linewidth]{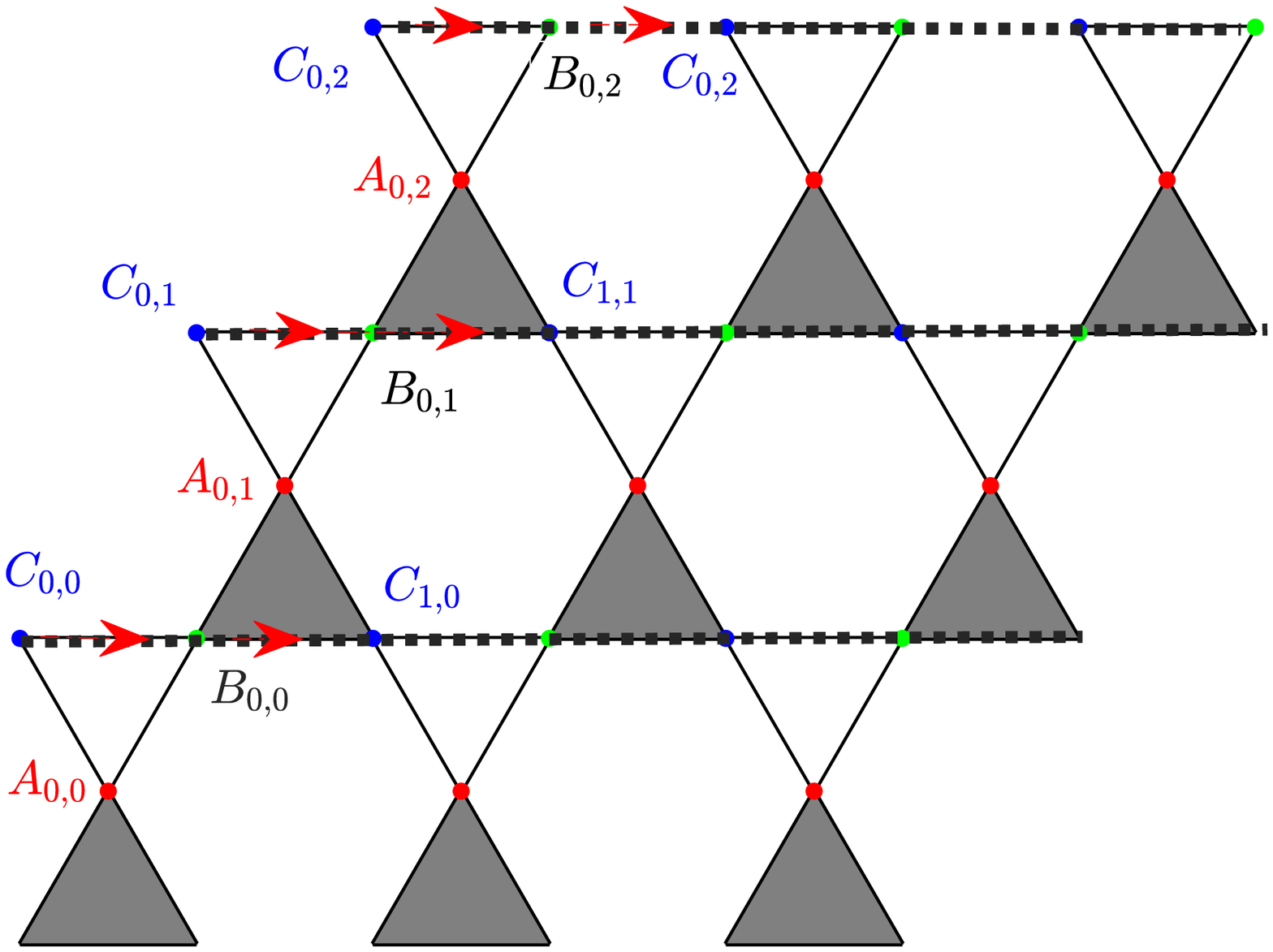}}
		\end{minipage}
		\caption{The three-by-one Fleck-Hutchinson modes and the three-by-one periodic mechanism: (a) the Fleck-Hucthinson mode $u_1^{1,3}(x)$; and (b) the Fleck-Hutchinson mode $u_2^{1,3}(x)$. We can see clearly that the two Fleck-Hutchinson modes are one-periodic in the horizontal direction; (c) the consistency condition checks the sum of quadratic parts over the two springs marked by arrows in each horizontal line.}
		\label{fig:three-by-one-consistency}
	\end{figure}
	To apply the consistency condition to $u_1^{s,N}(x)$ and $u_2^{s,N}(x)$ with fixed $s$, we need to check whether their quadratic parts sum to the same amount on each horizontal line.  Since $u_1^{s,N}(x)$ and $u_2^{s,N}(x)$ are $N$-by-1 periodic, we only need to check the quadratic part on a sum over two springs for each horizontal line. Let us take $u_1^{s,N}(x)$ when $N=3$ as an illustrative example. The consistency condition involves the sum of the quadratic part of $u_1^{s,3}(x)$ over 6 springs in each horizontal line (the 6 springs in each horizontal line are dotted in Figure \ref{fig:three-by-one-consistency}(c)). But due to the periodicity of $u_1^{s,3}(x)$ and $u_2^{s,3}(x)$ in the horizontal direction, the quadratic part over the six springs is a recurrence of the two springs marked by arrows in Figure \ref{fig:three-by-one-consistency}(c).  So the consistency condition for $u_1^{s,3}(x)$ requires us to check whether the following three terms are the same (vertices used below are as marked in Figure \ref{fig:three-by-one-consistency}(c))
	\begin{align*}
		\big|u_1^{s,3}(C_{0,0}) - u_1^{s,3}(B_{0,0})\big|^2 + \big|u_1^{s,3}(B_{0,0}) - u_1^{s,3}(C_{1,0})\big|^2,\\
		\big|u_1^{s,3}(C_{0,1}) - u_1^{s,3}(B_{0,1})\big|^2 + \big|u_1^{s,3}(B_{0,1}) - u_1^{s,3}(C_{1,1})\big|^2,\\
		\big|u_1^{s,3}(C_{0,2}) - u_1^{s,3}(B_{0,2})\big|^2 + |u_1^{s,3}(B_{0,2}) - u_1^{s,3}(C_{1,2})\big|^2.
	\end{align*}
	The two terms in each line are the same because $u_1^{s,3}(x)$ is three-by-one periodic, i.e. $u_1^{s,3}(C_{0,k}) = u_1^{s,3}(C_{1,k})$ for $k=\{0,1,2\}$. Therefore, we need only check whether the following three terms are the same:
	\begin{align*}
		2 \big|u_1^{s,3}(C_{0,0}) - u_1^{s,3}(B_{0,0})\big|^2, \qquad 2\big|u_1^{s,3}(C_{0,1}) - u_1^{s,3}(B_{0,1})\big|^2, \qquad 2\big|u_1^{s,3}(C_{0,2}) - u_1^{s,3}(B_{0,2})\big|^2.
	\end{align*}
	Similarly to the $N=3$ case, for a general $N$ and a fixed $s$, we need only check whether the following $N$ terms are the same
	\begin{align*}
		2\big|u_1^{s,N}(C_{0,0}) - u_1^{s,N}(B_{0,0})\big|^2,\qquad
		2\big|u_1^{s,N}(C_{0,1}) - u_1^{s,N}(B_{0,1})\big|^2,\qquad \dots, \qquad 
		2\big|u_1^{s,N}(C_{0,N-1}) - u_1^{s,N}(B_{0,N-1})\big|^2,
	\end{align*}
	where $k = 0,1,\dots, N-1$. From the explicit form \eqref{eqn:u_1} of $u_1^{s,N}(x)$, we see that 
	\begin{align*}
		\big|u_1^{s,N}(C_{0,k}) - u_1^{s,N}(B_{0,k})\big|^2 &= \cos^2(\frac{2ks\pi}{N}) \big|u_1^{0,N}(C_{0,0}) - u_1(B_{0,0})\big|^2 \neq \big|u_1^{0,N}(C_{0,0}) - u_1(B_{0,0})\big|^2,
	\end{align*}
	unless $\cos^2(\frac{2ks\pi i}{N})  = 1$ for every $k=0,1,\dots,N-1$. This is only true when $s = \frac{N}{2}$ since $1 \leq s \leq \frac{N}{2}$. Thus, the consistency condition in the horizontal direction only holds for $u_1^{s,N}(x)$ if and only if $N$ is even and $s= \frac{N}{2}$. By a similar calculation, the consistency condition in the horizontal direction for $u_2^{s,N}(x)$ checks whether the following $N$ terms are the same
	\begin{align}
		\big|u_2^{s,N}(C_{0,k}) - u_2^{s,N}(B_{0,k})\big|^2 &= \sin^2(\frac{2ks\pi}{N}) \big|u_1^{s,N}(C_{0,0}) - u_1^{s,N}(B_{0,0})\big|^2, \label{eqn:consistency-N-2}
	\end{align}
	for every $k = 0,1,\dots, N-1$. However, the consistency condition never holds for non-trivial $u_2^{s,N}(x)$. To see why, we observe that the term in \eqref{eqn:consistency-N-2} vanishes when $k=0$, but it does not vanish when $k = 1$ unless $s = \frac{N}{2}$. However, when $N$ is even, the special mode $u_2^{\frac{N}{2}, N}(x) = 0$. Thus, among all $u_1^{s,N}(x), u_2^{s,N}(x)$, only the special Fleck-Hutchinson mode $u_1^{\frac{N}{2},N}(x)$ satisfies the consistency condition with even $N$. Moreover, the special mode $u_1^{\frac{N}{2},N}(x)$ is two-by-one periodic (for any even $N$) and it comes from the two-by-one periodic mechanism, as shown in Figure \ref{fig:two-periodic-special}(a) (see the end of Appendix \ref{appendix-c} for a detailed discussion).
	
	We turn now to our assertion (3), which addresses whether a linear combination of the two non-zero Fleck-Hutchinson modes with the same Bloch-type boundary condition, i.e. $u_1^{s,N}(x)$ and $u_2^{s,N}(x)$ with the same $s$ and $N$, comes from a mechanism. The answer is yes when $\frac{N}{4} \in \mathbb{Z}$, but no for all other cases. We shall show that the consistency condition is satisfied in the case $\frac{N}{4} \in \mathbb{Z}$ \textit{only} at two special linear combinations $u_1^{s,N}(x) \pm u_2^{s,N}(x)$ (or a scaled version of them). Let us consider $u(x) = t_1u_1^{s,N}(x) + t_2 u_2^{s,N}(x)$ with the same $s,N$ and constrain $t_1t_2 \neq 0$ to eliminate the case where one of $t_1,t_2$ vanishes. The consistency condition requires that the quadratic part of $t_1 u_1^{s,N}(x) + t_2 u_2^{s,N}(x)$ sums to the same amount on the two springs in every horizontal line, i.e. 
	\begin{align}
		e_{k}^{s,N} = & 2\big|t_1 \big(u_1^{s,N}(C_{0,k}) - u_1^{s,N}(B_{0,k})\big) + t_2 \big(u_2^{s,N}(C_{0,k}) - u_2^{s,N}(B_{0,k})\big)\big|^2 \label{eqn:quadratic_ek}
	\end{align}
	must be the same for $k = 0,1,\dots,N-1$. We observe from \eqref{eqn:u_1}-\eqref{eqn:u_2} that for any $s$ and $N$,
	\begin{align*}
		u_1^{s,N}(B_{0,0}) &= \begin{pmatrix}
			\frac{\sqrt{3}}{2}, -\frac{1}{2}
		\end{pmatrix}^T & u_1^{s,N}(C_{0,0}) &= \begin{pmatrix}
			\frac{\sqrt{3}}{2}, \frac{1}{2}
		\end{pmatrix}^T\\
		u_1^{s,N}(B_{0,k}) &= \cos(\frac{2ks\pi}{N}) u_1^{s,N}(B_{0,0}) & u_1^{s,N}(C_{0,k}) &= \cos(\frac{2ks\pi}{N}) u_1^{s,N}(C_{0,0})\\
		u_2^{s,N}(B_{0,k}) &= \sin(\frac{2ks\pi}{N})u_1^{s,N}(B_{0,0}) & u_2^{s,N}(B_{0,k}) &= \sin(\frac{2ks\pi}{N})u_1^{s,N}(C_{0,0}).
	\end{align*}
	Using this relationship, the quadratic part $e_{k}^{s,N}$ in \eqref{eqn:quadratic_ek} becomes
	\begin{align*}
		e_k^{s,N} &= 2\bigg(t_1 \cos(\frac{2ks\pi}{N})+ t_2 \sin(\frac{2ks\pi}{N})\bigg)^2 \big|u_1^{s,N}(C_{0,0}) - u_1^{s,N}(B_{0,0})\big|^2.
	\end{align*}
	Therefore, the consistency condition requires that $\bigg(t_1 \cos(\frac{2ks\pi}{N})+ t_2 \sin(\frac{2ks\pi}{N})\bigg)^2$ be the same for all $k=0,1,\dots,N-1$. We shall show that this can hold only when $s = \frac{N}{4}$ or $s = \frac{N}{2}$. To see why, we first take $k=1$ and $k = N-1$, and observe that the equality $e_1^{s,N} = e_{N-1}^{s,N}$ gives 
	\begin{align*}
		\Big(t_1 \cos(\frac{2s\pi}{N})+ t_2 \sin(\frac{2s\pi}{N})\Big)^2 &= \Big(t_1 \cos(\frac{2s(N-1)\pi}{N})+ t_2 \sin(\frac{2s(N-1)\pi}{N})\Big)^2\\
		&= \Big(t_1 \cos(\frac{2s\pi}{N})- t_2 \sin(\frac{2s\pi}{N})\Big)^2.
	\end{align*}
	We conclude that $4 t_1 t_2 \cos(\frac{2s\pi}{N}) \sin(\frac{2s\pi}{N}) = 2t_1 t_2 \sin(\frac{4s\pi}{N}) = 0$. Due to our constraint $t_1 t_2 \neq 0$, we must have $\sin(\frac{4s\pi}{N}) = 0$, i.e. $\frac{4s}{N} \in \mathbb{Z}$. Since $1 \leq s \leq \frac{N}{2}$, the consistency condition is only satisfied when $s = \frac{N}{4}$ or $s = \frac{N}{2}$. When $s = \frac{N}{2}$, this becomes the special case we discussed earlier where $u_1^{\frac{N}{2}, N}(x)$ is two-periodic and $u_2^{\frac{N}{2}, N}(x)$ vanishes. When $s = \frac{N}{4}$, the consistency condition requires $\bigg(t_1 \cos(\frac{k\pi}{2})+ t_2 \sin(\frac{k\pi}{2})\bigg)^2$ to be the same for all $k=0,1,\dots, N-1$. This condition is satisfied if and only if $t_1^2 = t_2^2$. As a summary, for the two non-zero Fleck-Hutchinson modes $u_1^{\frac{N}{2}, N}(x), u_2^{\frac{N}{2}, N}(x)$ with the same $s,N$, a linear combination $t_1 u_1^{s,N}(x) + t_2 u_2^{s,N}(x)$ satisfies the consistency condition if and only if $s = \frac{N}{4} \in \mathbb{Z}$ and $t_1 = \pm t_2$. Moreover, the two special linear combinations $u_1^{s,N}(x) \pm u_2^{s,N}(x)$ when $s = \frac{N}{4}$ come from four-by-one periodic mechanisms (see Appendix \ref{appendix-f-2}).
	
	\section{$N$-by-one periodic and non-periodic mechanisms}\label{subsec:non-periodic}
	In section \ref{subsection:fleck-hutchinson}, we have seen that while the Fleck-Hutchinson modes $u_1^{s,N}(x)$ and $u_2^{s,N}(x)$ are $N$-by-one periodic, they are usually not associated with $N$-by-one periodic mechanisms. Thus, the use of Fleck-Hutchinson modes is not an efficient means of finding nonlinear mechanisms. This section offers an entirely different approach to understanding the $N$-by-one periodic mechanisms of the standard Kagome lattice. Our approach, which is based on layering, is relatively simple; moreover, besides providing a classification of all $N$-by-one periodic mechanisms, it also provides examples of non-periodic mechanisms.
	
	We start by considering the four-by-one periodic mechanism shown in Figure \ref{fig:non-periodic-four}(a) (see Appendix \ref{appendix-f-2} for its details), which provides the building blocks we shall use to construct $N$-by-one periodic mechanisms. There are four layers in this four-by-one periodic mechanisms, namely $G_1, G_2, W_1, W_2$. The two shaded layers $G_1, G_2$ are one-periodic layers achieving the same compression ratio $\frac{1}{2} \leq c \leq 1$ (we choose $\frac{1}{2}$ as the lower bound to avoid triangles intersecting each other). Note that they rotate the triangles in their unit cells in opposite directions. The two unshaded layers $W_1, W_2$ come from the two-by-one periodic mechanism in Figure \ref{fig:two-periodic-special}(a), and they must achieve the same compression ratio $c$ to fit the one-periodic layers. We observe from the four-by-one periodic mechanism that the four layers fit each other perfectly in a corresponding relationship: a $W_1$ layer fits above a $G_1$ layer, a $G_2$ layer fits above a $W_1$ layer, etc (see Figure \ref{fig:non-periodic-four}(a)). We also know from the one-periodic and the two-by-one periodic mechanism that a $G_1$ layer fits above a $G_1$ layer (same for $G_2$) and a $W_1$ layer fits above a $W_2$ layer (same for $W_2$). Therefore, we summarize the layering relationship for the four basic layers $G_1, G_2, W_1, W_2$ as
	\begin{align}
		G_1 &\begin{casesnew}[-stealth,black]
			G_1\\
			W_1
		\end{casesnew}, & G_2 & \begin{casesnew}[-stealth,black]
			G_2\\
			W_2
		\end{casesnew}, & W_1 & \begin{casesnew}[-stealth,black]
			G_2\\
			W_2
		\end{casesnew}, & W_2 & \begin{casesnew}[-stealth,black]
			G_1\\
			W_1
		\end{casesnew}, \label{eqn:layering-relationship}
	\end{align}
	where the arrow in $W_1 \rightarrow W_2$ means $W_2$ fits above $W_1$.
	\begin{figure}[!htb]
		\begin{minipage}[b]{.33\linewidth}
			\centering
			\subfloat[]{\includegraphics[width=0.85\linewidth]{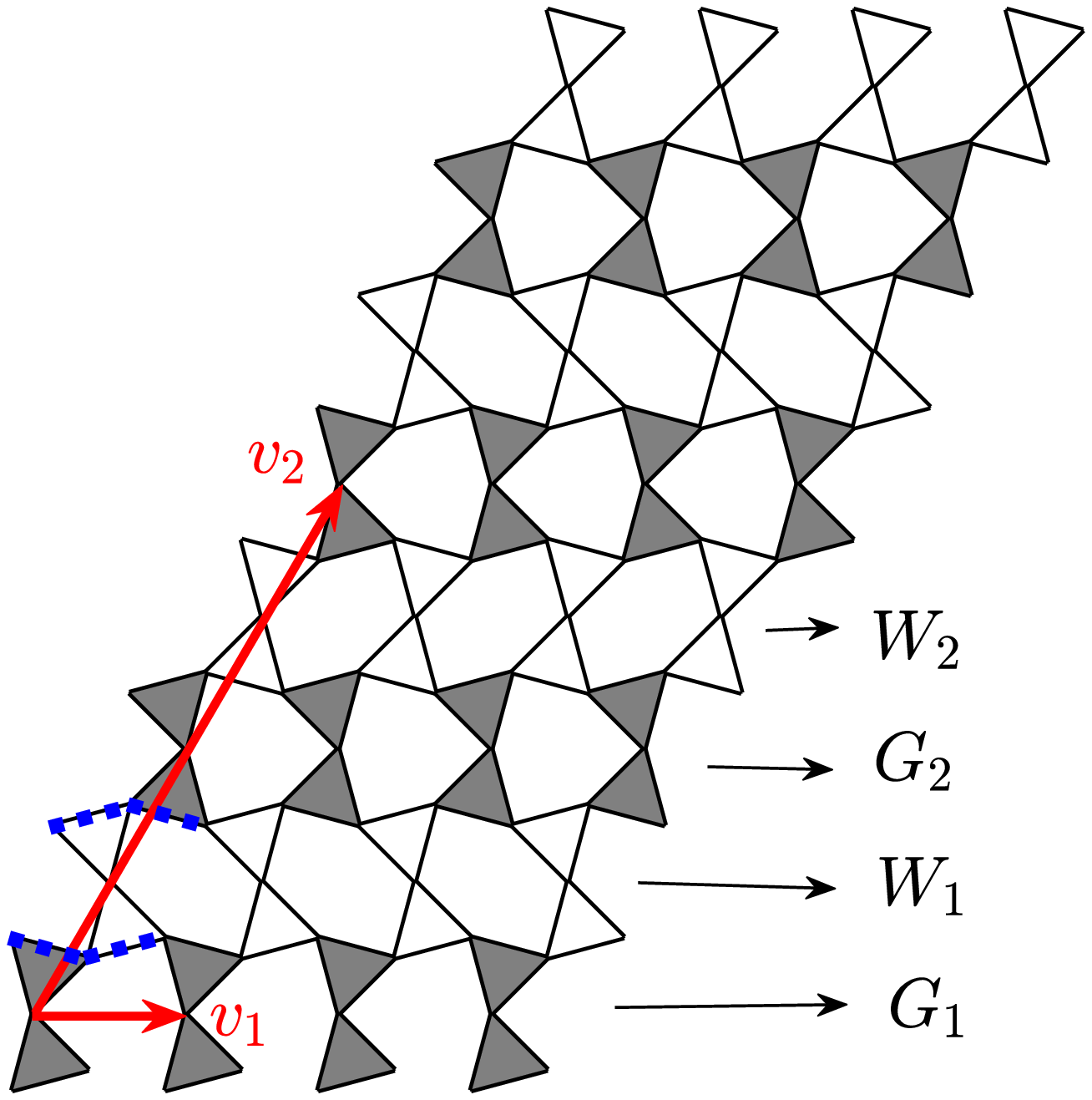}}
		\end{minipage}
		\begin{minipage}[b]{.33\linewidth}
			\centering
			\subfloat[]{\includegraphics[width=0.85\linewidth]{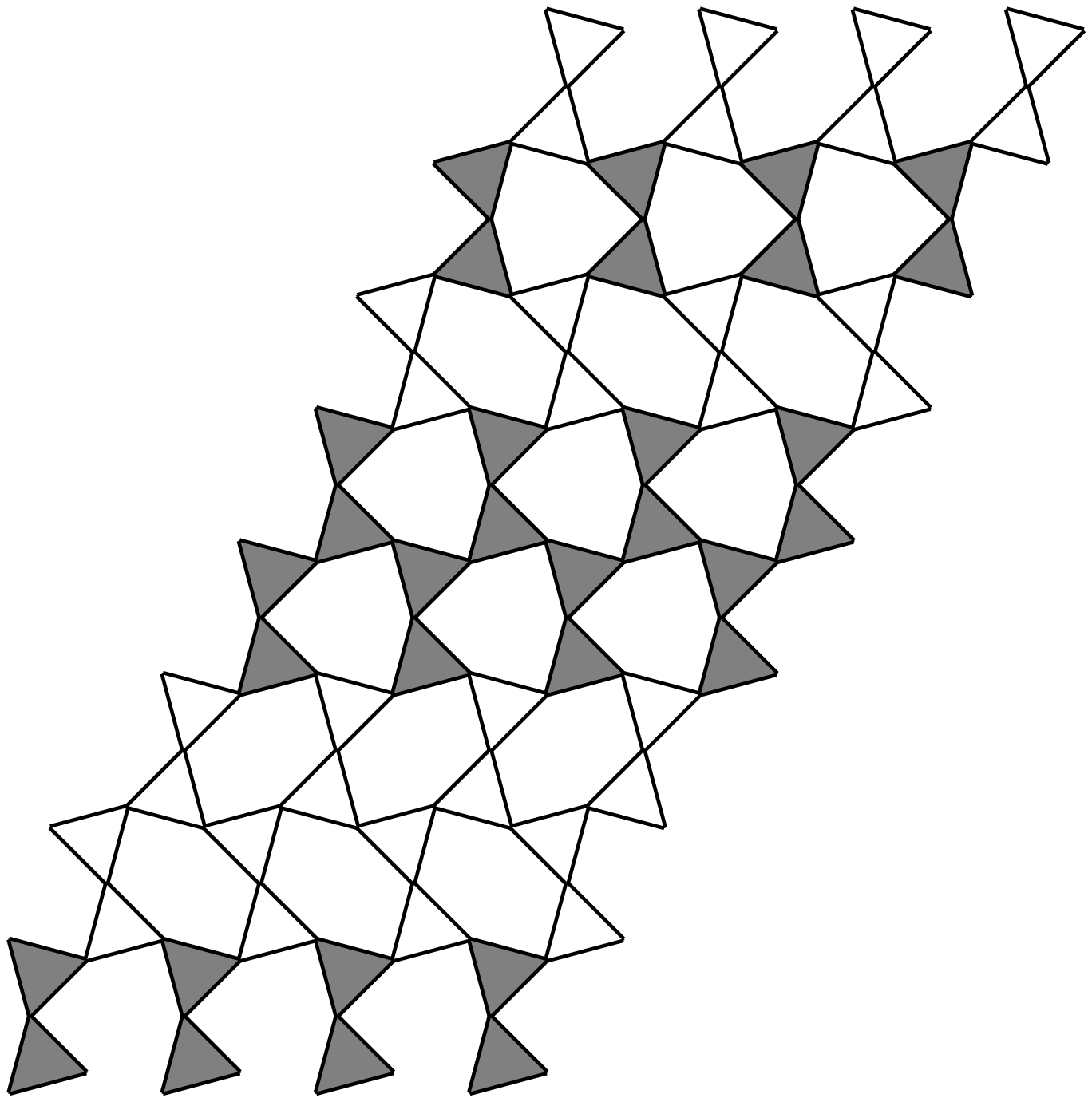}}
		\end{minipage}
		\begin{minipage}[b]{0.33\linewidth}
			\centering
			\subfloat[]{\includegraphics[width=0.85\linewidth]{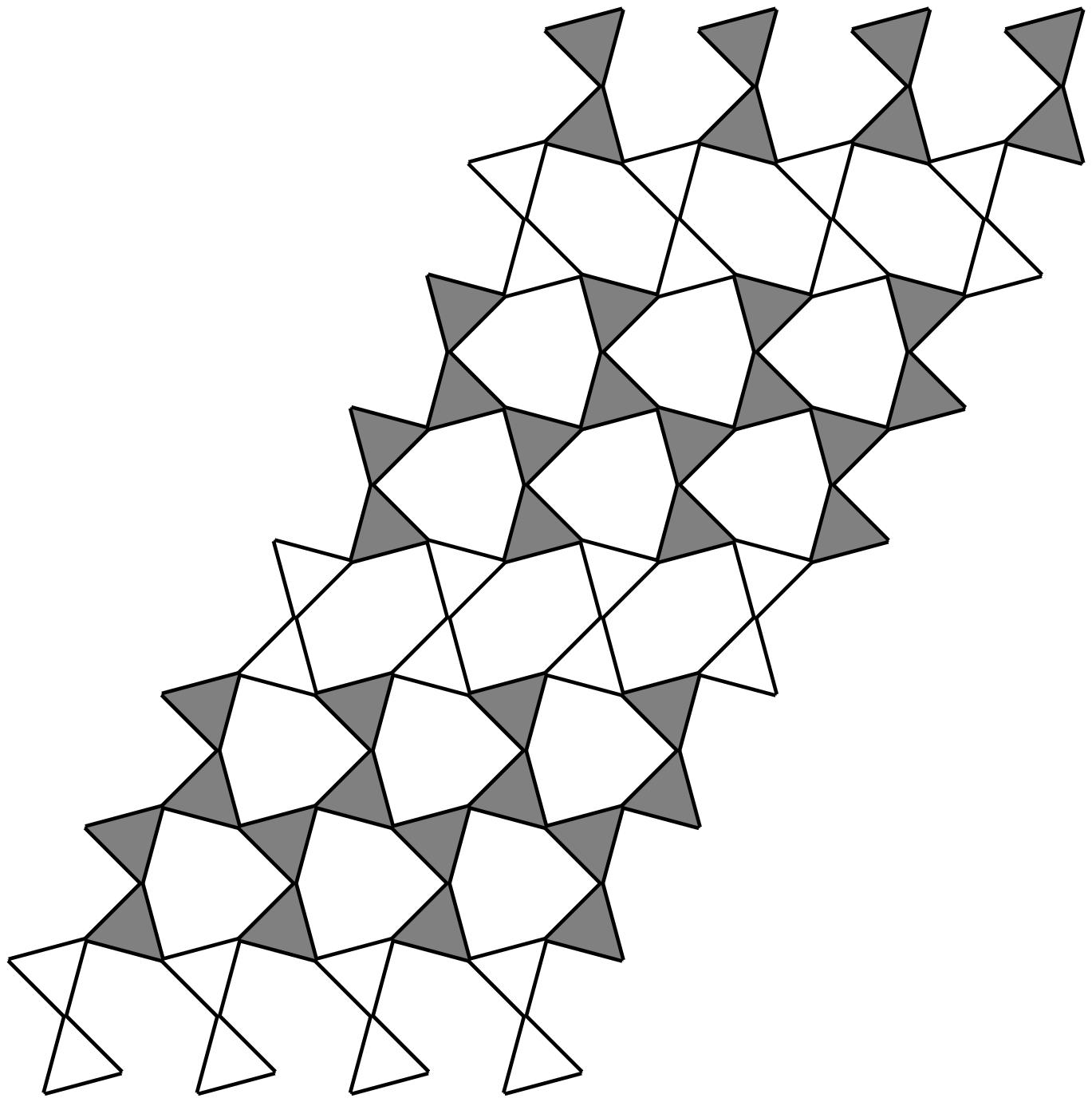}}
		\end{minipage}
		\caption{A layering scheme and all mechanisms achieve a macroscopic compression ratio $c = \cos(\frac{\pi}{12})$: (a) the four-by-one periodic mechanism; (b) a mechanism achieved by the sequence $\{G_1, W_1,W_2,G_1,G_1,W_1,G_2,W_2,\dots\}$; (c) another mechanism achieved by the sequence $\{W_1,G_2,G_2,W_2,G_1,G_1,W_1,G_2,\dots\}$.}
		\label{fig:non-periodic-four}
	\end{figure}
	
	A geometric explanation of the layering relationship \eqref{eqn:layering-relationship} is that the zigzag lines in the horizontal direction can deform only in two ways (marked as dotted) in Figure \ref{fig:non-periodic-four}(a). Each zigzag line must have a symmetric wedge pointing either upwards or downwards. Take the zigzag line formed by $G_1$ and $W_1$ layers as an example. The dotted line as an edge of the gray triangle has a negative slope. Hence, the layer above it must provide a deformed line with a positive slope. This is why $G_1$ and $W_1$ layers fit above a $G_1$ layer. The other layering relationship holds for the same reason.
	
	Using this relationship between layers, we can construct many $N$-by-one periodic mechanisms for any $N$, and also many non-periodic mechanisms. In fact, we can choose a sequence $\{a_n\}_{n \in \mathbb
		N}$ with $a_n \in \{G_1, G_2, W_1, W_2\}$. This sequence must satisfy the layering relationship in \eqref{eqn:layering-relationship} with $\color{black}a_{n} \rightarrow a_{n+1}$ for every $n \in \mathbb{N}$. We can construct a mechanism $u_{a_n}(x)$ based on the sequence $\{a_n\}_{n \in \mathbb{N}}$ by stacking the corresponding $a_n$ in the $n$th layer (see Figure \ref{fig:non-periodic-four}(b)-(c) for examples). If $\{a_n\}_{n \in \mathbb{N}}$ is $N$-periodic, then the corresponding mechanism $u_{a_n}(x)$ is $N$-by-one periodic. If the sequence $\{a_n\}_{n \in \mathbb{N}}$ is not periodic, then the mechanism $u_{a_n}(x)$ is non-periodic.
	
	Our argument actually finds \textit{all} $N$-by-one periodic mechanisms with period 1 in the horizontal direction. To explain why, we observe that for any $N$-by-one periodic mechanism with period 1 in the horizontal direction, we can separate it into a sequence of layers as we move along the 60 degree direction. Due to being period 1 in the horizontal direction, these mechanisms must deform the horizontal lines into zigzags with symmetric wedges pointing upwards or downwards, as shown in Figure \ref{fig:non-periodic-four}(a). It is easy to check that such zigzag lines can only be achieved by the one-periodic mechanisms and two-by-one periodic mechanisms that are our building blocks. Therefore, any $N$-by-one periodic mechanism with period 1 in the horizontal direction arises from our layering procedure.
	
	\section{A special case: some Maxwell lattices must have mechanisms} \label{sec:open-questions}
	{\color{black}We have used the GH modes of the standard Kagome lattice with different periodicities as examples to show that not every GH mode in a Maxwell lattice comes from a periodic mechanism. We also note that not every Maxwell lattice has a mechanism: Borcea and Streinu \cite{borcea2010periodic} found a 2D Maxwell lattice with overlapping springs that has no mechanisms. These observations make us wonder whether there is a sufficient condition on a Maxwell lattice such that every GH mode of it must come from a mechanism.}
	
	{\color{black}This section offers a result of this type, for a rather special class of Maxwell lattices. Briefly, we show that a non-degenerate 2D Maxwell lattice must have a mechanism if its space of GH modes is one-dimensional (the one-periodic standard Kagome lattice satisfies this assumption).  In fact, we shall prove the following Proposition.}
	\begin{proposition}
		For a non-degenerate 2D Maxwell lattice, if $\varphi_1(x)$ is the only GH mode (up to scalar multiplication), then we can find a mechanism parameterized by $t$ in the form
		\begin{align}
			u(x,t) = x + t \varphi_1(x) + t^2\big[\xi_2 x + \varphi_2(x)\big] + \dots, \label{eqn:sec-6-ift-ansatz}
		\end{align}
		where $\xi_2, \varphi_2(x)$ satisfies the necessary condition \eqref{eqn:necessary_condition_original}.
	\end{proposition}
	
	\begin{proof}
		The proof proceeds in two steps: (1) We show that $\xi_2, \varphi_2(x)$ are determined \textit{uniquely} by $\varphi_1(x)$; (2) We prove that there really is a $u(x,t)$ with this leading order Taylor expansion using the implicit function theorem.
		
		We first show how $\xi_2, \varphi_2(x)$ can be determined \textit{uniquely} by $\varphi_1(x)$. To determine $\xi_2, \varphi_2(x)$, we know that $u(x,t)$ must keep the lengths of all springs at second order, i.e. \eqref{eqn:necessary_condition_original} and \eqref{eqn:necessary_condition} hold for every spring. Multiplying $L^{-1}$ to both sides of \eqref{eqn:necessary_condition}, we get
		\begin{align}
			L^{-1}\bm{e}_{\bm{\varphi_1}} + L^{-1}\bm{d}_{\xi_2} = C \bm{\varphi_2},
		\end{align}
		where each entry of $L^{-1}\bm{d}_{\xi_2}$ is in the form $l_{ij} \hat{b}_{ij}^T \xi_2 \hat{b}_{ij}$. Notice that the vector form of $l_{ij} \hat{b}_{ij}^T \xi_2 \hat{b}_{ij}$ is in fact $\bm{b}_{\xi_2}$ (see \eqref{linear_algebra_effective_energy}), i.e. $L^{-1}\bm{d}_{\xi_2} = \bm{b}_{\xi_2}$. The symmetric matrix $\xi_2$ has three degrees of freedom and can be found uniquely using self-stresses. We start by choosing three linearly independent self-stresses $\bm{s}_1, \bm{s}_2, \bm{s}_3$, i.e. $C^T\bm{s}_i = 0$ for $i=1,2,3$. We can choose them because $\ker(C) = \ker(C^T) = \dim\{\text{GH modes}\}+2 = 3$. Left multiplying $\bm{s}_i$ to both sides of \eqref{eqn:necessary_condition}, we achieve
		\begin{align}
			\bm{s}_i^T L^{-1}\bm{e}_{\bm{\varphi_1}} + \bm{s}_i^T  \bm{b}_{\xi_2} =  0, \qquad i = 1,2,3. \label{eqn:appendix-g-linear-system}
		\end{align}
		This system {\color{black}for} $\xi_2$ has three equations and three degrees of freedom. We claim that it always has a unique solution based on our non-degenerate assumption. Suppose the system does not have a unique solution, then there must exist a non-trivial $\xi_2^0 \neq 0$ such that $\bm{s}_i^T  \bm{b}_{\xi_2^0} =  0, i = 1,2,3$. This indicates that for any self-stress $\bm{s}$ such that $C^T \bm{s} = 0$, we must have
		\begin{align}
			\bm{s}^T  \bm{b}_{\xi_2^0} =  0. \label{eqn:appendix-g-self-stress}
		\end{align}
		We shall show that $E_{\text{eff}}(\xi_2^0) = \langle A_{\text{eff}}\xi_2^0, \xi_2^0\rangle = 0$. We start by choosing the self-stress $\bm{t}_{\xi_2^0}^* = K (\bm{b}_{\xi_2^0} + C \bm{\varphi}_{\xi_2^0}^*)$ that corresponds to the macroscopic stress $A_{\text{eff}} \xi_2^0$ in \eqref{tension_vector_form}. By plugging this self-stress $\bm{t}_{\xi_2^0}^*$ into \eqref{eqn:appendix-g-self-stress}, we get
		\begin{align}
			0 &= (\bm{t}_{\xi_2^0}^*)^T \bm{b}_{\xi_2^0} =\bm{b}_{\xi_2^0}^T K \bm{b}_{\xi_2^0} + (\varphi_{\xi_2^0}^*)^T C^T K \bm{b}_{\xi_2^0}.
			\label{eqn:appendix-g-vanishing-stress}
		\end{align}
		We have another equality from the optimal condition for $\varphi_{\xi_2^0}^*$ in \eqref{optimal_phi_condition}. Left multiplying $\bm{\varphi}_{\xi_2^0}^*$ to both sides of \eqref{optimal_phi_condition}, we get
		\begin{align}
			0 &=(\bm{\varphi}_{\xi_2^0}^*)^T C^T K C \bm{\varphi}_{\xi_2^0}^* + C^T K \bm{b}_{\xi_2^0}. \label{eqn:optimal-condition}
		\end{align}
		Combining \eqref{eqn:appendix-g-vanishing-stress}- \eqref{eqn:optimal-condition} and using the formula of $E_{\text{eff}}(\xi)$ in \eqref{effective_energy_appendix}-\eqref{linear_algebra_effective_energy}, we get
		\begin{align*}
			E_{\text{eff}}(\xi_2^0) = \langle A_{\text{eff}}\xi_2^0, \xi_2^0\rangle = (\bm{b}_{\xi_2^0} + C \bm{\varphi}_{\xi_2^0}^*)^T K (\bm{b}_{\xi_2^0} + C \bm{\varphi}_{\xi_2^0}^*) = 0.
		\end{align*}
		This violates our assumption that $A_{\text{eff}}$ is non-degenerate. Thus, a unique $\xi_2$ is determined by $\varphi_1(x)$ using \eqref{eqn:appendix-g-linear-system}. The periodic $\varphi_2(x)$ can be found correspondingly by solving \eqref{eqn:appendix-g-linear-system}. Notice that the solution to \eqref{eqn:appendix-g-linear-system} is not unique since the null space of $C$ is 3-dimensional. We can ensure its uniqueness by imposing an extra condition that $\varphi_2(x)$ is orthogonal to the null space of $C$. This extra condition is in fact equivalent to three linear conditions:
		\begin{align*}
			\varphi_2(x) \perp d_1(x), \qquad \varphi_2(x) \perp d_2(x),  \qquad \varphi_2(x) \perp \varphi_1(x),
		\end{align*}
		where $d_1(x), d_2(x)$ are two linearly independent 2-dimensional translations and $\varphi_2(x) \perp d_1(x)$ means the vector forms of $\varphi_2(x), d_1(x)$ are orthogonal. With these three conditions, we can uniquely determine $\varphi_2(x)$ in \eqref{eqn:appendix-g-linear-system}.
		
		Thus far, we have shown how $\xi_2, \varphi_2(x)$ can be uniquely determined by $\varphi_1(x)$, so that the deformation $u(x,t)$ in the form \eqref{eqn:sec-6-ift-ansatz} preserves the lengths of all springs to second order. We now show that there exists $u(x,t)$ preserving the lengths of springs at higher order using the implicit function theorem. In other words, we seek a mechanism 
		\begin{align}
			u(x,t) = x + t \varphi_1(x) + t^2 \Big[\widetilde{\xi}_2(t) x + \widetilde{\varphi}_2(x,t)\Big], \label{eqn:mechanism-ift}
		\end{align}
		where $\widetilde{\xi}_2(0) = \xi_2$ and $\widetilde{\varphi}_2(x,0) = \varphi_2(x)$ are determined uniquely by the given $\varphi_1(x)$. To rigorously show the existence of a mechanism with the form \eqref{eqn:mechanism-ift}, we must show that $\widetilde{\xi}_2(t), \widetilde{\varphi}_2(x,t)$ can be chosen such that
		\begin{align}
			& 0 = \frac{|\widetilde{u}(x_i, t) - \widetilde{u}(x_j, t)|^2 -|x_i - x_j|^2 }{t^2} = |\varphi_1(x_i) - \varphi_1(x_j)\big|^2 + 2\big\langle x_i - x_j, \widetilde{\xi}_2(x_i - x_j)\big\rangle \nonumber\\
			& \qquad \qquad \qquad \qquad \qquad \qquad \qquad \qquad + 2\big\langle x_i - x_j, \widetilde{\varphi}_2(x_i) - \widetilde{\varphi}_2(x_j)\big\rangle + O(t), \label{eqn:appendix-g-constraint}\\
			& \widetilde{\varphi}_2(x) \perp d_1(x), \qquad \widetilde{\varphi}_2(x) \perp d_2(x),  \qquad \widetilde{\varphi}_2(x) \perp \varphi_1(x). \label{eqn:orthonal-condition-varphi-2}
		\end{align}
		For simplicity, we refer to the vector form of $\widetilde{\varphi}_2(x)$ as $\bm{\widetilde{\varphi}_2}$ and the system in \eqref{eqn:appendix-g-constraint}-\eqref{eqn:orthonal-condition-varphi-2} as $F(\widetilde{\xi}_2, \bm{\widetilde{\varphi}_2}, t) = 0$. This system has $e+3$ equations and $2d+3$ degrees of freedom, where $e,d$ are the number of springs and vertices in the unit cell ($2d=e$ in the case of Maxwell lattices). We notice that at $t=0$, these constraints can be satisfied by taking $\widetilde{\varphi}_2 (x)= \varphi_2(x)$ and $\widetilde{\xi}_2 = \xi_2$, i.e. $F(\xi_2, \bm{\varphi}_2, 0) = 0$. It can also be checked that $\partial_{\widetilde{\xi}_2, \bm{\widetilde{\varphi}_2}} F$ evaluated at $\widetilde{\varphi}_2(x) = \varphi_2(x), \widetilde{\xi}_2 = \xi_2, t=0$ is invertible, since $\xi_2, \varphi_2(x)$ are uniquely determined by $\varphi_1(x)$ in \eqref{eqn:appendix-g-constraint} at $t=0$. Thus, the implicit function theorem indicates that there exists $\widetilde{\xi}_2(t), \bm{\widetilde{\varphi}_2}(t)$ around $t=0$ such that $u(x,t) = x + t \varphi_1(x) + t^2 \Big[\widetilde{\xi}_2(t) x + \widetilde{\varphi}_2(x,t)\Big]$ is a mechanism. Evidently, the infinitesimal version of this mechanism is the given GH mode $\varphi_1(x)$.
	\end{proof}
	
	A similar argument seems unavailable when the space of GH modes has dimension greater than one. To briefly explain the difficulty: in general, $\varphi_1$ must satisfy the necessary condition we obtained in section \ref{sec:necessary-condition}. This assures the existence of $\xi_2, \varphi_2$ such that $|u(x_i)-u(x_j))|^2$ vanishes to second order for every connected $x_i,x_j$. But when we continue to the next order, the existence of $\xi_3, \varphi_3$ requires a new necessary condition involving $\varphi_1, \xi_2, \varphi_2$. Though $\varphi_2$ as a solution to \eqref{eqn:necessary_condition} is not unique, we do not see how to use the freedom in $\varphi_2$ to assure the required necessary condition at the next order.
	
	In fact, finding a mechanism amounts to finding a one-parameter family of solutions to a system of quadratic equations, while the existence of a GH mode solves a linearized version of that system. In such a setting, the linearized system does not necessarily contain enough information to know the dimension of the actual solution set. A simple example in 3D shows that a linearized system might not count the nonlinear solutions correctly: if two spheres in 3D only meet at the origin, the intersection of their tangent spaces at the origin is a plane.
	
	\section*{Acknowledgements}
	This research was partially supported by NSF (through grant DMS-2009746) and by the Simons Foundation (through grant 733694). We are grateful to Raghavendra Venkatraman for many useful discussions and questions. We also thank Ian Tobasco for some very helpful comments; {\color{black} and we thank the anonymous referees for their helpful suggestions, which significantly improved this paper.}
	%%%%%%%%%%%%%%%%%%%%%%%%%%%%%%%%%%%%%%%%%%%%%%%%%%%%%%%%%%%%%%%%%%%%%%%%%%%%%%%%%%
	\appendix
	\section*{Appendix}
	\addcontentsline{toc}{section}{Appendix}
	\renewcommand{\thesubsection}{\Alph{subsection}}
	\setcounter{equation}{0}
	\renewcommand{\theequation}{\thesubsection.\arabic{equation}}
	\subsection{The effective Hooke's law}\label{appendix-a}
	The effective linear elastic energy in its variational form is
	\begin{align}\label{effective_energy_appendix}
		E_{\text{eff}}(\xi) &= \frac{1}{S}\min_{\substack{\varphi(x) \text{ is }\\ Q\text{-periodic}}} F(\xi, \varphi), &\text{with }F(\xi, \varphi) &= {\color{black}\frac{1}{2}}\sum_{i \sim j} k_{ij} \bigg(l_{ij} \hat{b}_{ij}^T \xi \hat{b}_{ij} + \big\langle\varphi(x_i) - \varphi(x_j), \hat{b}_{ij} \big \rangle \bigg)^2. 
	\end{align}
	We would like to write it as a quadratic minimization problem using linear algebra. Since we only consider displacements $v(x_j)$ for vertices $x_j$ in the unit cell, the minimization problem in \eqref{effective_energy_appendix} is indeed a finite-dimensional quadratic optimization. Using the notation in \eqref{compatibility_matrix}, we have $\big\langle \varphi(x_i) - \varphi(x_j), \hat{b}_{ij}\big\rangle$ is an entry of $C \bm{\varphi}$, where $\bm{\varphi}$ is the vector form of periodic function $\varphi(x)$. We can also gather $l_{ij} \hat{b}_{ij}^T \xi \hat{b}_{ij}$ as a vector $\bm{b}_\xi$ in the same order of gathering $\big\langle \varphi(x_i) - \varphi(x_j), \hat{b}_{ij}\big\rangle$. Then $F(\xi, \varphi)$ can be expressed as
	\begin{align}\label{linear_algebra_effective_energy}
		F(\xi, \varphi) = {\color{black}\frac{1}{2}} \sum_{i \sim j} k_{ij} \bigg(l_{ij} \hat{b}_{ij}^T \xi \hat{b}_{ij} + \big\langle\varphi(x_i) - \varphi(x_j), \hat{b}_{ij} \big \rangle \bigg)^2 = {\color{black}\frac{1}{2}} \big(\bm{b}_\xi + C \bm{\varphi}\big)^T K \big(\bm{b}_\xi + C \bm{\varphi}\big),
	\end{align}
	where $K$ is a diagonal matrix whose entries are the spring constants $k_{ij}$. For a given strain $\xi$, minimizing $F(\xi, \varphi)$ over all periodic functions $\varphi(x)$ amounts to minimizing a convex, quadratic function of the vector $\bm{\varphi}$ (the Hessian w.r.t $\bm \varphi$ is $C^T K C \succeq 0$). Thus, the optimal $\bm{\varphi}^*_\xi$ must exist and satisfy the equation
	\begin{align}\label{optimal_phi_condition}
		\nabla_{\varphi}  F(\xi, \varphi) = C^T K \bm{b}_\xi + C^T K C \bm\varphi = 0.
	\end{align}
	We know that the compatibility matrix $C$ is not full rank (its null space includes at least two translations), thus the matrix $C^T K C$ is not full rank as well. Solutions to \eqref{optimal_phi_condition} are therefore not unique, but they share the same value of $F(\xi, \varphi)$, since the objective function $F(\xi, \varphi)$ is convex on $\bm{\varphi}$. To avoid future confusion on the non-uniqueness of optimal $\bm{\varphi}^*_\xi$, we stick to the notation $\bm{\varphi}^*_\xi$ for the optimal solution in \eqref{optimal_phi_condition} with the smallest norm, i.e. $\bm{\varphi}^*_\xi = \argmin\limits_{C^T K (C \bm{\varphi} + \bm{b}_\xi)=0}|\bm{\varphi}|_2$. We recall that for the minimum norm solution for an undetermined system $A\bm x = \bm b$ (if $A$ has full rank) is $\bm x^* = A^T (A A^T)^{-1}\bm b $. In our case, the linear constraint $C^T K C$ is a square matrix but not full rank. Therefore, we take a $QR$ decomposition $C^T K C = QR$ to grab the full rank part, and the minimum norm problem for $\bm{\varphi}^*_\xi$ can be written as
	\begin{align}
		\bm{\varphi}^*_\xi = \argmin_{R \bm \varphi = -Q^T C^T K \bm{b}_{\xi}} |\bm \varphi|_2 = -R^T (R R^T)^{-1} Q^T C^T K \bm{b}_{\xi}\label{eqn:unique-varphi}
	\end{align}
	This special $\bm{\varphi}^*_\xi$ is unique and linear in $\xi$, since $\bm{b}_{\xi}$ is linear in $\xi$. The following proposition tells that this optimal $\bm{\varphi}^*_\xi$ yields a self-stress. In fact, any optimal solution in \eqref{optimal_phi_condition} can give a self-stress.
	\begin{proposition}\label{proposition:special_self_stress}
		For a given strain $\xi$, there is a self-stress $\bm{t}^*_\xi$ that depends linearly on $\xi$. Its tension $t^*_{\xi, ij}$ in the spring between $x_i, x_j$ is
		\begin{align}\label{special_tension}
			t^*_{\xi, ij}&= k_{ij} \bigg(l_{ij} \hat{b}_{ij}^T \xi \hat{b}_{ij} + \big\langle \varphi^*_\xi(x_i) - \varphi^*_\xi(x_j), \hat{b}_{ij}\big\rangle\bigg).
		\end{align}
	\end{proposition}
	
	\begin{proof}
		It is easy to check that the vector form of this special tension $t^*_{\xi, ij}$ is
		\begin{align}\label{tension_vector_form}
			\bm{t}^*_\xi = K (\bm{b}_\xi + C \bm{\varphi}^*_\xi).
		\end{align}
		We know that the optimal solution $\bm{\varphi}^*_\xi$ satisfies the optimality condition \eqref{optimal_phi_condition}. Multiplying $C^T$ on the left gives $C^T \bm{t}^*_\xi = C^T K \big(\bm{b}_\xi + C \bm{\varphi}^*_\xi\big) = 0$. Thus, $\bm{t}^*_\xi$ is a self-stress; moreover it is linear in $\xi$ since $\bm{b}_\xi$ and $\bm{\varphi}_\xi^*$ are linear in $\xi$.
	\end{proof}
	
	In the following proposition and lemma, we prove that the effective linear elastic energy $E_\text{eff}(\xi)$ is independent of the size of the unit cell, and quadratic in the symmetric strain $\xi$. We also provide an exact formula for effective linear elastic energy $E_{\text{eff}}(\xi) = \langle A_{\text{eff}} \xi, \xi\rangle$ in Lemma \ref{quadratic_law}.
	\begin{proposition}\label{prop:independence-of-size}
		The effective linear elastic energy $E_{\text{eff}}(\xi)$ does not depend on the size of the unit cell.
	\end{proposition}
	
	\begin{proof}
		Let us denote $E_{\text{eff}}^N(\xi)$ as the effective linear elastic energy for a given strain $\xi$ on a $NQ$-periodic unit cell, i.e. repeating the smallest unit cell $N^2$ times,
		\begin{align}\label{N_periodic_effective}
			E_{\text{eff}}^N(\xi) &= \frac{1}{S_N} \min_{\substack{\varphi_N(x) \text{ is }\\ NQ\text{-periodic}}} F(\xi, \varphi_N),
		\end{align}
		where\footnote{Here $Q$ is the smallest unit cell of our lattice. This has, of course, nothing to do with the orthogonal matrix $Q$ earlier in the $QR$ decomposition of $C^T K C$.} $S_N$ is the area of the $NQ$-periodic unit cell ($S_N = N^2 S_1$.). We prove that $E_{\text{eff}}^N(\xi) = E_{\text{eff}}^1(\xi)$ for any choice of $N$.
		
		First, notice that a $Q$-periodic function $\varphi_1(x)$ is also a $NQ$-periodic function by repeating itself on the $NQ$-periodic unit cell. The optimal $\varphi_{\xi, 1}^*(x)$ satisfies \eqref{optimal_phi_condition} on the unit cell $Q$ also satisfies the optimality condition on the enlarged unit cell $NQ$. In fact, the linear system in \eqref{optimal_phi_condition} on the unit cell $NQ$ becomes $N^2$ copies of the linear system on the unit cell $Q$ if we constrain $\varphi_N(x)$ to be $Q$-periodic instead of $NQ$-periodic. Since the optimization problem in \eqref{N_periodic_effective} is convex, all optimal solutions reach the same optimum,
		\begin{align*}
			E_{\text{eff}}^N(\xi) &= \frac{N^2}{S_N} F(\xi, \varphi_{\xi, 1}^*) = E_{\text{eff}}^1(\xi).
		\end{align*}
		Using the same method, we can prove the effective linear elastic energy is the same on any $MQ \times NQ$ unit cell with $M\neq N$.
	\end{proof}
	
	\begin{remark}
		Since the effective linear elastic energy is independent of the size of the unit cell, we use the notation $E_{\text{eff}}(\xi)$ to denote the effective linear elastic energy on the smallest unit cell, i.e. $E_{\text{eff}}(\xi) = E_{\text{eff}}^1(\xi)$, to avoid confusion.
	\end{remark}
	
	\begin{lemma}\label{quadratic_law}
		The effective linear elastic energy $E_{\text{eff}}(\xi)$ is quadratic in $\xi$, so it has the form
		\begin{align}
			& E_{\text{eff}}(\xi) = {\color{black} \frac{1}{2}} \langle A_\text{eff}\xi, \xi\rangle,
		\end{align}
		where $A_{\text{eff}}$ is a constant symmetric 4-tensor, i.e. $A_\text{eff}\xi$ is symmetric and linear in $\xi$. The exact formula for $A_\text{eff}\xi$ is
		\begin{align}\label{macroscopic_stress}
			A_\text{eff}\xi = \frac{1}{S} \sum_{i \sim j} t^*_{\xi, ij} l_{ij} \hat{b}_{ij} \otimes \hat{b}_{ij},
		\end{align}
		where $t^*_{\xi, ij}$ is the self-stress on spring between $x_i, x_j$ in \eqref{special_tension}.
	\end{lemma}
	
	\begin{proof}
		First, we denote the macroscopic stress as $\bar{\sigma} = \frac{1}{S} \sum_{i \sim j} t^*_{\xi, ij} l_{ij} \hat{b}_{ij} \otimes \hat{b}_{ij}$. We know $\bar{\sigma}$ is symmetric and linear in $\xi$ since the self-stress $t^*_{\xi, ij}$ is linear in $\xi$. We can write $\bar{\sigma} = A_{\text{eff}}\xi$, where $A_{\text{eff}}$ is a constant symmetric 4-tensor. We claim that $E_{\text{eff}}(\xi) = \langle A_\text{eff}\xi, \xi \rangle$.
		
		To prove this, we write both sides in the matrix-vector form
		\begin{align*}
			E_{\text{eff}}(\xi) &= \frac{1}{S} F(\xi, \varphi^*_\xi) = \frac{1}{S} \sum_{i \sim j} {\color{black} \frac{1}{2}} k_{ij} \bigg(l_{ij} \hat{b}_{ij}^T \xi \hat{b}_{ij} + \big\langle\varphi^*_\xi(x_i) - \varphi^*_\xi(x_j), \hat{b}_{ij} \big \rangle \bigg)^2 = {\color{black} \frac{1}{2S}} \big(\bm{b}_\xi + C \bm{\varphi}^*_\xi\big)^T K \big(\bm{b}_\xi + C \bm{\varphi}^*_\xi\big),\\
			\langle A_\text{eff}\xi, \xi \rangle &= {\color{black} \frac{1}{S}} \sum_{i \sim j} \big\langle t^*_{\xi, ij} l_{ij} \hat{b}_{ij} \otimes \hat{b}_{ij}, \xi \big\rangle  = \frac{1}{S} \sum_{i \sim j} t^*_{\xi, ij} \big(l_{ij} \hat{b}_{ij}^T \xi \hat{b}_{ij}\big) = \frac{1}{S} \big(\bm{t}^*_\xi\big)^T \bm{b}_\xi = \frac{1}{S}\big(\bm{b}_\xi + C \bm{\varphi}^*_\xi\big)^T K \bm{b}_{\xi}\\
			&= E_{\text{eff}}(\xi) - \frac{1}{S}\big(\bm{b}_\xi + C \bm{\varphi}^*_\xi\big)^T K C \bm{\varphi}_\xi^*.
		\end{align*}
		The difference $\big(\bm{b}_\xi + C \bm{\varphi}^*_\xi\big)^T K C \bm{\varphi}_\xi^* = 0$ because the optimal $\bm{\varphi}^*_\xi$ satisfies the optimal condition \eqref{optimal_phi_condition}. Therefore, $E_{\text{eff}}(\xi) = {\color{black} \frac{1}{2}} \langle A_\text{eff}\xi, \xi \rangle$.
	\end{proof}
	\setcounter{equation}{0}
	\subsection{The one-periodic mechanism and the corresponding GH mode} \label{appendix-b}
	We present the exact formula for the one-periodic mechanism $u_{\frac{\pi}{3} \shortto \theta}(x)$ discussed in section \ref{one-periodic-mechanism}. For this appendix, we always fix the side length of each equilateral triangle to be 1. We classify vertices in the reference lattice into three types: $A,B$ and $C$ and vertices in the deformed twisted Kagome lattice into three types $\tilde{A}, \tilde{B}$ and $\tilde{C}$.
	\begin{figure}[!htb]
		\centering
		\includegraphics[width=0.8\linewidth]{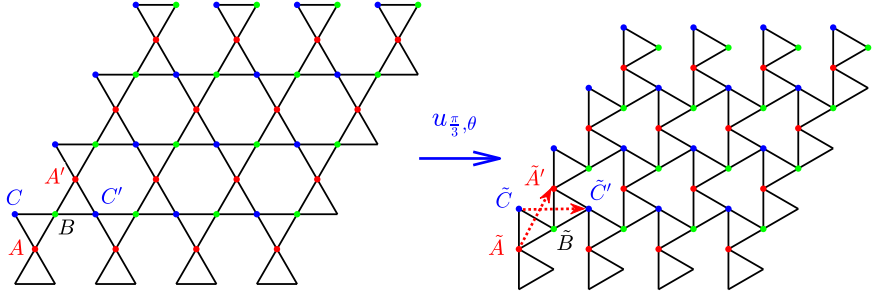}
		\caption{The one-periodic mechanism $u_{\frac{\pi}{3}\shortto \theta}(x)$ from the standard Kagome lattice to the twisted Kagome lattice $L_{\theta}$ ($\theta = \frac{\pi}{6}$). The unit cell of the one-periodic Kagome lattice contains $A,B,C$ three vertices (marked in red, green and blue), and the unit cell of the twisted Kagome lattice contains three vertices $\tilde{A}, \tilde{B}, \tilde{C}$.}
		\label{fig:one-periodic-mechanism-dots}
	\end{figure}
	Once we know the deformation on the five vertices $A,B,C,A',C'$ in Figure \ref{fig:one-periodic-mechanism-dots}, we know everything about the one-periodic mechanism, including the macroscopic deformation gradient. The exact formulas of $u_{\frac{\pi}{3} \shortto \theta}(x)$ on these vertices are
	\begin{align*}
		u_{\frac{\pi}{3} \shortto \theta}(A) &= \tilde{A} = (0,0), & u_{\frac{\pi}{3} \shortto \theta}(B) &= \tilde{B} = (\cos \theta, \sin \theta), & u_{\frac{\pi}{3} \shortto \theta}(C) &=\tilde{C}= \left(\cos(\theta + \frac{\pi}{3}), \sin(\theta + \frac{\pi}{3})\right),\\
		u_{\frac{\pi}{3} \shortto \theta}(A') &= \tilde{A}' = \cos(\theta - \frac{\pi}{3}) (1, \sqrt{3}), & && u_{\frac{\pi}{3} \shortto \theta}(C')&=\tilde{C}'= \left(2\cos(\theta - \frac{\pi}{3}), \sin(\theta + \frac{\pi}{3})\right).
	\end{align*}
	The two primitive vectors $\bm{v}_1^{\text{def}}, \bm{v}_2^{\text{def}}$ for the deformed lattice are the two dotted vector in Figure \ref{fig:one-periodic-mechanism-dots}.  A brief calculation reveals that
	\begin{align*}
		\bm{v}_1^{\text{def}} &= \overrightarrow{\tilde{C}\tilde{C}'} = \cos(\frac{\pi}{3}-\theta) (2,0),\\
		\bm{v}_2^{\text{def}} &= \overrightarrow{\tilde{A}\tilde{A}'} = \cos(\frac{\pi}{3}-\theta) (1,\sqrt{3}).
	\end{align*}
	
	In section \ref{one-periodic-mechanism}, we write the one-periodic mechanism $u_{\frac{\pi}{3} \shortto \theta}(x)$ as $u_{\frac{\pi}{3} \shortto \frac{\pi}{3}+t}(x) = F(t)\cdot x + \varphi(x,t)$ in \eqref{local-one-periodic-mechanism} by changing $\theta = \frac{\pi}{3} + t$. We have seen that $\dot{F}(0) = 0$ and $\dot{\varphi}(x,0)$ is a one-periodic GH mode. We denote this one-periodic GH mode as $\varphi_1^1(x)$; its explicit value at vertex $x$ is
	\begin{align*}
		\dot{\varphi}(x,0) &= \frac{d u_{\frac{\pi}{3} \shortto \frac{\pi}{3}+t} (x)}{d \theta}\bigg|_{t=0} = \frac{d u_{\frac{\pi}{3} \shortto \theta} (x)}{d \theta}\bigg|_{\theta = \frac{\pi}{3}}.
	\end{align*}
	By plugging the exact formulas for $u_{\frac{\pi}{3} \shortto \theta}(x)$, we get 
	\begin{align*}
		\varphi_1^1(A) &= \frac{d (0,0)}{d \theta}\bigg|_{\theta = \frac{\pi}{3}} = (0,0),\\
		\varphi_1^1(B) &= \frac{d (\cos \theta, \sin \theta)}{d \theta}\bigg|_{\theta = \frac{\pi}{3}} = \left(-\frac{\sqrt{3}}{2}, \frac{1}{2}\right)\\
		\varphi_1^1(C)&= \frac{d \left(\cos(\theta + \frac{\pi}{3}), \sin(\theta + \frac{\pi}{3})\right)}{d \theta}\bigg|_{\theta = \frac{\pi}{3}} = \left(-\frac{\sqrt{3}}{2}, -\frac{1}{2}\right).
	\end{align*}
	\setcounter{equation}{0}
	\subsection{The two-periodic mechanism $u_{\theta_1, \theta_2, \theta_3}(x)$ and its corresponding GH modes} \label{appendix-c}
	We provide the explicit formula for the two-periodic mechanism $u_{\theta_1, \theta_2, \theta_3}(x)$ in section \ref{sec:two-periodic-mechanism}. The parameters $\theta_1, \theta_2, \theta_3$ are shown in Figure \ref{fig:two-periodic-unit-cell}. We use the same notation as in section \ref{sec:two-periodic-mechanism} and classify the vertices in the reference lattice into three types: $A$, $B$ and $C$. The unit cell for the two-periodic standard Kagome lattice has four vertices of each type $A,B,C$. Therefore, we denote them as $A_{i,j}, B_{i,j}, C_{i,j}$ with $i,j \in \{0,1\}$; we denote vertices in the deformed lattice as $\tilde{A}_{i,j}, \tilde{B}_{i,j}, \tilde{C}_{i,j}$ with $i,j \in \{0,1\}$ (shown in Figure \ref{fig:two-periodic-mechanism-dots}).
	\begin{figure}[!htb]
		\centering
		\includegraphics[width=0.9\linewidth]{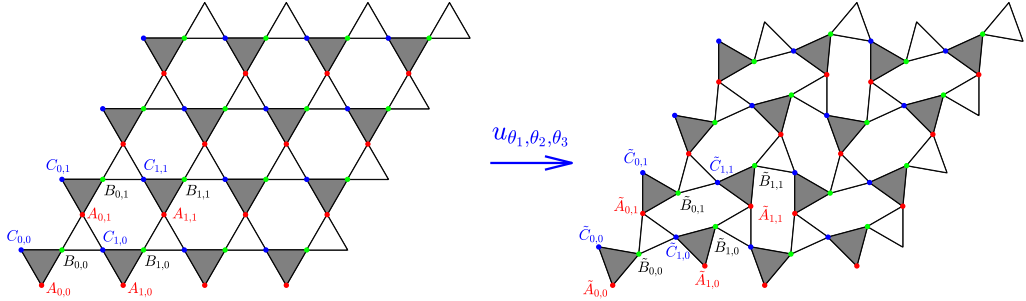}
		\caption{The three-parameter two-periodic mechanism $u_{\theta_1, \theta_2, \theta_3}(x)$.}
		\label{fig:two-periodic-mechanism-dots}
	\end{figure}
	For simplicity, we denote $u_{\theta_1, \theta_2, \theta_3}(x) = u(x)$. Then the explicit formula for each vertex under the two-periodic mechanism $u(x)$ is the following:
	\begin{align*}
		u(A_{0,0})&=\tilde{A}_{0,0} = (0,0), \qquad u(B_{0,0})=\tilde{B}_{0,0} \left(\cos\theta_1, \sin \theta_1\right), \qquad u(C_{0,0})=\tilde{C}_{0,0} \left(\cos(\theta_1+\frac{\pi}{3}), \sin(\theta_1+\frac{\pi}{3})\right),\\
		u(A_{1,0})&=\tilde{A}_{1,0}= \left(\cos\theta_1 - \cos(\theta_2 + \frac{\pi}{3}) + \cos(\theta_4 - \frac{\pi}{3}), \sin\theta_1 - \sin(\theta_2 + \frac{\pi}{3}) + \sin(\theta_4 - \frac{\pi}{3})\right),\\ 
		u(B_{1,0})&= \tilde{B}_{1,0}=\left(\cos\theta_1 + \cos(\theta_2 - \frac{\pi}{3}) + \cos(\theta_4 - \frac{\pi}{3}), \sin\theta_1 + \sin(\theta_2 - \frac{\pi}{3}) + \sin(\theta_4 - \frac{\pi}{3})\right),\\
		u(C_{1,0})&= \tilde{C}_{1,0}=\left(\cos\theta_1 + \cos(\theta_4 - \frac{\pi}{3}), \sin\theta_1 + \sin(\theta_4 - \frac{\pi}{3}) \right),\\
		u(A_{0,1})&= \tilde{A}_{0,1} =\left(\cos \theta_1 + \cos \theta_4, \sin \theta_1 + \sin \theta_4\right),\\ 
		u(B_{0,1})&= \tilde{B}_{0,1} =\left(\cos \theta_1 + \cos \theta_3 + \cos \theta_4, \sin \theta_1 + \sin \theta_3 + \sin \theta_4 \right),\\
		u(C_{0,1})&= \tilde{C}_{0,1} =\left(\cos \theta_1 +  \cos(\theta_3 + \frac{\pi}{3}) +\cos \theta_4, \sin \theta_1 + \sin(\theta_3 + \frac{\pi}{3}) + \sin \theta_4\right),\\
		u(A_{1,1}) &= \tilde{A}_{1,1}=\left(\cos \theta_1 + \cos(\theta_2 - \frac{\pi}{3}) +  \cos \theta_3 + \cos(\theta_4 - \frac{\pi}{3}), \sin \theta_1 + \sin(\theta_2 - \frac{\pi}{3}) +  \sin \theta_3 + \sin(\theta_4 - \frac{\pi}{3}) \right),\\
		u(B_{1,1}) &= \tilde{B}_{1,1}=\left(\cos \theta_1 + \cos(\theta_2 - \frac{\pi}{3}) + \cos \theta_3 + \sqrt{3}\cos(\theta_4 - \frac{\pi}{6}), \sin \theta_1 + \sin(\theta_2 - \frac{\pi}{3}) + \sin \theta_3 + \sqrt{3} \sin(\theta_4 - \frac{\pi}{6})\right),\\
		u(C_{1,1}) &= \tilde{C}_{1,1}=\left(\cos \theta_1 + \cos(\theta_2 - \frac{\pi}{3}) + \cos \theta_3 + \cos \theta_4, \sin \theta_1 + \sin(\theta_2 - \frac{\pi}{3}) + \sin \theta_3 + \sin \theta_4\right).
	\end{align*}
	
	Using the same method as in Appendix \ref{appendix-b}, we can compute the three GH modes $\varphi_1^2(x), \varphi_2^2(x), \varphi_3^2(x)$ in section \ref{subsection:two-periodic-GH} in \eqref{eqn:two-periodic-GH-def}. The explicit formula for $\varphi_1^2(x)$ is
	\begin{align*}
		\varphi_1^2(A_{0,0})&= (0,0), & \varphi_1^2(B_{0,0})&= \left(-\frac{\sqrt{3}}{2}, \frac{1}{2}\right),& \varphi_1^2(C_{0,0})&= \left(-\frac{\sqrt{3}}{2}, -\frac{1}{2}\right),\\
		\varphi_1^2(A_{1,0})&= \left(-\frac{\sqrt{3}}{2}, -\frac{1}{2}\right),& \varphi_1^2(B_{1,0})&= \left(-\frac{\sqrt{3}}{2}, -\frac{1}{2}\right),& \varphi_1^2(C_{1,0})&= \left(-\frac{\sqrt{3}}{2}, -\frac{1}{2}\right),\\
		\varphi_1^2(A_{0,1})&= (0,0), & \varphi_1^2(B_{0,1})&= (0,0), & \varphi_1^2(C_{0,1})&= (0,0),\\
		\varphi_1^2(A_{1,1})&= (0,0), & \varphi_1^2(B_{1,1})&= (0,-1),& \varphi_1^2(C_{1,1})&= (0,0).
	\end{align*}
	The explicit formula for $\varphi_2^2(x)$ is
	\begin{align*}
		\varphi_2^2(A_{0,0})&= (0,0), & \varphi_2^2(B_{0,0})&=(0,0),& \varphi_2^2(C_{0,0})&= (0,0), \\
		\varphi_2^2(A_{1,0})&= \left(-\frac{\sqrt{3}}{2}, -\frac{1}{2}\right),& \varphi_2^2(B_{1,0})&= (0,0), & \varphi_2^2(C_{1,0})&= (0,-1), \\
		\varphi_2^2(A_{0,1})&= \left(\frac{\sqrt{3}}{2}, -\frac{1}{2}\right), & \varphi_2^2(B_{0,1})&= \left(\frac{\sqrt{3}}{2}, -\frac{1}{2}\right), & \varphi_2^2(C_{0,1})&= \left(\frac{\sqrt{3}}{2}, -\frac{1}{2}\right),\\
		\varphi_2^2(A_{1,1})&= (0,0), & \varphi_2^2(B_{1,1})&= \left(\frac{\sqrt{3}}{2}, -\frac{1}{2}\right), & \varphi_2^2(C_{1,1})&= \left(\frac{\sqrt{3}}{2}, \frac{1}{2}\right).
	\end{align*}
	The explicit formula for $\varphi_3^2(x)$ is
	\begin{align*}
		\varphi_3^2(A_{0,0})&= (0,0), & \varphi_3^2(B_{0,0})&=(0,0),& \varphi_3^2(C_{0,0})&= (0,0), \\
		\varphi_3^2(A_{1,0})&= (0,-1),& \varphi_3^2(B_{1,0})&= (0,-1), & \varphi_3^2(C_{1,0})&= (0,-1), \\
		\varphi_3^2(A_{0,1})&= \left(\frac{\sqrt{3}}{2}, -\frac{1}{2}\right), & \varphi_3^2(B_{0,1})&= (0,0), & \varphi_3^2(C_{0,1})&= (0,-1),\\
		\varphi_3^2(A_{1,1})&= \left(-\frac{\sqrt{3}}{2}, -\frac{1}{2}\right), & \varphi_3^2(B_{1,1})&= (0,-1), & \varphi_3^2(C_{1,1})&= (0,0).
	\end{align*}
	
	In section \ref{subsec:special-two-periodic}, we noted two special cases of $u_{\theta_1, \theta_2, \theta_3}(x)$ that have additional symmetries. One, shown in Figure \ref{fig:two-periodic-special}(a), is actually a two-by-one periodic mechanism. We now show that the associated GH mode is the Fleck-Hutchinson mode $u_1(x)$ obtained by taking $N=2$ in \eqref{eqn:u1}. We recall that the two-by-one periodic mechanism is a special case of $u_{\theta_1, \theta_2, \theta_3}(x)$ by choosing $\theta_1 = \gamma, \theta_2 = \gamma, \theta_3 = \frac{2\pi}{3} - \gamma$. The associated two-by-one periodic GH mode, which we shall refer to as $\varphi^{2,1}(x)$, is the infinitesimal version of this mechanism. A brief calculation gives $\varphi^{2,1}(x) = -\varphi_1^2(x) - \varphi_2^2(x) + \varphi_3^2(x)$ (shown in Figure \ref{fig:two-by-one-fleck}). The values of $\varphi^{2,1}(x)$ on the vertices of the two-by-one periodic unit cell are
	\begin{align*}
		\varphi^{2,1}(A_{0,0})&= (0,0), & \varphi^{2,1}(B_{0,0})&= \left(\frac{\sqrt{3}}{2}, -\frac{1}{2}\right),& \varphi^{2,1}(C_{0,0})&= \left(\frac{\sqrt{3}}{2}, \frac{1}{2}\right), \\
		\varphi^{2,1}(A_{0,1})&= (0,0), & \varphi^{2,1}(B_{0,1})&= \left(-\frac{\sqrt{3}}{2}, \frac{1}{2}\right), & \varphi^{2,1}(C_{0,1})&= \left(-\frac{\sqrt{3}}{2}, -\frac{1}{2}\right).
	\end{align*}
	It is easy to check that $\varphi^{2,1}(x)$ is exactly $u_1^{\frac{ N}{2}, N}(x)$ when $N$ is even in \eqref{eqn:u_1} since
	\begin{align*}
		u_1^{\frac{ N}{2}, N}(A_{0,k}) &= (0,0) & u_1^{\frac{ N}{2}, N}(B_{0,k}) &= \cos(k\pi) \left(\frac{\sqrt{3}}{2}, -\frac{1}{2}\right) &u_1^{\frac{ N}{2}, N}(C_{0,k}) &= \cos(k\pi) \left(\frac{\sqrt{3}}{2}, \frac{1}{2}\right),
	\end{align*}
	with $k=0,1,\dots, N-1$.
	\begin{figure}[!htb]
		\centering
		\includegraphics[width=0.4\linewidth]{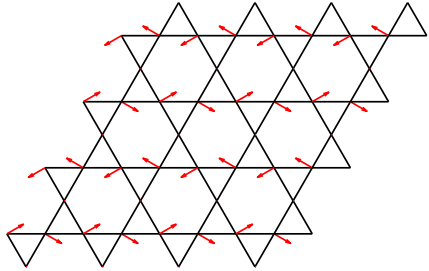}
		\caption{The two-by-one periodic GH mode $\varphi^{2,1}(x)$.}
		\label{fig:two-by-one-fleck}
	\end{figure}
	
	\setcounter{equation}{0}
	\subsection{More details on the Fleck-Hutchinson modes} \label{appendix-e}
	\subsubsection{Review of Fleck-Hutchinson modes} \label{appendix-e-1}
	The Fleck-Hutchinson modes are obtained by considering complex-valued displacements with vanishing linear elastic strains and Bloch-type boundary conditions. In section \ref{sec:preliminary-GH}, we have seen that GH modes are periodic displacements whose linear elastic strains vanish. Fleck-Hutchinson modes are similar, but with two different assumptions: (1) the displacements $d(x)$ are now complex, i.e. $d(x) \in \mathbb{C}^2$; and (2) the Bloch-type boundary condition requires the displacement $d(x)$ to satisfy
	\begin{align}
		d(\bm{j} + \bm{x}) &= d(\bm{j}) \exp(2\pi i \bm{x}\cdot \bm{w}), \label{eqn:bloch-appendix}
	\end{align}
	where $\bm{j}$ are vertices in the unit cell and $\bm{x} = n_1\bm{v}_1 + n_2 \bm{v}_2$ is a translation vector with integer-valued $n_1, n_2$ and primitive vectors $\bm{v}_1, \bm{v}_2$ of the reference lattice. The vector $\bm{w}$ is the so-called Bloch wave number; it is chosen as $\bm{w} = w_1 \bm{a}_1 + w_2 \bm{a_2}$, where $\bm{a}_1, \bm{a}_2$ are primitive vectors in the Brillouin zone and $w_1, w_2 \in (0,1]$. 
	
	From now on, we shall focus on the one-periodic standard Kagome lattice and find the corresponding Fleck-Hutchinson modes. There are three vertices $A,B,C$ in the unit cell of the one-periodic standard Kagome lattice; vertex $\bm{j}$ are chosen from the three vertices. We choose a different pair of primitive vectors with $\bm{v}_1 = (1,\sqrt{3})$ in the 60 degree direction and $\bm{v}_2 = (-1,\sqrt{3})$ in the 120 degree direction. The  corresponding primitive vectors $\bm{a}_1, \bm{a}_2$ in the Brillouin zone are $\bm{a}_1 = \frac{1}{2\sqrt{3}}\left(\sqrt{3},1\right), \bm{a}_2 = \frac{1}{2\sqrt{3}}\left(-\sqrt{3},1\right)$ satisfying
	\begin{align}
		\bm{a}_i \cdot \bm{v}_j = \delta_{ij}, \qquad i,j \in \{1,2\}, \label{eqn:brillouin}
	\end{align}
	where $\delta_{ij}$ is the Kronecker delta. 
	
	There are three special choices of $w_1, w_2$ such that the displacement $d(x)$ becomes $N$-periodic: (1) $w_1 = w_2 = \frac{s}{N}$; (2) $w_1 = 0, w_2 = \frac{s}{N}$; and (3) $w_2 = 0, w_1 = \frac{s}{N}$. In all three cases, $s$ is an integer in the range $0 \leq s \leq N-1$. Let us focus on the first case $w_1 = w_2 = \frac{s}{N}$, and the other two cases are similar. We shall show that the displacement $d(x)$ is $N$-periodic when $w_1 = w_2 = \frac{s}{N}$, i.e. $d(\bm{j} + \bm{x}) = d(\bm{j})$ for any translation vector $\bm{x} = n_1 \bm{v}_1 + n_2 \bm{v}_2$ with $n_1, n_2$ as multiples of $N$ ($n_1 = m_1N, n_2 = m_2 N$ and $m_1, m_2 \in \mathbb{Z}$). The $N$-periodicity comes from a simple calculation: the factor $\exp(2\pi i \bm{x} \cdot \bm{w})$ in \eqref{eqn:bloch-appendix} using \eqref{eqn:brillouin} becomes
	\begin{align*}
		\exp(2\pi i \bm{x} \cdot \bm{w}) &= \exp(2\pi i (n_1 w_1 + n_2 w_2)).
	\end{align*}
	When $w_1 = w_2 = \frac{s}{N}$ and $n_1 = m_1N, n_2 = m_2 N$, this factor becomes 1. Thus, we obtain the displacement $d(x)$ is $N$-periodic.
	
	Actually, the displacement $d(x)$ associated to $w_1 = w_2 = \frac{s}{N}$ is not just $N$-periodic, it is indeed $N$-by-one periodic. By this, we mean that $d(x)$ is indeed one-periodic in the horizontal direction, i.e. $d(\bm{j} + \bm{x}) = d(\bm{j})$ and $\bm{x} = \bm{v}_1 - \bm{v}_2 = (2,0)$ is the smallest translation vector in the horizontal direction. This is true because the factor $\exp(2\pi i \bm{x} \cdot \bm{w})$ becomes $\exp(2\pi i \bm{x} \cdot \bm{w}) = \exp(2\pi i (w_1 - w_2)) = 1$. 
	
	\begin{remark}\label{rmk:remark-symmetry}
		For the other two cases, for example $w_1 = 0, w_2 = \frac{s}{N}$, we get $N$-by-one displacements $d(x)$ and their shorter period occurs in the 60 degree direction. For the $w_2 = 0, w_1 = \frac{s}{N}$ case, we get $N$-by-one periodic displacements $d(x)$ with shorter period in the 120 degree direction. 
	\end{remark}
	
	So far, we have seen that the Bloch-type boundary condition becomes periodic boundary conditions with special choices of bloch wave number $\bm{w}$. Now we find the displacement $d(x)$ with vanishing linear elastic strain for a general bloch wave number $\bm{w}$. For a chosen $\bm{w}$, the first-order spring extension $e_{ij}(\bm{w})$ on the spring between $x_i, x_j$ is
	\begin{align}
		e_{ij}(\bm{w}) &= \langle d(x_i) - d(x_j), \hat{b}_{ij}\rangle, \label{eqn:complex_first_order_bond}
	\end{align}
	where $\hat{b}_{ij} = \frac{x_i - x_j}{|x_i - x_j|}$ indicates the spring direction. Notice that $e_{ij}(\bm{w})$ is complex-valued. Similarly to the real compatibility matrix in \eqref{compatibility_matrix}, we can write the linear relationship between the displacement value $d(x)$ on vertices in the unit cell and the first-order spring extension $e_{ij}(\bm{w})$ in terms of the complex version of the compatibility matrix $C(\bm{w})$ w.r.t. the Bloch wave number $\bm{w}$. The complex compatibility matrix $C(\bm{w})$ for the standard Kagome lattice with the smallest unit cell was found in \cite{hutchinson2006structural}. Our $C(\bm{w})$ looks a little different from the one in \cite{hutchinson2006structural} because we choose a different set of springs and vertices in the unit cell as shown in Figure \ref{fig:standardkagomespring}(b); our $C(\bm{w})$ transforms the vector form of displacement $d(x)$, i.e. $\begin{pmatrix}
		d(A) & d(B) & d(C)
	\end{pmatrix}^T \in \mathbb{C}^6$, to the vector form of $e_{ij}(\bm{w})$, i.e. $\begin{pmatrix}
		e_1 & e_2 & e_3 & e_4 & e_5 & e_6
	\end{pmatrix}^T \in \mathbb{C}^6$, and its explicit form is
	\begin{align*}
		C(\bm{w}) &= \begin{pmatrix}
			-\frac{1}{2}& -\frac{\sqrt{3}}{2} & \frac{1}{2} & \frac{\sqrt{3}}{2} & 0 & 0\\
			0 &  0 & 1 & 0 & -1 & 0\\
			\frac{1}{2}& -\frac{\sqrt{3}}{2} & 0 & 0 & -\frac{1}{2} & \frac{\sqrt{3}}{2}\\
			0 &  0 & -1 & 0 & z_1 \bar{z}_2 & 0\\
			-\frac{1}{2}z_2& \frac{\sqrt{3}}{2}z_2 & 0 & 0 & \frac{1}{2} & -\frac{\sqrt{3}}{2}\\
			\frac{1}{2}& \frac{\sqrt{3}}{2} & -\frac{1}{2}\bar{z}_1 & -\frac{\sqrt{3}}{2}\bar{z}_1 & 0 & 0
		\end{pmatrix},
	\end{align*}
	where $\bm{z}_j = \exp(2\pi i w_j)$ and $j=1,2$. This complex compatibility matrix has null vectors in three cases: (1) $w_1 = w_2$; (2) $w_1 = 0$; and (3) $w_2 = 0$. In each case, the null space is one-dimensional. For example, when $w_1 = w_2$, the null vector $d_{\bm{w}}(x)$ has values on the three vertices in the smallest unit cell
	\begin{align}
		d_{\bm{w}}(A) &= \begin{pmatrix}
			0, 0
		\end{pmatrix}^T & d_{\bm{w}}(B) &= \begin{pmatrix}
			\frac{\sqrt{3}}{2}, -\frac{1}{2}
		\end{pmatrix}^T & d_{\bm{w}}(C) &= \begin{pmatrix}
			\frac{\sqrt{3}}{2}, \frac{1}{2}
		\end{pmatrix}^T. \label{eqn:complex-displacement}
	\end{align}
	The values of $d_{\bm{w}}(x)$ on the remaining vertices are determined by \eqref{eqn:bloch-appendix}. Notice that the three cases here contain the three cases where $d(x)$ is $N$-periodic.
	
	For the standard Kagome lattice, when the compatibility matrix $C(\bm{w})$ has a complex null vector $d_{\bm{w}}(x)$, its real and imaginary parts will be GH modes, provided $d_{\bm{w}}(x)$ is periodic. Therefore, the real and complex parts of the displacement $d_{\bm{w}}(x)$ are two $N$-by-one periodic GH modes when $w_1 = w_2 = \frac{s}{N}$ (fixing $s$), and similar results hold for the other two cases. Let us focus on case $w_1 = w_2 = \frac{s}{N}$ and compute the exact values of $d_{\bm{w}}(x)$ on vertices in the $N$-by-one periodic unit cell. Using \eqref{eqn:bloch-appendix} and \eqref{eqn:complex-displacement}, we have 
	\begin{align*}
		d_{\bm{w}}(A_{0,1}) &= \begin{pmatrix}
			0, 0
		\end{pmatrix}^T & d_{\bm{w}}(B_{0,1}) &= \exp(\frac{2s\pi i}{N}) \begin{pmatrix}
			\frac{\sqrt{3}}{2}, -\frac{1}{2}
		\end{pmatrix}^T & d_{\bm{w}}(C_{0,1}) &= \exp(\frac{2s\pi i}{N}) \begin{pmatrix}
			\frac{\sqrt{3}}{2}, \frac{1}{2}
		\end{pmatrix}^T,\\
		d_{\bm{w}}(A_{0,2}) &= \begin{pmatrix}
			0, 0
		\end{pmatrix}^T & d_{\bm{w}}(B_{0,2}) &= \exp(\frac{4s\pi i}{N})\begin{pmatrix}
			\frac{\sqrt{3}}{2}, -\frac{1}{2}
		\end{pmatrix}^T & d_{\bm{w}}(C_{0,2}) &= \exp(\frac{4s\pi i}{N}) \begin{pmatrix}
			\frac{\sqrt{3}}{2}, \frac{1}{2}
		\end{pmatrix}^T,\\
		& \dots & &\dots & & \dots\\
		d_{\bm{w}}(A_{0,N-1}) &= \begin{pmatrix}
			0, 0
		\end{pmatrix}^T & d_{\bm{w}}(B_{0,N-1}) &= \exp(\frac{2s(N-1)\pi i}{N})\begin{pmatrix}
			\frac{\sqrt{3}}{2}, -\frac{1}{2}
		\end{pmatrix}^T & d_{\bm{w}}(C_{0,N-1}) &= \exp(\frac{2s(N-1)\pi i}{N})\begin{pmatrix}
			\frac{\sqrt{3}}{2}, \frac{1}{2}
		\end{pmatrix}^T,
	\end{align*}
	where $A_{0,k}, B_{0,k}, C_{0,k}, k = 0,1,\dots, N-1$ are the vertices in the $N$-by-one periodic unit cell (defined in section \ref{subsection:fleck-hutchinson}). We refer to the real part as $u_1^{s,N}(x)$ and the complex part as $u_2^{s,N}(x)$. Their values on the vertices are 
	\begin{align}
		u_1^{s,N}(A_{0,k}) &= (0,0)^T & u_1^{s,N}(B_{0,k}) &= \cos(\frac{2ks\pi}{N})\begin{pmatrix}
			\frac{\sqrt{3}}{2}, -\frac{1}{2}
		\end{pmatrix}^T & u_1^{s,N}(C_{0,k}) &= \cos(\frac{2ks\pi}{N})\begin{pmatrix}
			\frac{\sqrt{3}}{2}, \frac{1}{2}
		\end{pmatrix}^T \label{eqn:u1}\\
		u_2^{s,N}(A_{0,k}) &= (0,0)^T & u_2^{s,N}(B_{0,k}) &= \sin(\frac{2ks\pi}{N})\begin{pmatrix}
			\frac{\sqrt{3}}{2}, -\frac{1}{2}
		\end{pmatrix}^T & u_2^{s,N}(B_{0,k}) &= \sin(\frac{2ks\pi}{N})\begin{pmatrix}
			\frac{\sqrt{3}}{2},\frac{1}{2} 
		\end{pmatrix}^T. \label{eqn:u2}
	\end{align}
	\begin{proposition}\label{app-d:prop-basis}
		For a fixed $N$, there are in total $N$ linearly independent $u_1^{s,N}(x)$ and $u_2^{s,N}(x)$ by varying s in the range $0 \leq s \leq N-1$. Moreover, they form a basis for the space of $N$-by-one periodic GH modes.
	\end{proposition}
	\begin{proof}
		As mentioned in \cite{hutchinson2006structural}, the range of $s$ can be reduced to $0 \leq s \leq \frac{N}{2}$, since $u_1^{s,N}(x) = u_1^{N-s,N}(x)$ and $u_2^{s,N}(x) = -u_2^{N-s,N}(x)$. Therefore, we have $u_1^{0,N}(x), u_1^{1,N}(x), \dots, u_1^{\lfloor \frac{N}{2} \rfloor,N}(x)$ and $u_2^{1,N}(x), \dots, u_2^{\lfloor \frac{N}{2} \rfloor,N}(x)$ ($u_2^{0,N}(x)=0$ is neglected). When $N$ is even, it is easy to check that $u_2^{\frac{N}{2},N}(x)$ vanishes. Thus, for any $N$, we have $N$ distinct $N$-by-one periodic Fleck-Hutchinson modes: $u_1^{0,N}(x), u_1^{1,N}(x), \dots, u_1^{\lfloor \frac{N}{2} \rfloor,N}(x)$ and$u_2^{1,N}(x), \dots, u_2^{\lfloor \frac{N-1}{2} \rfloor,N}(x)$.
		
		Now we show that these $N$ Fleck-Hutchinson modes are linearly independent and non-trivial translations are not linear combinations of these modes. We first prove the linear independence: for a linear combination 
		\begin{align}
			a_0 u_1^{0,N}(x) + a_1 u_1^{1,N}(x) + \dots a_{\lfloor \frac{N}{2} \rfloor} u_1^{\lfloor \frac{N}{2} \rfloor,N}(x) + b_1 u_2^{1,N}(x) + \dots b_{\lfloor \frac{N-1}{2} \rfloor} u_2^{\lfloor \frac{N-1}{2} \rfloor,N}(x) = 0, \label{eqn:linear-combination}
		\end{align} 
		we plug in $x = B_{0,k}$. Using \eqref{eqn:u1}-\eqref{eqn:u2}, we get that for all $k = 0,1,\dots,N-1$,
		\begin{align}
			a_0 + a_1\cos(\frac{2k\pi}{N}) + \dots a_{\lfloor \frac{N}{2} \rfloor} \cos(\frac{2k\lfloor \frac{N}{2} \rfloor\pi}{N}) + b_1 \sin(\frac{2k\pi}{N}) + \dots b_{\lfloor \frac{N-1}{2} \rfloor} \sin(\frac{2k\lfloor \frac{N-1}{2}\rfloor\pi}{N})= 0. \label{eqn:linear-one}
		\end{align}
		We get the same equality if we plug in $C_{0,k}$ and \eqref{eqn:linear-combination} holds automatically for all $A_{0,k}$. Therefore, the two equalities \eqref{eqn:linear-combination} and \eqref{eqn:linear-one} are equivalent.  By writing \eqref{eqn:linear-one} in a matrix-vector form, we get
		\begin{align}
			\begin{pmatrix}
				1 & 1 & \dots & 1  & 0 & \dots & 0\\
				1 & \cos(\frac{2\pi}{N}) & \dots & \cos(\frac{2\lfloor \frac{N}{2} \rfloor\pi}{N}) & \sin(\frac{2\pi}{N}) & \dots & \sin(\frac{2\lfloor \frac{N-1}{2} \rfloor\pi}{N}) \\
				\vdots & \vdots  & \ddots & \vdots  & \vdots  & \ddots & \vdots \\
				1 & \cos(\frac{2(N-1)\pi}{N}) & \dots & \cos(\frac{2(N-1)\lfloor \frac{N}{2} \rfloor\pi}{N}) & \sin(\frac{2(N-1)\pi}{N}) & \dots & \sin(\frac{2(N-1)\lfloor \frac{N-1}{2} \rfloor\pi}{N}) 
			\end{pmatrix} \begin{pmatrix}
				a_0\\
				a_1\\
				\vdots\\
				a_{\lfloor \frac{N}{2} \rfloor}\\
				b_1\\
				\vdots\\
				b_{\lfloor \frac{N-1}{2} \rfloor}
			\end{pmatrix} = 0. \label{eqn:linear-expansion}
		\end{align} 
		It can be checked that the matrix in \eqref{eqn:linear-expansion} is invertible (briefly, the column space of this matrix is the same as the column space of the discrete Fourier transform matrix). Given the invertibility, \eqref{eqn:linear-expansion} only holds when $a_0 = a_1 = \dots = a_{\lfloor \frac{N}{2} \rfloor} = b_1 = \dots = b_{\lfloor \frac{N-1}{2} \rfloor} = 0$. Thus, the $N$ Fleck-Hutchinson modes are linearly independent. To show that a non-trivial translation is not a linear combination of these Fleck-Hutchinson modes, we observe that $u_1^{s,N}(A_{0,k}) = u_2^{s,N}(A_{0,k}) = 0$ for all $s$, $k$ and $N$. If a linear combination \eqref{eqn:linear-combination} gives a translation, then it mush actually vanish. Therefore, the space spanned by these $N$ Fleck-Hutchinson modes does not include translations.
		
		Lastly, we show that the $N$ linearly independent Fleck-Hutchinson modes form a basis for the space of $N$-by-one periodic GH modes. It is equivalent to show that the space of $N$-by-one periodic GH modes is $N$-dimensional. The argument is parallel to the one used in section \ref{sec:preliminary-GH}. First, we observe that if $C$ is the compatibility matrix introduced in section \ref{sec:preliminary-GH} (for the $N$-by-one periodic case), then $\ker(C) = \ker(C^T)$ is at least $(N+2)$-dimensional, since the self-stresses that are constant on a single line (and its periodic images) span an $(N+2)$-dimensional space (there are $N$ of them associated with horizontal lines, and two associated with lines in the 60 degree or 120 degree directions). Next, we observe that it is at most $(N+2)$-dimensional, by considering the linear equations $C d = 0$ and finding reductions similar to those used in section \ref{sec:preliminary-GH} (the details are left to the reader). So $\ker(C)$ has dimension exactly $N+2$. Eliminating the two translations, we conclude that the space of GH modes has dimension $N$.
	\end{proof}
	
	\begin{remark}\label{app-d-remark}
		By symmetry, when $w_1=0, w_2=\frac{s}{N}$, we get $N$ linearly independent Fleck-Hutchinson modes with period 1 in the 60 degree direction; and when $w_2=0, w_1 = \frac{s}{N}$, we get another $N$ linearly independent Fleck-Hutchinson modes with period 1 in the 120 degree direction. Each of these families includes the $N$-by-one periodic extension of the one-periodic GH mode. Aside from this, the three families can be shown to be linearly independent (by a calculation similar to the one done above). Therefore, taken together the three families of Fleck-Hutchinson modes span a $(3N-2)$-dimensional subspace of the $N$-periodic GH modes. We showed in section \ref{sec:preliminary-GH} that this space has dimension $3N-2$, so we have obtained a basis for the entire space of $N$-period GH modes. 
	\end{remark}
	
	\setcounter{equation}{0}
	\subsection{The four-by-one periodic mechanism}\label{appendix-f-2}
	In section \ref{subsec:non-periodic}, we showed that there is in fact a way to layer the one-periodic and two-by-one periodic mechanism as shown in Figure \ref{fig:non-periodic-four}. Here we present the details of the four-by-one periodic mechanism $u(x)$ in Figure \ref{fig:four-by-one-mechanism}(a) that achieves a given compression ratio $c = \cos(\frac{\pi}{3} - \theta)$. We use the same notation as in section \ref{subsection:consistency-condition} to denote the 12 vertices in the four-by-one periodic unit cell; they are $A_{0,k}, B_{0.k}, C_{0,k}$ with $k = 0,1,2,3$ as shown in Figure \ref{fig:three-by-one-consistency}(c). The values of $u(x)$ on these vertices are
	\begin{align*}
		u(A_{0,0}) &= (0,0), \qquad u(B_{0,0}) = (\cos \theta, \sin \theta), \qquad u(C_{0,0}) = (\cos(\theta + \frac{\pi}{3}), \sin(\theta + \frac{\pi}{3})),\\
		u(A_{0,1}) &= u(B_{0,0}) + \left(\cos(\frac{2\pi}{3}-\theta),\sin(\frac{2\pi}{3}-\theta)\right),\\
		u(B_{0,1}) &= u(A_{0,1}) + \left(\cos(\frac{2\pi}{3}-\theta),\sin(\frac{2\pi}{3}-\theta)\right),\\
		u(C_{0,1}) &= u(A_{0,1}) + \left(-\cos\theta, \sin \theta\right),\\
		u(A_{0,2}) &= u(B_{0,1}) + \left(\cos\theta, \sin \theta\right),\\
		u(B_{0,2}) &= u(A_{0,2}) + \left(\cos(\frac{2\pi}{3}-\theta),\sin(\frac{2\pi}{3}-\theta)\right),\\
		u(C_{0,2}) &= u(A_{0,2}) + \left(-\cos\theta, \sin \theta\right),\\
		u(A_{0,3}) &= u(B_{0,2}) + \left(\cos\theta, \sin \theta\right),\\
		u(B_{0,3}) &= u(A_{0,3}) + \left(\cos\theta,\sin\theta\right),\\
		u(C_{0,3}) &= u(A_{0,3}) + \left(\cos(\theta+\frac{\pi}{3}),\sin(\theta+\frac{\pi}{3})\right).
	\end{align*}
	The four-by-one periodic GH mode $\varphi_1^4(x)$ corresponding to the four-by-one periodic mechanism that rotates the two shaded triangles in the bottom layer towards each other in Figure \ref{fig:four-by-one-mechanism}(c) (by taking $\theta = \frac{\pi}{3} - t$) has values
	\begin{align*}
		\varphi_1^4(A_{0,0}) &= \left(0,0\right), & \varphi_1^4(B_{0,0}) &= \left(\frac{\sqrt{3}}{2}, -\frac{1}{2}\right), & \varphi_1^4(C_{0,0}) &= \left(\frac{\sqrt{3}}{2}, \frac{1}{2}\right),\\
		\varphi_1^4(A_{0,1}) &= \left(0,0\right), & \varphi_1^4(B_{0,1}) &= -\left(\frac{\sqrt{3}}{2}, -\frac{1}{2}\right), & \varphi_1^4(C_{0,1}) &= -\left(\frac{\sqrt{3}}{2}, \frac{1}{2}\right),\\
		\varphi_1^4(A_{0,2}) &= \left(0,0\right), & \varphi_1^4(B_{0,2}) &= -\left(\frac{\sqrt{3}}{2}, -\frac{1}{2}\right), & \varphi_1^4(C_{0,2}) &= -\left(\frac{\sqrt{3}}{2}, \frac{1}{2}\right),\\
		\varphi_1^4(A_{0,3}) &= \left(0,0\right), & \varphi_1^4(B_{0,3}) &= \left(\frac{\sqrt{3}}{2}, -\frac{1}{2}\right), & \varphi_1^4(C_{0,3}) &= \left(\frac{\sqrt{3}}{2}, \frac{1}{2}\right).
	\end{align*}
	This is in fact the linear combination of the two Fleck-Hutchinson modes $u_1^{\frac{N}{4},N}(x) - u_2^{\frac{N}{4},N}(x)$ that satisfies the consistency condition when $s = \frac{N}{4}$. To see why, from \eqref{eqn:u1} and \eqref{eqn:u2}, we get 
	\begin{align*}
		u_1^{\frac{N}{4},N}(A_{0,k}) - u_2^{\frac{N}{4},N}(A_{0,k}) &= (0,0), \\
		u_1^{\frac{N}{4},N}(B_{0,k}) - u_2^{\frac{N}{4},N}(B_{0,k}) &= \left(\cos(\frac{k\pi}{2}) - \sin(\frac{k\pi}{2})\right) \left(\frac{\sqrt{3}}{2}, -\frac{1}{2}\right), \\
		u_1^{\frac{N}{4},N}(C_{0,k}) - u_2^{\frac{N}{4},N}(C_{0,k}) &= \left(\cos(\frac{k\pi}{2}) - \sin(\frac{k\pi}{2})\right) \left(\frac{\sqrt{3}}{2}, \frac{1}{2}\right),
	\end{align*}
	for $k=0,1,\dots, N-1$. It is easy to see that $u_1^{\frac{N}{4},N}(x) - u_2^{\frac{N}{4},N}(x)$ is four-by-one periodic and $u_1^{\frac{N}{4},N}(x) - u_2^{\frac{N}{4},N}(x) = \varphi_1^4(x)$ because the factor $\cos(\frac{k\pi}{2}) - \sin(\frac{k\pi}{2})$ is $1$ when $\mod(k,4) = 0,3$; it is $-1$ when $\mod(k,4) = 1,2$. A similar calculation shows that the other linear combination $-u_1^{\frac{N}{4},N}(x) - u_2^{\frac{N}{4},N}(x)$ that satisfy the consistency condition corresponds the four-by-one periodic mechanism in Figure \ref{fig:four-by-one-mechanism}(b). The two four-by-one periodic mechanisms in Figure \ref{fig:four-by-one-mechanism}(a) and (b) are essentially the same, except that the mechanism in Figure \ref{fig:four-by-one-mechanism}(b) starts with an unshaded layer.
	\begin{figure}[!htb]
		\begin{minipage}[b]{.48\linewidth}
			\centering
			\subfloat[]{\includegraphics[width=0.65\linewidth]{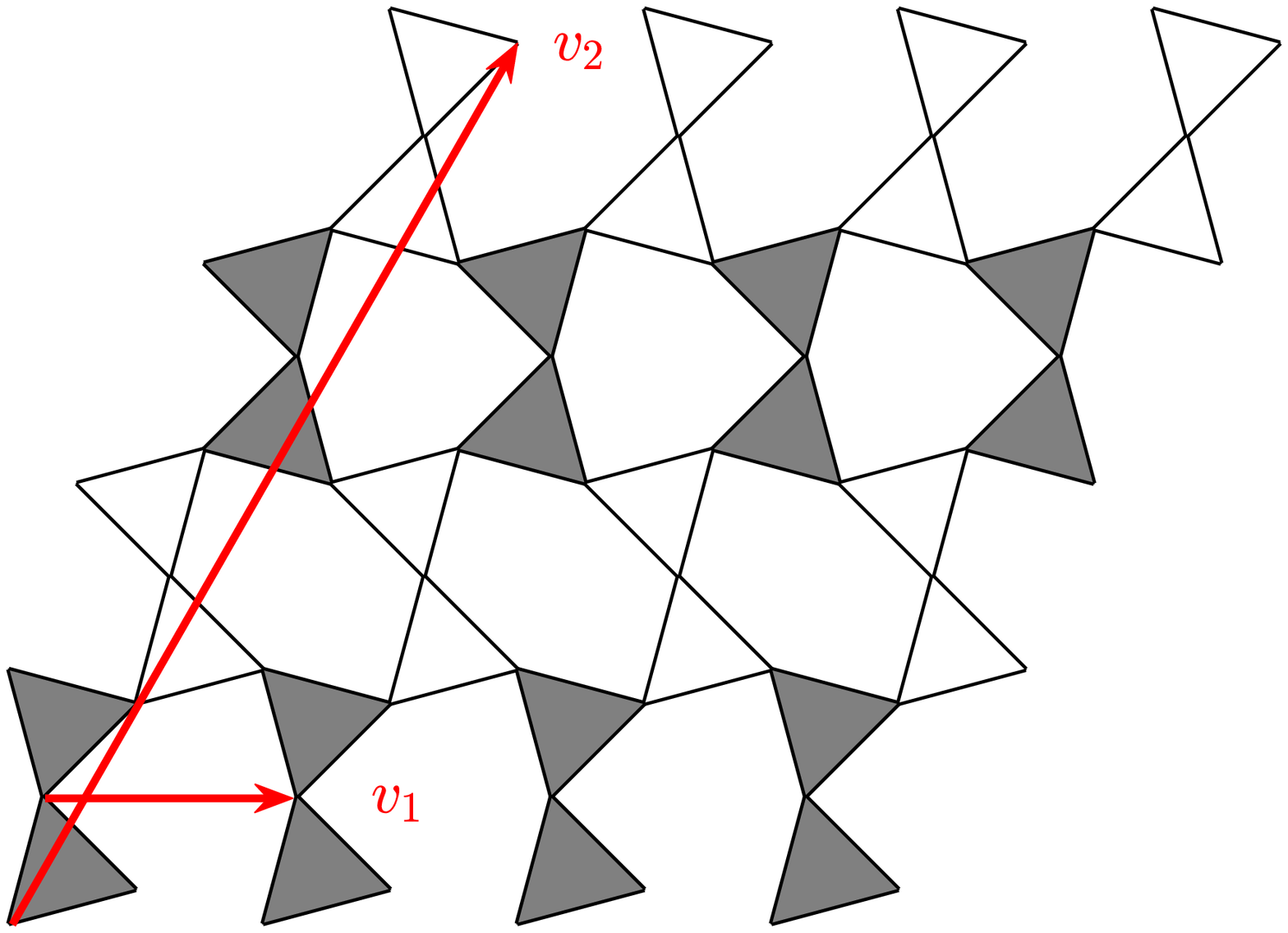}}
		\end{minipage}
		\begin{minipage}[b]{.48\linewidth}
			\centering
			\subfloat[]{\includegraphics[width=0.65\linewidth]{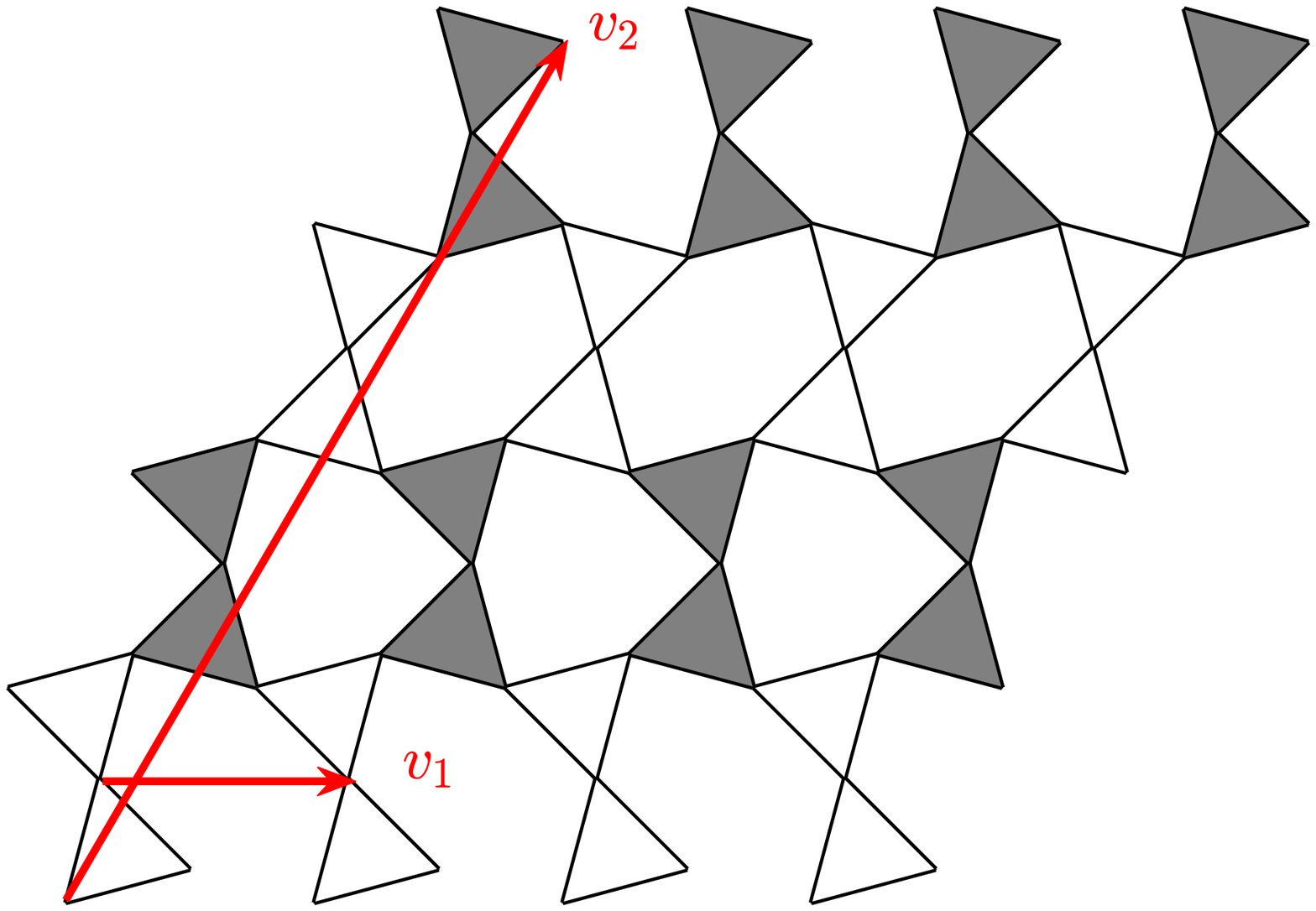}}
		\end{minipage} \par \medskip
		\begin{minipage}[b]{.48\linewidth}
			\centering
			\subfloat[]{\includegraphics[width=0.7\linewidth]{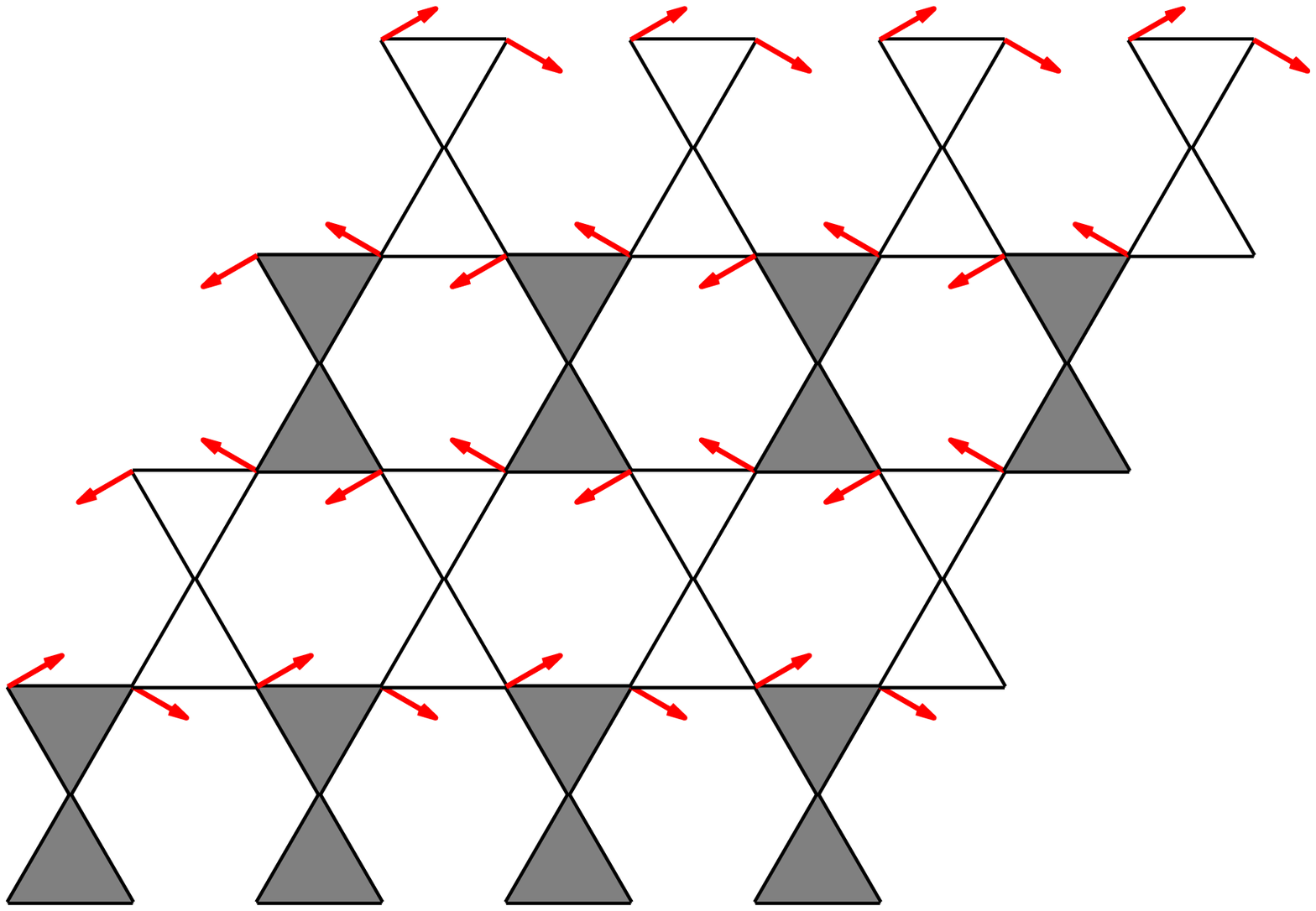}}
		\end{minipage}
		\begin{minipage}[b]{.48\linewidth}
			\centering
			\subfloat[]{\includegraphics[width=0.75\linewidth]{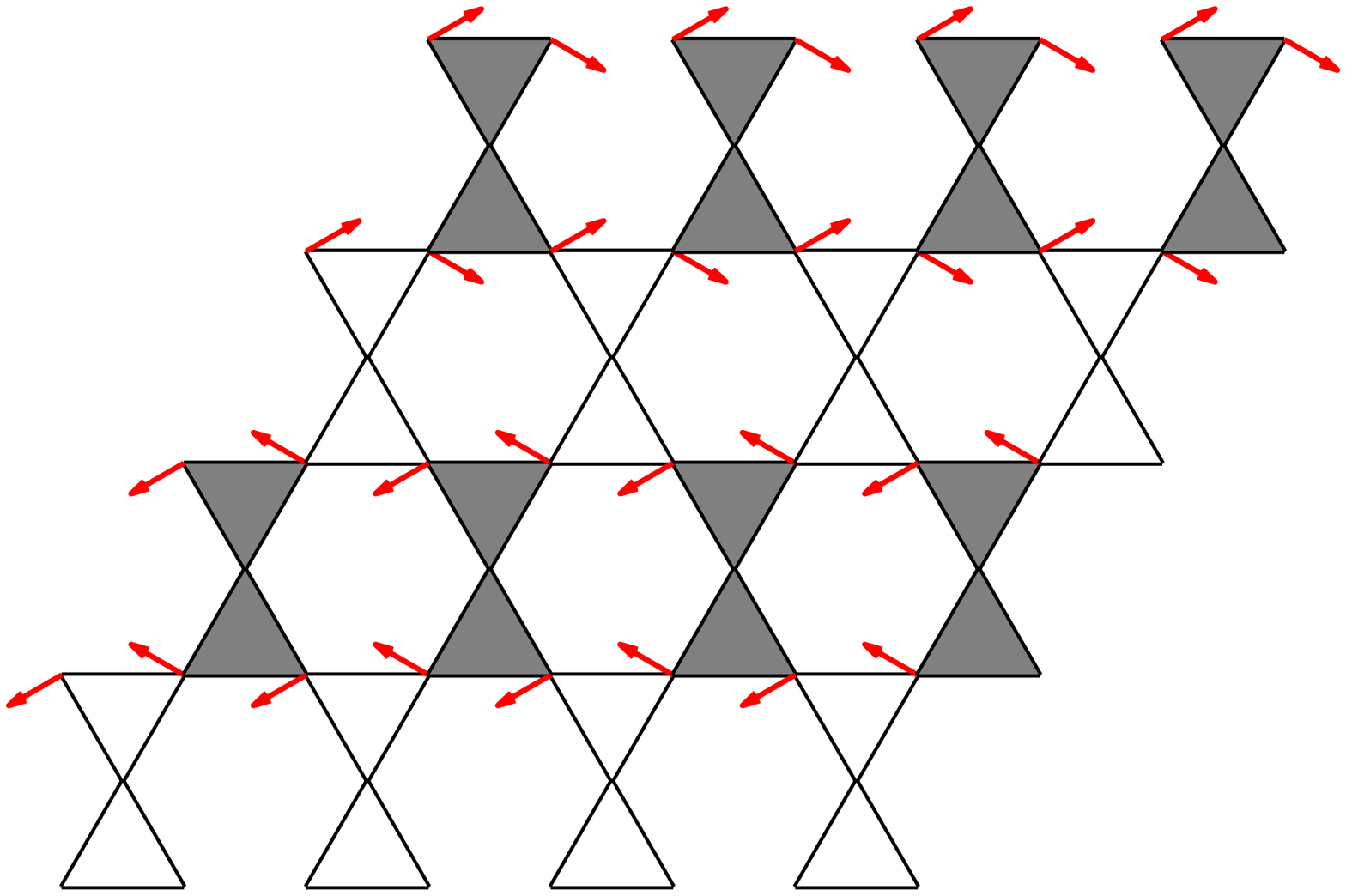}}
		\end{minipage}
		\caption{The four-by-one periodic GH modes that satisfy the consistency condition and their corresponding four-by-one periodic mechanisms: (a) the four-by-one periodic mechanism; (b) the same four-by-one periodic mechanism but starting with a different layer; (c) the four-by-one periodic GH mode as the infinitesimal version of (a); (d) the four-by-one periodic GH mode as the infinitesimal version of (b).}
		\label{fig:four-by-one-mechanism}
	\end{figure}

	\bibliography{ref}
	\bibliographystyle{plain}
\end{document}